\PassOptionsToPackage{hidelinks}{hyperref}

\documentclass{ws-m3as}

\usepackage{amsmath,amssymb}
\usepackage{graphicx}
\usepackage{algorithm}
\usepackage{algpseudocode}
\usepackage{hyperref}

\newtheorem{theorem}{Theorem}
\newtheorem{lemma}{Lemma}
\newtheorem{corollary}{Corollary}
\newtheorem{definition}{Definition}
\newtheorem{assumption}{Assumption}
\newenvironment{myproof}
{\par\noindent\textbf{Proof.}\enspace}
  {\unskip\nobreak\hfill\rule{0.5em}{0.5em}\par}
\usepackage[hidelinks]{hyperref}          
\begin{document}

\markboth{W. Hao and S. Wang}{Parameter Identifiability in data-driven modeling}

%
\catchline{}{}{}{}{}
%

\title{Parameter Identifiability beyond the Fisher Information Approximation in Data-Driven Modeling}

\author{Shun Wang}

\address{Department of Mathematics, Penn State University, University Park, Pennsylvania, United States of America}

\author{Wenrui Hao\footnote{Corresponding author}}

\address{Department of Mathematics, Penn State University, University Park, Pennsylvania, United States of America}

\maketitle


\begin{abstract}
Incorporating high-dimensional biological data into data-driven mechanistic modeling requires rigorous parameter identifiability to serve as an interpretable tool, avoid overfitting, and improve generalizability. Assessing parameter identifiability in high-dimensional systems remains challenging: profile likelihood entails costly repeated optimization, while criteria based solely on the Fisher information matrix discard the second-order Hessian information of the model and may therefore misjudge weakly constrained parameter directions, where this curvature can dominate the Fisher term.
Starting from a well-posedness definition of parameter identifiability,
we characterize the stability of the estimated parameter vector under
data perturbations while retaining the model's full Hessian information
beyond the Fisher approximation. This second-order information yields
H\"older-$\tfrac12$ stability under the zero-residual assumption when
the Fisher information is degenerate and, for a nonzero residual,
corrects the Fisher information matrix; the residual size and the
perturbation size together decide the order of stability, with
Lipschitz recovery below a crossover scale proportional to the residual
norm and H\"older-$\tfrac12$ recovery above it. Within a local quadratic
approximation of the least-squares loss, we derive coordinate
identifiability metrics allowing singular Jacobians and nuisance
Hessian blocks, and spectral perturbation bounds link these corrections to individual parameter identifiability and to uncertainty in parameter
estimates and model predictions.
Applications to three mechanistic models at distinct biological
scales---intracellular JAK--STAT signaling, within-host HIV dynamics,
and brain-wide amyloid-$\beta$ accumulation---using experimental and
clinical data support the coordinate identifiability criteria, the
Lipschitz-to-H\"older crossover, and the bounds on residual-induced
corrections, together with their implications for parameter uncertainty
and prediction intervals. By retaining Hessian information at the cost
of local derivative evaluations alone, the framework provides a
computationally tractable basis for extending identifiability analysis
to large-scale biological models and for designing experiments that
jointly improve parameter estimation and predictive reliability.
\end{abstract}


\tableofcontents

\section{Introduction}
\noindent
Incorporating high-dimensional biological data into data-driven mechanistic modeling has become a robust framework for elucidating complex biological processes and forecasting disease trajectories across multiple scales. At the subcellular and cellular levels, integrating single-cell transcriptomics with kinetic rate equations—particularly through the lens of RNA velocity \cite{la2018rna}—enables reconstruction of latent developmental trajectories by modeling transitions from unspliced to spliced mRNA \cite{qiu2022mapping,bocci2022splicejac,su2024hodge,weiler2024cellrank,yan2026decoding}. At the mesoscopic tissue scale, spatial proteomics and fluorescent labeling of the tumor microenvironment \cite{anderson2024global,gatenbee2022immunosuppressive,li2019infiltration,anderson2024global,miao2016binding} support the development of spatiotemporal interaction models that capture the biophysical constraints governing immune cell infiltration and tumor-immune dynamics \cite{liu2025multiscale,liao2022mathematical}. At the macroscopic clinical level, longitudinal biomarker datasets such as those from the Alzheimer’s Disease Neuroimaging Initiative (ADNI) facilitate cascade modeling \cite{petrella2024personalized,hao2022optimal}, mapping the stereotypical deposition patterns of Amyloid-$\beta$ and Tau proteins across the temporal progression of neurodegeneration \cite{rabiei2025data,zheng2022data}.

Despite these advancements, translating phenomenological observations into predictive clinical modeling remains hindered by challenges in parameter identifiability \cite{raue2009structural} and model uncertainty \cite{mitra2019parameter}. A fundamental prerequisite for the reliability of data-driven frameworks is determining which model parameters can be uniquely inferred from the limited, noisy, and sparse biological data \cite{miao2011identifiability}. Moreover, rigorously quantifying uncertainty distributions is essential for assessing the confidence and reliability of model predictions \cite{portela2025conformal}. Addressing these modeling bottlenecks is therefore critical for the robust translation of data-driven models into actionable precision medicine tools and clinical decision-support systems \cite{rockne2026future}.

Parameter (practical) identifiability determines whether model parameters can be uniquely resolved from specific experimental observations. The profile likelihood method has long served as the gold standard for this assessment \cite{raue2009structural,liu2024parameter,bergman2025efficient,gallo2022lack,ciocanel2024parameter,eisenberg2017confidence,simpson2026parameter}. Nevertheless, applying traditional approaches—particularly the profile likelihood method—to high-dimensional biological systems is often infeasible due to prohibitive computational costs. Recently, a systematic computational framework has been proposed that formalizes practical identifiability on a rigorous mathematical foundation, introduces a high-efficiency metric to alleviate computational burdens, and bridges the gap between the theoretical definition of identifiability and its practical assessment \cite{wang2025systematic}.

An important question remains, however: how much of the local curvature
information is used when practical identifiability and uncertainty are
assessed. Most local criteria rely on first-order sensitivities only,
replacing the Hessian of the loss by the Fisher information matrix built
from the Jacobian of the model output and discarding the second-order
derivatives of the model. Mechanistic models are typically fitted to
sparse and noisy data, and weakly constrained parameter directions
arise in diverse systems, including morphogen gradients
\cite{simpson2026parameter,rodriguez2022concentration}, virus--host
dynamics \cite{nguyen2015ebola,sharomi2010mathematical}, and
taxol-treated cancer models
\cite{borisov2020confidence,mirzaei2023investigating}. Along such
directions the Fisher information is small and the neglected Hessian
term can be comparable to or larger than it, so that a parameter may be
classified as non-identifiable although second-order curvature
constrains it, or its uncertainty may be misjudged by more than an
order of magnitude, as the HIV example below demonstrates.

Profile likelihood or Bayesian approach does not rely on this approximation and
can assess whether confidence intervals are bounded at a prescribed
statistical threshold \cite{kreutz2018easy,hines2014determination}, but it requires repeated
re-optimization and does not answer the complementary question of how
strongly each estimated parameter responds to perturbations of the
observations. Addressing this question requires retaining the Hessian
information of the model in the local analysis, quantifying its effect
on the stability and uncertainty of individual parameter estimates, and
identifying when the Fisher approximation is adequate and when it is
not. The matrix that retains this curvature is not new. Under Gaussian noise, $\sigma^{-2}K_\eta$ is the observed information and $\sigma^{-2}F_\eta$ its expectation, the Fisher information; conditioning on the former goes back to Efron and Hinkley \cite{EfronHinkley1978}, and the Laplace approximation built on it is the standard local Gaussian posterior approximation \cite{TierneyKadane1986,MacKay1992}. In nonlinear least squares the same matrices are the full and the Gauss--Newton Hessian \cite{BatesWatts1988}, a recurring choice in large-scale Bayesian inverse problems \cite{MartinWilcoxBursteddeGhattas2012,BuiThanhGhattasMartinStadler2013}. What is missing is a theory of what the discarded term does to parameter identifiability: whether, and at what rate, the estimate is stable under data perturbations when $F_\eta$ is degenerate, and how this depends on the residual. That question is the subject of this paper.

In this paper, building on the well-posedness framework of
\cite{wang2025systematic}, we retain the full second-order (Hessian)
information of the model output rather than its Fisher (Gauss--Newton)
approximation and derive coordinate identifiability metrics from a local
quadratic approximation of the squared $\ell_2$ loss. Theoretically, the
Hessian information yields H\"older-$\tfrac12$ stability of the parameter
estimate under the zero-residual assumption when the Fisher information
matrix is degenerate, whereas first-order information guarantees Lipschitz
stability only in the full-rank case; between these limits, the residual
size and the perturbation size together decide the order of stability,
with a crossover scale proportional to the residual norm. Practically, the
Hessian information corrects the Fisher information matrix, and we
characterize how this correction changes its eigenvalues and eigenvectors
and bound the resulting changes in individual identifiability metrics.
Applications to JAK--STAT signaling, HIV host--virus dynamics, and
amyloid-$\beta$ (A$\beta$) accumulation show agreement with
profile-likelihood classifications and the predicted Lipschitz-to-H\"older
crossover (JAK--STAT), improved agreement with MCMC uncertainty estimates
when residual corrections are large (HIV), and a shift in the timing of
minimum predictive uncertainty (A$\beta$), clarifying when Fisher-based
approximations are adequate and when residual-dependent curvature must be
retained.

\section{Methods}
\subsection{Overview of parameter identifiability analysis}

Our framework links model fitting, practical identifiability, and uncertainty quantification through the geometry of the least-squares loss (Fig.~\ref{fig:1}). A model $\boldsymbol\varphi(t;\boldsymbol\theta)$, specified explicitly or through differential equations, is observed through $\boldsymbol h(\boldsymbol\varphi(t;\boldsymbol\theta))$ and compared with measurements $\{t_n,\hat{\boldsymbol h}_n\}_{n=1}^{N}$. For positive parameters we use logarithmic coordinates $\boldsymbol\eta=\ln\boldsymbol\theta$ and obtain the fitted solution $\boldsymbol\eta^*$ by minimizing $l(\boldsymbol\eta;\hat{\boldsymbol h})=\tfrac12\sum_{n=1}^{N}\|\boldsymbol h(t_n,\boldsymbol\eta)-\hat{\boldsymbol h}_n\|_2^2$, with residual $\boldsymbol r^*=\boldsymbol h(t_n,\boldsymbol\eta^*)-\hat{\boldsymbol h}_n$. At $\boldsymbol\eta^*$ we compute the Fisher information matrix $F_\eta=J_\eta^{\top}J_\eta$ from the Jacobian $J_\eta$ of the model output and the Hessian tensor $\boldsymbol H_\eta$ of the stacked observations. Together they form the full loss Hessian $K_\eta=F_\eta+\boldsymbol r^*\otimes\boldsymbol H_\eta$, whereas the Fisher approximation $K_\eta\approx F_\eta$ used in \cite{wang2025systematic} discards the second-order term. An eigen-decomposition of $F_\eta$, $Q=[U\ V]$, separates the directions informed by first-order sensitivities from its null space.

\begin{figure}[ht]
\centering
\includegraphics[width=1.0\linewidth]{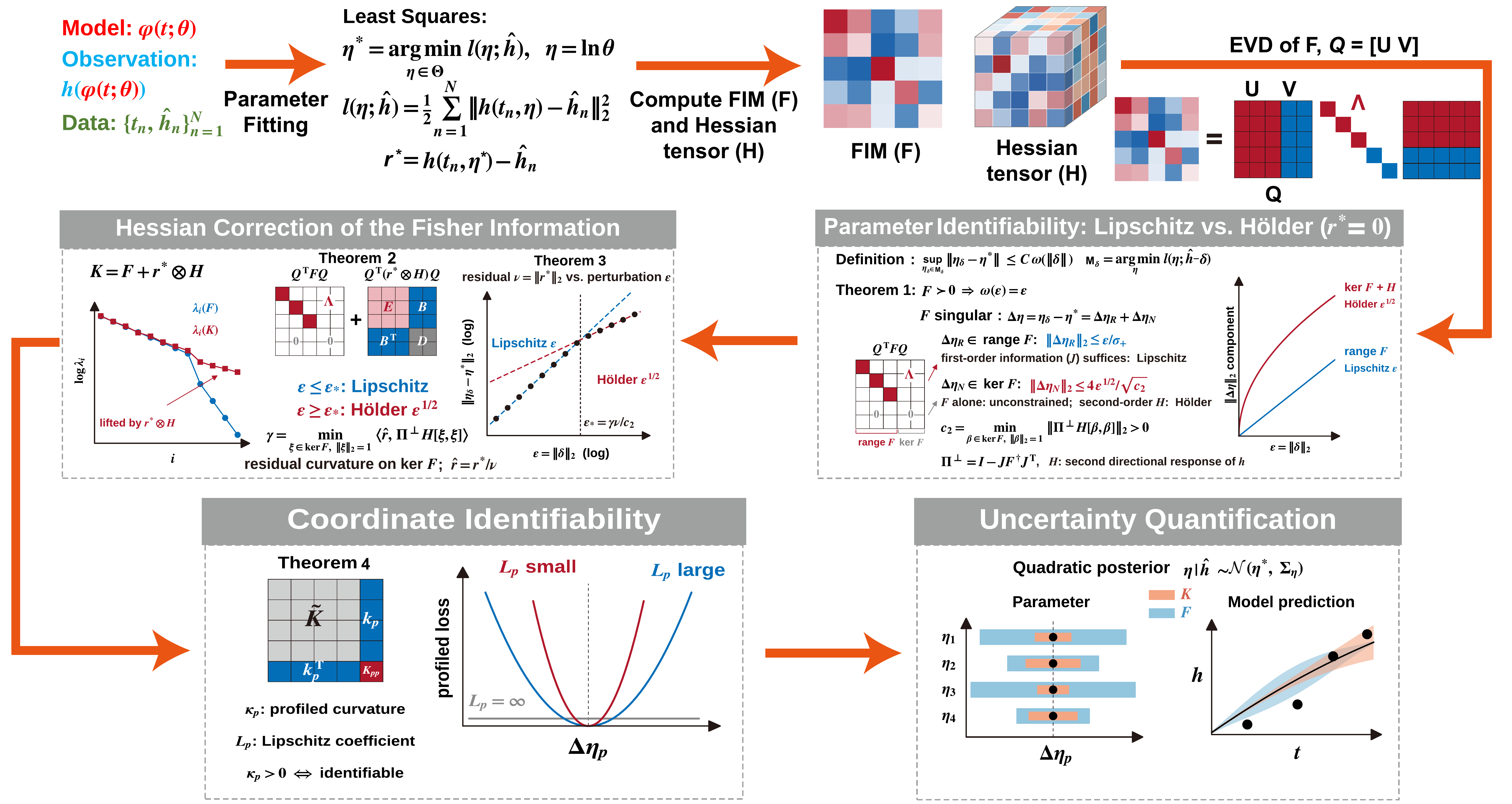}
\caption{\scriptsize\textbf{Overview of parameter identifiability analysis beyond the Fisher information approximation.}
Arrows indicate the order of the analysis.
\emph{Top}: a model $\boldsymbol\varphi(t;\boldsymbol\theta)$, observed through $\boldsymbol h(\boldsymbol\varphi(t;\boldsymbol\theta))$, is fitted to data by least squares in logarithmic coordinates $\boldsymbol\eta=\ln\boldsymbol\theta$, leaving the residual $\boldsymbol r^*$. At $\boldsymbol\eta^*$, the Fisher information $F$ and the Hessian tensor $H$ give the loss Hessian $K=F+\boldsymbol r^*\otimes H$; the eigendecomposition of $F$ yields $Q=[U\ V]$, with $U$ (red) spanning its range and $V$ (blue) its null space.
\emph{Parameter identifiability} (right; $\boldsymbol r^*=\boldsymbol0$): stability of all perturbed minimizers, $\sup\|\boldsymbol\eta_\delta-\boldsymbol\eta^*\|_2\le C\,\omega(\|\boldsymbol\delta\|_2)$. If $F\succ0$, recovery is Lipschitz; if $F$ is singular, the range component stays Lipschitz, whereas the null-space component is controlled only by $H$, at the H\"older rate $\varepsilon^{1/2}$, when $c_2>0$.
\emph{Hessian correction} (left): for a nonzero residual, $\boldsymbol r^*\otimes H$ lifts the weak eigenvalues of $F$; the residual size $\nu=\|\boldsymbol r^*\|_2$ and the perturbation size $\varepsilon$ decide the rate, Lipschitz below $\varepsilon_*=\gamma\nu/c_2$ and H\"older above (Theorem~3).
\emph{Coordinate identifiability}: $\kappa_p>0$ decides identifiability of $\eta_p$; $L_p=\infty$ marks a flat direction.
\emph{Uncertainty quantification}: Laplace-posterior intervals from $K$ (orange) and $F$ (blue).}
\label{fig:1}
\end{figure}

Parameter identifiability asks whether small perturbations $\boldsymbol\delta$ produce controlled changes in all minimizers of the perturbed loss, $\sup_{\mathcal{M}_{\boldsymbol{\delta}}
}\|\boldsymbol\eta_\delta-\boldsymbol\eta^*\|\le C\,\omega(\|\boldsymbol\delta\|)$ (Definition~\ref{def:lip}). With a zero residual, $K_\eta=F_\eta$ and only the second directional response $\boldsymbol H_\eta$ carries information beyond $F_\eta$: if $F_\eta$ is singular, the displacement splits into a range component, controlled at the Lipschitz rate by $J_\eta$, and a null-space component that $F_\eta$ leaves free but that $\boldsymbol H_\eta$, through its coercivity $c_2>0$ on $\ker F_\eta$, controls at the H\"older rate $\varepsilon^{1/2}$ (Theorem~\ref{prop:holder}); for full-rank $F_\eta$ the Lipschitz rate is recovered (Corollary~\ref{cor:fullrank}). With a nonzero residual, $\boldsymbol r^*\otimes\boldsymbol H_\eta$ corrects $F_\eta$ and $K_\eta\succ0$ yields Lipschitz stability (Theorem~\ref{thm:local_lip}). Between these limits the residual size $\nu=\|\boldsymbol r^*\|_2$ and the perturbation size decide the rate: Lipschitz below the crossover $\varepsilon_*=\gamma\nu/c_2$, H\"older above (Theorem~\ref{thm:crossover}).

We next consider an individual coordinate $\eta_p$, allowing the remaining parameters to adjust conditionally. Within the local quadratic approximation of the loss, eliminating the nuisance block of $K_\eta$ by a generalized Schur complement yields the profiled curvature $\kappa_p^{K}$ and the Lipschitz coefficient $L_p^{K}$, which quantify the local constraint on the coordinate and its sensitivity to data perturbations; the coordinate is identifiable if and only if $\kappa_p^{K}>0$ (Theorem~\ref{thm:profile-quadratic-lipschitz}), and global bounds relate $L_p^{K}$ to its Fisher counterpart (Proposition~\ref{thm:rh-global-coordinate-bounds}). Finally, the Laplace approximation of the parameter posterior, $\boldsymbol\eta\,|\,\hat{\boldsymbol h}\approx\mathcal N(\boldsymbol\eta^*,\Sigma_\eta)$, with $\Sigma_\eta=\sigma^2K_\eta^{-1}$ when $K_\eta$ is positive definite and an explicitly regularized null space otherwise, yields credible intervals for each coordinate and, by propagation through the model, prediction intervals. Comparing these intervals with those from $F_\eta$ shows how the Hessian information changes parameter uncertainty and its evolution in model predictions.

\subsection{Parameter identifiability of full parameter vector}
\subsubsection{Parameter identifiability definitions}
For the data-driven modeling approach, let
$\boldsymbol\Theta\subseteq(0,+\infty)^{k}$ be an open parameter set. The system
output $\boldsymbol\varphi(t,\boldsymbol\theta)\in\mathbb R^{M}$ is observed
through a map $\boldsymbol h:\mathbb R^{M}\to\mathbb R^{L}$ with $L$ observation
channels, and we write
\begin{equation}
\label{eq:model-map}
  \boldsymbol h(t_n,\boldsymbol\theta)
  :=\boldsymbol h\big(\boldsymbol\varphi(t_n,\boldsymbol\theta)\big)
  \in\mathbb R^{L},
  \quad n=1,\dots,N,
\end{equation}
where $N$ is the number of experimental time points. Let $\boldsymbol\varphi: \mathcal T\times\boldsymbol\Theta \longrightarrow \mathbb R^M$denote a general time-dependent parameter-to-state map, and let $\mathcal H:\mathbb R^M\longrightarrow\mathbb R^L$ denote the observation operator. The model prediction at time $t$ is
then $\boldsymbol h(t,\boldsymbol\theta)
:=\mathcal H\bigl(\boldsymbol\varphi(t,\boldsymbol\theta)\bigr)$. No particular evolution equation is imposed on $\boldsymbol\varphi$. We assume throughout that $\mathcal H\in C^3(\mathbb R^M;\mathbb R^L)
~\text{and}~
\boldsymbol\varphi(t_n,\cdot)
\in C^3(\boldsymbol\Theta;\mathbb R^M),
~n=1,\ldots,N$. Consequently,
$\boldsymbol h(t_n,\cdot)\in
C^3(\boldsymbol\Theta;\mathbb R^L)$ at every experimental time point. This setting includes, but is not restricted to, solution maps generated by ordinary or partial differential equations, discrete-time dynamical systems, analytic time-series models, and differentiable numerical simulators.

Given observations $\{(t_n,\hat{\boldsymbol h}_n)\}_{n=1}^{N}$ with
$\hat{\boldsymbol h}_n\in\mathbb R^{L}$, stack data, model output and residual
into single vectors indexed by the pair $(n,l)$, ordered lexicographically $\hat{\boldsymbol h}:=\big[\hat{\boldsymbol h}_1;\dots;\hat{\boldsymbol h}_N\big]\in\mathbb R^{NL}$
so that the least-squares loss reads
\begin{equation}
\label{eq:1}
  l(\boldsymbol\theta;\hat{\boldsymbol h})
  =\frac12\sum_{n=1}^{N}
    \big\|\boldsymbol h(t_n,\boldsymbol\theta)-\hat{\boldsymbol h}_n\big\|_2^{2},
\end{equation}
and the estimated parameter is
$ \boldsymbol\theta^{*}=\arg\min_{\boldsymbol\theta\in\boldsymbol\Theta}
    l(\boldsymbol\theta;\hat{\boldsymbol h}),
~\boldsymbol\theta^{*}\in\operatorname{int}\boldsymbol\Theta $.
We express all positive model parameters $\boldsymbol\theta$ in logarithmic coordinates and define $ \boldsymbol\eta=\ln\boldsymbol\theta$,
\begin{equation}
\boldsymbol\eta-\boldsymbol\eta^*
=\ln\left(\frac{\boldsymbol\theta}{\boldsymbol\theta^*}\right)
\approx \frac{\boldsymbol\theta-\boldsymbol\theta^*}{\boldsymbol\theta^*}
\end{equation}
near $\boldsymbol\theta^*$. Distances in the $\boldsymbol\eta$ measure relative parameter errors. This parameterization is particularly suitable for biological models, in which parameters often span several orders of magnitude. In what follows, $\boldsymbol\eta$ always denotes the parameter vector in logarithmic coordinates.

\paragraph{Profile likelihood definition.}
From the statistical perspective, parameter practical identifiability is commonly formulated through
likelihood-based confidence regions
\cite{raue2009structural,raue2013joining,simpson2026parameter}. This notion asks whether,
at the noise level and confidence level of the present experiment, the data restrict a parameter to a finite range. It is therefore a finite-resolution and generally global
notion of identifiability. Suppose that the measurement errors are independent Gaussian random variables with a common
known variance $\sigma^2$. With the least-squares convention of Eq.~\eqref{eq:1},
\begin{equation}
l(\boldsymbol\eta;\hat{\boldsymbol h})
:=
l\big(
\boldsymbol h(\boldsymbol t,\boldsymbol\eta),
\hat{\boldsymbol h}
\big)
=
\frac{1}{2}
\sum_{n=1}^{N}
\big\|
\boldsymbol h(t_n,\boldsymbol\eta)-\hat{\boldsymbol h}_n
\big\|_2^2,
\end{equation}
the negative log-likelihood differs from
$\sigma^{-2}l(\boldsymbol\eta;\hat{\boldsymbol h})$ only by an additive constant.
Consequently, the likelihood-ratio statistic relative to the maximum-likelihood estimate
$\boldsymbol\theta^*$ is
\begin{equation}
\label{eq:lrstat}
\Lambda(\boldsymbol\eta)
:=
\frac{2}{\sigma^2}
\left[
l(\boldsymbol\eta;\hat{\boldsymbol h})
-
l(\boldsymbol\eta^*;\hat{\boldsymbol h})
\right].
\end{equation}
Equivalently, if the loss is expressed directly in units of the noise variance,
$\Lambda$ plays the usual role of the increase in the $\chi^2$ statistic.

The interpretation is frequentist: the parameter $\boldsymbol\eta$ is fixed but unknown,
whereas the measurements $\hat{\boldsymbol h}$ are random under repeated sampling.
For a nominal confidence level $1-\alpha$, a likelihood-based confidence region consists
of parameter values for which the likelihood-ratio statistic does not exceed the
corresponding threshold $\Delta_{1-\alpha}$. Under the usual asymptotic likelihood-ratio
approximation, $\Delta_{1-\alpha}$ is taken from an appropriate $\chi^2$ distribution.

For the $p$-th parameter coordinate, the profile likelihood is obtained by fixing
$\eta_p$ and re-optimizing all remaining coordinates over the original parameter space :
\begin{equation}
\label{eq:pl}
l_{\mathrm{PL}}(\eta_p)
:=
\min_{\substack{ \boldsymbol e_p^\top\boldsymbol\eta=\eta_p}}
l(\boldsymbol\eta;\hat{\boldsymbol h}).
\end{equation}
The corresponding profile likelihood-ratio statistic and pointwise likelihood-based confidence interval are
\[\Lambda_{\mathrm{PL}}(\eta_p)
:=
\frac{2}{\sigma^2}
\left[
l_{\mathrm{PL}}(\eta_p)
-
l(\boldsymbol\eta^*;\hat{\boldsymbol h})
\right],~\mathrm{CI}_p(1-\alpha)
:=
\left\{
\eta_p:
\Lambda_{\mathrm{PL}}(\eta_p)
\le
\Delta_{1-\alpha}
\right\}.\]
For a one-dimensional pointwise interval, the usual asymptotic choice is
$\Delta_{1-\alpha}=\chi^2_{1,\,1-\alpha}$.

\begin{definition}[Profile likelihood-based parameter identifiability]
\label{def:like}
For a prescribed confidence level $1-\alpha$, the parameter coordinate $\eta_p$ is said
to be practically identifiable if its profile likelihood confidence
interval $\mathrm{CI}_p(1-\alpha)$ , computed over the original
parameter space, is bounded in both directions. Equivalently, there
exist finite values
$
\eta_p^-<\eta_p^*<\eta_p^+
$
such that
$
\Lambda_{\mathrm{PL}}(\eta_p^-)
>
\Delta_{1-\alpha},
~
\Lambda_{\mathrm{PL}}(\eta_p^+)
>
\Delta_{1-\alpha}.
$
If the confidence interval is unbounded in at least one direction, then $\eta_p$ is said
to be practically non-identifiable at that confidence level.
\end{definition}

\paragraph{Bayesian posterior definition}
Let $\pi(\boldsymbol\eta)$ be a prior density on the domain
$\Theta\subseteq\mathbb R^k$ and let
\begin{equation}
\label{eq:bayes-posterior}
\pi(\boldsymbol\eta\mid\widehat{\boldsymbol h})
=\frac{p(\widehat{\boldsymbol h}\mid\boldsymbol\eta)\,\pi(\boldsymbol\eta)}
      {\int_\Theta p(\widehat{\boldsymbol h}\mid\boldsymbol\zeta)\,\pi(\boldsymbol\zeta)\,d\boldsymbol\zeta},
~
p(\widehat{\boldsymbol h}\mid\boldsymbol\eta)
\propto\exp\!\Bigl(-\tfrac{1}{\sigma^{2}}\,l(\boldsymbol\eta;\widehat{\boldsymbol h})\Bigr),
\end{equation}
be the posterior under the Gaussian observation model, so that the
least-squares loss $l$ of Eq.~\eqref{eq:1} is the negative log-likelihood
up to the factor $\sigma^{-2}$.
For a partition $\boldsymbol\eta=(\boldsymbol\eta_1,\boldsymbol\eta_2)$ we
write $\pi(\boldsymbol\eta_2\mid\boldsymbol\eta_1)$ and
$\pi(\boldsymbol\eta_2\mid\boldsymbol\eta_1,\widehat{\boldsymbol h})$ for the
prior and posterior conditionals, and
$\pi(\eta_p\mid\widehat{\boldsymbol h})
=\int\pi(\boldsymbol\eta\mid\widehat{\boldsymbol h})\,d\boldsymbol\eta_{-p}$
for the marginal posterior of the $p$-th coordinate.

\begin{definition}[Bayesian parameter identifiability]
\label{def:bayes-practical}
Let $\pi(\boldsymbol\eta)$ be proper, and fix a credible level
$1-\alpha\in(0,1)$ and a precision tolerance $\tau_p>0$ for the $p$-th
coordinate. Denote by $Q_{p}^{\mathrm{prior}}(\alpha)$ and
$Q_{p}(\alpha;\widehat{\boldsymbol h})$ the shortest $(1-\alpha)$ credible
intervals of the prior and posterior marginals of $\eta_p$, and by
\begin{equation}
\label{eq:profile-posterior}
\mathrm{PP}(\eta_p\mid\widehat{\boldsymbol h})
:=\max_{\boldsymbol\eta_{-p}}\pi\bigl((\boldsymbol\eta_{-p},\eta_p)\mid\widehat{\boldsymbol h}\bigr)
\end{equation}
the profile posterior \cite{raue2013joining}.
The coordinate $\eta_p$ is \emph{practically identifiable in the Bayesian
sense} at $\widehat{\boldsymbol h}$ if
\begin{enumerate}
\item[(i)] the posterior is concentrated relative to the prior,
\begin{equation}
\label{eq:bayes-shrinkage}
\bigl|Q_{p}(\alpha;\widehat{\boldsymbol h})\bigr|
\le\tau_p
<\bigl|Q_{p}^{\mathrm{prior}}(\alpha)\bigr|,
\end{equation}
where $|\cdot|$ denotes interval length; and
\item[(ii)] the profile posterior decays to zero on both sides of its
maximizer, i.e.\ the level set
$\{\eta_p:\ \mathrm{PP}(\eta_p\mid\widehat{\boldsymbol h})
\ge e^{-\Delta_\alpha/2}\,\pi(\boldsymbol\eta^{\mathrm{MAP}}\mid\widehat{\boldsymbol h})\}$ is a bounded interval strictly inside the prior support for the
$\chi^2_1$ threshold $\Delta_\alpha$, where $\boldsymbol\eta^{\mathrm{MAP}}$ is posterior mode.
\end{enumerate}
Otherwise $\eta_p$ is practically non-identifiable in the Bayesian
sense: the marginal posterior is unbounded, multimodal, or spans the
prior support, or the profile posterior tails off to a plateau on at least
one side \cite{raue2013joining,hines2014determination,simpson2020practical}.
\end{definition}

\paragraph{Well-posedness definition.}
From the perspective of inverse problems, parameter identifiability is formulated as a well-posedness property of the data-to-estimate relation
\cite{maclaren2019can,krivorotko2021sensitivity,ahmad2026data}: sufficiently small perturbations
of the measurements should induce only controllably small displacements of all
parameter estimates. This formulation is local in parameter space and does not, by itself,
exclude distant parameter values or disconnected branches that fit the data equally well.
The local Lipschitz condition used in our previous work is a special case of this stability
viewpoint \cite{wang2025systematic}.

To formalize this notion, we introduce four objects in the order in which they are used:
the perturbation, the local inverse modulus, and the index functions
used to quantify its rate of decay.

\medskip\noindent\emph{(i) The perturbation.}
A perturbed measurement is written as
$\hat{\boldsymbol h}-\boldsymbol\delta$, where
$\boldsymbol\delta=[\boldsymbol\delta_1,\dots,\boldsymbol\delta_N]^\top$
is blocked conformally with the data, $\boldsymbol\delta_n\in\mathbb R^{L}$, and
$\|\boldsymbol\delta\|_2
=
\Big(
\sum_{n=1}^{N}\|\boldsymbol\delta_n\|_2^{2}
\Big)^{1/2}$.
The corresponding perturbed objective is
\begin{equation}
\label{eq:lossdelta}
l\big(
\boldsymbol h(\boldsymbol t,\boldsymbol\eta),
\hat{\boldsymbol h}-\boldsymbol\delta
\big)
=
\frac{1}{2}
\sum_{n=1}^{N}
\big\|
\boldsymbol h(t_n,\boldsymbol\eta)
-\hat{\boldsymbol h}_n
+\boldsymbol\delta_n
\big\|_{2}^{2},
\end{equation}
which is continuous in both $\boldsymbol\eta$ and $\boldsymbol\delta$ under the regularity
assumptions below.

\medskip\noindent\emph{(ii) Local smoothness and uniqueness.}
Because the present notion concerns local stability rather than global parameter ambiguity,
we restrict the inverse problem to a fixed compact neighborhood of the reference estimate
$\boldsymbol{\eta^*}$. This restriction defines the Local smoothness and uniqueness on which conditional
stability is studied.

\begin{assumption}[Local smoothness and uniqueness]
\label{ass:admissible}
There exists a closed ball
$U_{\eta} =\overline{B}(\boldsymbol\eta^*,R)$
with radius $R>0$ fixed independently of $\boldsymbol\delta$, such that
$\boldsymbol\eta\mapsto
\boldsymbol h(\boldsymbol t,\boldsymbol\eta)$
is twice continuously differentiable on a neighborhood of $U_{\eta}$.
Here, $ \boldsymbol\eta^*
= \ln\boldsymbol\theta^*$
is the logarithmic representation of the reference parameter
$\boldsymbol\theta^*$ .
The point $\boldsymbol\eta^*$ lies in the interior of $U_{\eta}$ and is
the unique minimum of the unperturbed loss within $U_{\eta}$.
\end{assumption}

The uniqueness assumption in {\bf Assumption}~\ref{ass:admissible} is imposed only at the
unperturbed data. For perturbed data, multiple localized minimums are allowed; stability
will require all such minimums to remain close to $\boldsymbol{\eta^*}$.

\medskip\noindent\emph{(iii) The local inverse modulus.}
For $\boldsymbol\eta\in U_\eta$, define the output displacement from
$\boldsymbol{\eta^*}$ by
$d(\boldsymbol\eta)
:=
\Bigg(
\sum_{n=1}^{N}
\big\|
\boldsymbol h(t_n,\boldsymbol\eta)
-
\boldsymbol h(t_n,\boldsymbol{\eta^*})
\big\|_{2}^{2}
\Bigg)^{1/2}$.
The corresponding local modulus of continuity of the inverse relation, anchored at
$\boldsymbol{\eta^*}$, is
\begin{equation}
\label{eq:modulus}
\omega_U(\varepsilon)
:=
\sup
\Big\{
\|\boldsymbol\eta-\boldsymbol{\eta^*}\|_2:
\boldsymbol\eta\in U_\eta,\;
d(\boldsymbol\eta)\le\varepsilon
\Big\}.
\end{equation}
Thus, $\omega_U(\varepsilon)$ measures the largest parameter displacement within $U_\eta$
that remains indistinguishable from $\boldsymbol{\eta^*}$ at output resolution
$\varepsilon$. Local continuity of the inverse relation at $\boldsymbol{\eta^*}$ is
equivalent to $\omega_U(\varepsilon)\rightarrow 0
~\text{as}~
\varepsilon\to0 $.

\medskip\noindent\emph{(iv) Index functions.}
A function
$\omega:[0,\infty)\to[0,\infty)$ is called an index function if it is continuous, strictly increasing, and satisfies $\omega(0)=0$. Since the modulus $\omega_U$ need not itself be continuous or strictly increasing, an index function is used to describe the rate at which the local inverse uncertainty vanishes: for sufficiently small $\varepsilon$, $\omega_U(\varepsilon)
\le C\,\omega(\varepsilon),$
with a constant $C>0$ independent of the perturbation.

\begin{definition}[Parameter identifiability as local conditional stability]
\label{def:lip}
Under {\bf Assumption}~\ref{ass:admissible}, let $U_\eta$ be a closed
neighbourhood of $\boldsymbol\eta^*$ and, for a perturbation
$\boldsymbol\delta\in\mathbb R^{NL}$, let
\begin{equation}
\label{eq:3}
\mathcal M_{\delta}
:=\arg\min_{\boldsymbol\eta\in U_\eta}
l\big(\boldsymbol h(\boldsymbol t,\boldsymbol\eta),\hat{\boldsymbol h}-\boldsymbol\delta\big)
\end{equation}
denote the set of local minimizers of the perturbed loss in $U_\eta$.
The parameter $\boldsymbol\eta$ is said to be \emph{practically identifiable}
at the data $\hat{\boldsymbol h}$ with stability index $\omega$ if there exist
constants $\varepsilon_0>0$ and $C>0$ and an index function $\omega$ such that,
for every $\boldsymbol\delta$ with $\|\boldsymbol\delta\|_2<\varepsilon_0$,
$\mathcal M_\delta\neq\emptyset$ and every
$\boldsymbol\eta_\delta\in\mathcal M_\delta$ satisfies the conditional
stability estimate
\begin{equation}
\label{eq:condstab-general}
\|\boldsymbol\eta_\delta-\boldsymbol\eta^*\|
\le C\,\omega\big(\|\boldsymbol\delta\|\big),
\quad\text{i.e.}\quad
\sup_{\boldsymbol\eta_\delta\in\mathcal M_\delta}
\|\boldsymbol\eta_\delta-\boldsymbol\eta^*\|
\le C\,\omega\big(\|\boldsymbol\delta\|\big).
\end{equation}
\end{definition}

\noindent
Practical identifiability in the Definition \ref{def:lip} requires that all localized estimates compatible with a sufficiently small perturbation remain within a shrinking neighborhood of $\boldsymbol{\eta^*}$. The case $\omega\big(\|\boldsymbol\delta\|_2\big)=\|\boldsymbol\delta\|_2$ gives
the Lipschitz estimate
$\|\boldsymbol{\eta_\delta}-\boldsymbol{\eta^*}\|_2
\le
C\,\|\boldsymbol\delta\|_2$, recovering the stability notion used in our previous work \cite{wang2025systematic}. 

\paragraph{Differences in definitions of parameter identifiability.}
Parameter identifiability concerns whether, and how reliably, a parameter can be recovered from finite and noisy observations~\cite{simpson2026parameter}. The three definitions quantify this reliability by examining the observation--parameter relationship in different directions.

Profile likelihood fixes the observed data $\hat{\boldsymbol h}$ and evaluates the loss along a selected parameter $\eta_p$, reoptimizing the remaining parameters at each value; the confidence interval $\mathrm{CI}_p$ collects values whose profiled loss increase stays below the cutoff $\Delta_{1-\alpha}$. It thus characterizes a fixed-data likelihood sublevel set rather than the stability of the recovered parameter under data perturbations, and its confidence interpretation relies on cutoff calibration: the large-sample Wilks approximation \cite{wilks1938large} may be unreliable when the nuisance dimension is large relative to the data.

Bayesian analysis fixes the data and examines the posterior $\pi(\boldsymbol\eta\mid\hat{\boldsymbol h})\propto\exp(-l(\boldsymbol\eta;\hat{\boldsymbol h})/\sigma^{2})\,\pi(\boldsymbol\eta)$, typically by Markov chain Monte Carlo sampling \cite{hines2014determination,raue2013joining,brooks1998markov}. A parameter is practically identifiable if its marginal or profile posterior is unimodal and concentrated, so that the credible interval is bounded and markedly narrower than the prior support \cite{raue2013joining,simpson2020practical}. This again characterizes a fixed-data set; moreover, the criterion depends on the prior, since a bounded prior makes every credible interval finite, and on sampler convergence, which fails along weakly constrained directions where chains drift or multimodality appears even when the posterior is proper \cite{raue2013joining,hines2014determination}.

The well-posedness notion examines the same relationship in the opposite direction: it perturbs the observations and measures the deviation of the recovered parameter. By relating a prescribed observational resolution to the largest local parameter deviation compatible with it, this notion quantifies data-to-parameter error propagation directly and characterizes reliability through continuous dependence on the data. It is also suited to high-dimensional analysis, since after local linearization coordinatewise and joint reliability are described by singular values, null spaces and effective information operators.
\subsubsection{Derivatives and the quadratic loss model}
Let \(\boldsymbol\eta^*=\log\boldsymbol\theta^*\) be an interior
local minimum of the least-squares loss, and write
\(\boldsymbol d:=\boldsymbol\eta-\boldsymbol\eta^*\).
For each observation time \(t_n\), define 
$J_{\eta,n}
:=
\nabla_\eta\boldsymbol h(t_n,\boldsymbol\eta^*)
\in\mathbb R^{L\times k},
~
\bigl(\boldsymbol H_{\eta,n}[\boldsymbol u,\boldsymbol w]\bigr)_l
:=
\boldsymbol u^\top
\nabla_\eta^2 h_l(t_n,\boldsymbol\eta^*)\boldsymbol w$. Stacking these quantities over observation times gives
\(J_\eta\in\mathbb R^{NL\times k}\) and
\(\boldsymbol H_\eta[\boldsymbol u,\boldsymbol w]\in\mathbb R^{NL}\).
Since the observation map is twice continuously differentiable,
\begin{equation}
\label{eq:h-expand}
\boldsymbol h(\boldsymbol t,\boldsymbol\eta^*+\boldsymbol d)
=
\boldsymbol h(\boldsymbol t,\boldsymbol\eta^*)
+
J_\eta\boldsymbol d
+
\frac12\boldsymbol H_\eta[\boldsymbol d,\boldsymbol d]
+
o(\|\boldsymbol d\|_2^2).
\end{equation}
Let \(D_*:=\operatorname{diag}(\boldsymbol\theta^*)\).
The chain rule gives \(J_\eta=J_\theta D_*\) and
\begin{equation}
\label{eq:log-hessian}
\nabla_\eta^2h_l
=
D_*\nabla_\theta^2h_l D_*
+
\operatorname{diag}
\bigl(\boldsymbol\theta^*\odot\nabla_\theta h_l\bigr),
\end{equation}
where the derivatives are evaluated at the corresponding
reference parameters. Equivalently,
\begin{equation}
\label{eq:H-transform}
\boldsymbol H_\eta[\boldsymbol d,\boldsymbol d]
=
\boldsymbol H_\theta[D_*\boldsymbol d,D_*\boldsymbol d]
+
J_\theta D_*(\boldsymbol d\odot\boldsymbol d).
\end{equation}
Because \(D_*\) is invertible,
\(\operatorname{range}(J_\eta)=\operatorname{range}(J_\theta)\).
Thus, with \(Q_\eta:=I_{NL}-J_\eta J_\eta^\dagger\),
\begin{equation}
\label{eq:projected-H-transform}
Q_\eta\boldsymbol H_\eta[\boldsymbol d,\boldsymbol d]
=
Q_\eta\boldsymbol H_\theta[D_*\boldsymbol d,D_*\boldsymbol d].
\end{equation}
The additional term introduced by the logarithmic transformation
therefore lies in the first-order response space and contributes
no normal second-order component.

Define the residual
\(\boldsymbol r^*
:=\boldsymbol h(\boldsymbol t,\boldsymbol\eta^*)
-\widehat{\boldsymbol h}\).
Stationarity gives \(J_\eta^\top\boldsymbol r^*=0\), and the
loss Hessian is
\begin{equation}
\label{eq:K-eta}
K_\eta
:=
\nabla_\eta^2l(\boldsymbol\eta^*;\widehat{\boldsymbol h})
=
F_\eta+C_\eta
\succeq0,
~
F_\eta:=J_\eta^\top J_\eta,
~
C_\eta:=
\sum_{n=1}^{N}\sum_{l=1}^{L}
r_{n,l}^*\nabla_\eta^2h_l(t_n,\boldsymbol\eta^*).
\end{equation}
For perturbed observations
\(\widehat{\boldsymbol h}-\boldsymbol\delta\), it follows that
\begin{equation}
\label{eq:joint-quadratic-loss-expansion}
\begin{aligned}
l(\boldsymbol\eta^*+\boldsymbol d;
   \widehat{\boldsymbol h}-\boldsymbol\delta)
=
l(\boldsymbol\eta^*;
   \widehat{\boldsymbol h}-\boldsymbol\delta)
+
\frac12\boldsymbol d^\top K_\eta\boldsymbol d
+
\boldsymbol\delta^\top J_\eta\boldsymbol d
+
o\bigl(\|\boldsymbol d\|_2^2+\|\boldsymbol\delta\|_2^2\bigr),
\end{aligned}
\end{equation}
as \((\boldsymbol d,\boldsymbol\delta)\to(0,0)\), with
\(\boldsymbol r^*\) fixed.

\begin{assumption}[Quadratic loss approximation]
\label{ass:quadratic-loss-model}
In the subsequent quadratic-model analysis, the perturbed loss
increment is replaced by
\begin{equation}
\label{eq:quadratic-loss-model}
q^{\boldsymbol\delta}(\boldsymbol d)
:=
\frac12\boldsymbol d^\top K_\eta\boldsymbol d
+
\boldsymbol\delta^\top J_\eta\boldsymbol d.
\end{equation}
This retains terms of total degree at most two in
\((\boldsymbol d,\boldsymbol\delta)\), including the residual
correction \(C_\eta\).
\end{assumption}

For the fixed compact neighborhood
\(\mathcal D_\eta:=U_\eta-\boldsymbol\eta^*\), define
$\boldsymbol{\eta}_{\mathrm{quad}}^{\boldsymbol\delta}
 \in \mathcal {M}_{\mathrm{quad}}^{\boldsymbol\delta}
:=
\operatorname*{argmin}_{\boldsymbol d\in\mathcal D_\eta}
q^{\boldsymbol\delta}(\boldsymbol d)$.
Lipschitz stability of the quadratic loss approximation means that, for some
\(C,\varepsilon_0>0\),
$
\sup_{\mathcal {M}_{\mathrm{quad}}^{\boldsymbol\delta}}\|\boldsymbol d\|_2
\leq C\|\boldsymbol\delta\|_2,
~
\|\boldsymbol\delta\|_2<\varepsilon_0$.
This conclusion concerns the quadratic loss
approximation; higher-order terms may alter identifiability in
the original nonlinear problem.

\subsubsection{Second-order derivative effects on parameter identifiability}
Definition~\ref{def:lip} is decided by the loss Hessian $K_\eta$ of
Eq.~\eqref{eq:K-eta}, which retains the second-order information that
the Fisher approximation $K_\eta\approx F_\eta$ discards. With a zero
residual, $K_\eta=F_\eta$ and the Hessian tensor $\boldsymbol H_\eta$
acts only through the second directional response: it controls
$\ker(F_\eta)$ at a H\"older rate (Theorem~\ref{prop:holder}), and the
Lipschitz rate returns for full-rank $F_\eta$
(Corollary~\ref{cor:fullrank}). With a nonzero residual,
$C_\eta=\boldsymbol r^*\otimes\boldsymbol H_\eta$ corrects the Fisher
information, and $K_\eta\succ0$ yields Lipschitz stability
(Theorem~\ref{thm:local_lip},
Corollaries~\ref{cor:K-nonidentifiability}--\ref{cor:quadratic-identifiability}).
Theorem~\ref{thm:crossover} shows that the residual size $\nu$ and the
perturbation size $\varepsilon$ select the rate, with crossover at
$\varepsilon_*=\gamma\nu/c_2$.

Lemma~\ref{lem:localization}
keeps all perturbed minimizers near $\boldsymbol\eta^*$, which justifies
the local expansions. For a perturbation $\boldsymbol\delta\in\mathbb R^{NL}$ write
\begin{equation}
\label{eq:perturbed-loss-notation}
l_{\boldsymbol\delta}(\boldsymbol\eta)
:=l(\boldsymbol\eta;\widehat{\boldsymbol h}-\boldsymbol\delta)
=\tfrac12\bigl\|\boldsymbol h(\boldsymbol t,\boldsymbol\eta)
-\widehat{\boldsymbol h}+\boldsymbol\delta\bigr\|_2^2,
\quad
\boldsymbol{\eta}_{\boldsymbol\delta}
:=\operatorname*{argmin}_{\boldsymbol\eta\in U_\eta}
l_{\boldsymbol\delta}(\boldsymbol\eta),
\end{equation}
where $U_\eta=\overline B(\boldsymbol\eta^*,R)$ is the closed ball of
Assumption~\ref{ass:admissible}.
By continuity and compactness, $\boldsymbol{\eta}_{\boldsymbol\delta}$ is continuous for $\boldsymbol{\delta}$, and
$\boldsymbol{\eta}_{\boldsymbol0}=\boldsymbol\eta^*$ by the uniqueness in
Assumption~\ref{ass:admissible}.
The matrices $J_\eta$, $F_\eta=J_\eta^\top J_\eta$, $C_\eta$ and
$K_\eta=F_\eta+C_\eta$ are evaluated at $\boldsymbol\eta^*$.

\paragraph{Hessian information under the zero-residual assumption.}
When the fit is zero-residual,
$\boldsymbol h(\boldsymbol t,\boldsymbol\eta^*)=\widehat{\boldsymbol h}$,
the correction term is negligible, $C_\eta=\boldsymbol0$ and
$K_\eta= F_\eta=J_\eta^\top J_\eta$, so the loss Hessian carries no
more information than the Fisher matrix.
If $F_\eta$ is singular, then $\ker(K_\eta)=\ker(J_\eta)=\ker(F_\eta)$
and first-order information says nothing about the component of
$\boldsymbol\eta_{\boldsymbol\delta}-\boldsymbol\eta^*$ along
$\ker(F_\eta)$; in the quadratic model that component is completely
unconstrained (Corollary~\ref{cor:quadratic-identifiability}). The next theorem shows that the second directional response of the model output, i.e.\ the Hessian tensor $\boldsymbol H_\eta$ itself, nevertheless controls this component, at the slower H\"older rate
$\varepsilon^{1/2}$, provided it is coercive on $\ker(F_\eta)$.

\begin{theorem}[H\"older stability under the zero-residual]
\label{prop:holder}
Let Assumption~\ref{ass:admissible} hold with
$\boldsymbol h(\boldsymbol t,\cdot)\in C^2$ on a neighborhood of
$U_\eta=\overline B(\boldsymbol\eta^*,R)$, and suppose that the residual approximately
vanishes, $\boldsymbol h(\boldsymbol t,\boldsymbol\eta^*)=\widehat{\boldsymbol h}$,
so that $K_\eta = F_\eta$.
Decompose
$\mathbb R^k=\mathcal R_\eta\oplus\mathcal N_\eta,
~
\mathcal R_\eta:=\operatorname{range}(J_\eta^\top)=\operatorname{range}(F_\eta),
~
\mathcal N_\eta:=\ker(J_\eta)=\ker(F_\eta)$,
orthogonally, with projectors $P_{\mathcal R}:=F_\eta^\dagger F_\eta$ and
$P_{\mathcal N}:=I_k-F_\eta^\dagger F_\eta$.
Let $\Pi:=J_\eta F_\eta^\dagger  J_\eta^\top$ be the orthogonal projector of
$\mathbb R^{NL}$ onto $\operatorname{range}(J_\eta)$ and
$\Pi^\perp:=I_{NL}-\Pi$.
Let $\boldsymbol H_\eta[\boldsymbol u,\boldsymbol v]\in\mathbb R^{NL}$,
$\bigl(\boldsymbol H_\eta[\boldsymbol u,\boldsymbol v]\bigr)_{(n,l)}
=\boldsymbol u^\top\nabla_\eta^2 h_l(t_n,\boldsymbol\eta^*)\boldsymbol v$,
be the second directional response at $\boldsymbol\eta^*$, let
$\sigma_+:=\sigma_{\min}^{+}(J_\eta)=\sqrt{\lambda_{\min}^{+}(F_\eta)}$ be the
smallest nonzero singular value of $J_\eta$, and define
\begin{align}
\label{eq:holder-constants}
M_\Pi&:=\sup_{0<\|\boldsymbol d\|_2\le R}
\frac{\bigl\|\Pi\bigl(\boldsymbol h(\boldsymbol t,\boldsymbol\eta^*+\boldsymbol d)
-\boldsymbol h(\boldsymbol t,\boldsymbol\eta^*)\bigr)-J_\eta\boldsymbol d\bigr\|_2}
{\|\boldsymbol d\|_2^2}, \nonumber \\
M_\perp&:=\sup_{\|\boldsymbol\gamma\|_2=1}
\bigl\|\Pi^\perp\boldsymbol H_\eta[\boldsymbol\gamma,\boldsymbol\gamma]\bigr\|_2 ,
\end{align}
both finite by the $C^2$ regularity.
Assume the second-order coercivity constant on $\mathcal N_\eta$ is
positive,
\begin{equation}
\label{eq:c2-def}
c_2:=\min_{\boldsymbol\beta\in\mathcal N_\eta,\ \|\boldsymbol\beta\|_2=1}
\bigl\|\Pi^\perp\boldsymbol H_\eta[\boldsymbol\beta,\boldsymbol\beta]\bigr\|_2>0 .
\end{equation}
Then there exists $\varepsilon_0>0$ such that for every
perturbation with $\varepsilon:=\|\boldsymbol\delta\|_2<\varepsilon_0$
and $\boldsymbol\eta_{\boldsymbol\delta}$,
the displacement
\[
\Delta\boldsymbol\eta:=\boldsymbol\eta_{\boldsymbol\delta}-\boldsymbol\eta^*
=\Delta\boldsymbol\eta_{\mathcal R}+\Delta\boldsymbol\eta_{\mathcal N},
~
\Delta\boldsymbol\eta_{\mathcal R}:=P_{\mathcal R}\Delta\boldsymbol\eta,
\quad
\Delta\boldsymbol\eta_{\mathcal N}:=P_{\mathcal N}\Delta\boldsymbol\eta,
\]
satisfies
\begin{equation}
\label{eq:holder-bound}
\|\Delta\boldsymbol\eta_{\mathcal R}\|_2
\le\frac{8}{3\sigma_+}\Bigl(1+\frac{8M_\Pi}{c_2}\Bigr)\varepsilon,
~
\|\Delta\boldsymbol\eta_{\mathcal N}\|_2
\le\frac{4}{\sqrt{c_2}}\,\varepsilon^{1/2},
\end{equation}
and, uniformly over $\mathcal M^{\boldsymbol\delta}$,
\begin{equation}
\label{eq:holder-total}
\|\boldsymbol\eta_{\boldsymbol\delta}-\boldsymbol\eta^*\|_2
\le\frac{2}{\sqrt{c_2}}\,\varepsilon^{1/2}\bigl(1+o(1)\bigr)
~(\varepsilon\to0).
\end{equation}
The range component thus obeys a Lipschitz law governed by the
first-order stiffness $\sigma_+$, whereas the null component obeys a
H\"older-$\tfrac12$ law governed by the second-order coercivity
$c_2$, and the latter dominates.
Up to second order, $\Delta\boldsymbol\eta_{\mathcal N}$ enters the
constraints only through
$\boldsymbol H_\eta[\Delta\boldsymbol\eta_{\mathcal N},\Delta\boldsymbol\eta_{\mathcal N}]$,
so the null component is determined only up to the sign symmetry
$\Delta\boldsymbol\eta_{\mathcal N}\mapsto-\Delta\boldsymbol\eta_{\mathcal N}$.
\end{theorem}

\begin{corollary}[Lipschitz stability under the zero-residual]
\label{cor:fullrank}
In the setting of Theorem~\ref{prop:holder}, suppose in addition that
$F_\eta$ has full rank $k$, equivalently
$\mathcal N_\eta=\ker F_\eta=\{\boldsymbol0\}$ and $P_{\mathcal N}=\boldsymbol0$.
Then the coercivity condition Eq.~\eqref{eq:c2-def} is void (by convention
$c_2=+\infty$, the minimization being over the empty set), the null
component vanishes identically,
\begin{equation}
\label{eq:cor-null}
\Delta\boldsymbol\eta_{\mathcal N}=\boldsymbol0,
~
\|\Delta\boldsymbol\eta\|_2=\|\Delta\boldsymbol\eta_{\mathcal R}\|_2 ,
\end{equation}
and there exists $\varepsilon_0>0$ such that
$\boldsymbol\eta_{\boldsymbol\delta}$
obeys the Lipschitz bound
\begin{equation}
\label{eq:cor-lip}
\|\boldsymbol\eta_{\boldsymbol\delta}-\boldsymbol\eta^*\|_2
\le\frac{8}{3\sigma_{\min}(J_\eta)}\,\varepsilon,
~
\varepsilon=\|\boldsymbol\delta\|_2<\varepsilon_0 ,
\end{equation}
where $\sigma_+=\sigma_{\min}(J_\eta)=\sqrt{\lambda_{\min}(F_\eta)}>0$.
The exponent thus returns from $\tfrac12$ to $1$, and the governing
constant passes from the second-order coercivity $c_2$ to the
first-order stiffness $\sigma_{\min}(J_\eta)$; this is the classical
first-order criterion, $J_\eta$ of full column rank ($F_\eta\succ0$)
being necessary and sufficient for Lipschitz identifiability when the
residual is zero.
\end{corollary}

\paragraph{Hessian correction of the Fisher information.}
Away from the zero-residual assumption the second-order information enters
$K_\eta$ through the contraction of the residual with the Hessian tensor,
$C_\eta=\boldsymbol r^*\otimes\boldsymbol H_\eta=\sum_a r^*_a\nabla^2_\eta h_a$,
which can be comparable to, or larger than, $F_\eta$ along weakly
informed directions. The residual $\boldsymbol r^*$ is now retained
throughout, no zero-residual assumption is imposed, and the stability
of the estimate is governed by $K_\eta$ rather than by $F_\eta$.

\begin{theorem}[Local Lipschitz stability of the $\ell_2$ loss]
\label{thm:local_lip}
Under Assumption~\ref{ass:admissible}, consider the original
nonlinear least-squares loss, and let $J_\eta$ and $K_\eta$ be
evaluated at $\boldsymbol\eta^*$, with $K_\eta$ defined by
Eq.~\eqref{eq:K-eta}.
The following statements hold.
\begin{enumerate}
\item
\emph{(Necessary condition.)}
If there exist $C>0$ and $\varepsilon_0>0$ such that
$\|\boldsymbol\eta_{\boldsymbol\delta}-\boldsymbol\eta^*\|_2
\le C\|\boldsymbol\delta\|_2,~
\|\boldsymbol\delta\|_2\le\varepsilon_0$, then
\begin{equation}
\label{eq:K-J-necessary}
\ker(K_\eta)\subseteq\ker(J_\eta).
\end{equation}

\item
\emph{(Sufficient condition.)}
If $K_\eta\succ0$, then there exists $\varepsilon_0>0$ such that
\begin{equation}
\label{eq:K-lipschitz-bound}
\|\boldsymbol\eta_{\boldsymbol\delta}-\boldsymbol\eta^*\|_2
\le
\frac{2\|J_\eta\|_2}{\lambda_{\min}(K_\eta)}
\|\boldsymbol\delta\|_2
~
\text{for all }
\boldsymbol\eta_{\boldsymbol\delta}\in\mathcal M^{\boldsymbol\delta},
\ \|\boldsymbol\delta\|_2\le\varepsilon_0 .
\end{equation}
Consequently, the parameter vector is practically identifiable
(Lipschitz stable) in the sense of Definition~\ref{def:lip}.

\item
\emph{(Equivalence under full column rank.)}
If $J_\eta$ has full column rank, equivalently $F_\eta\succ0$, then
the local Lipschitz estimate
$\|\boldsymbol\eta_{\boldsymbol\delta}-\boldsymbol\eta^*\|_2
\le C\|\boldsymbol\delta\|_2$
holds for some $C>0$ and all sufficiently small $\boldsymbol\delta$
if and only if $K_\eta\succ0$.
\end{enumerate}
\end{theorem}

\begin{corollary}[Practical non-identifiability of the $\ell_2$ loss]
\label{cor:K-nonidentifiability}
Under the hypotheses of Theorem~\ref{thm:local_lip}, the
parameter vector is \emph{practically non-identifiable in the Lipschitz
sense} if Lipschitz stability fails, i.e., if for every $C>0$ and
$\varepsilon_0>0$ there exist $\boldsymbol\delta$ with
$0<\|\boldsymbol\delta\|_2<\varepsilon_0$ and
$\boldsymbol\eta_{\boldsymbol\delta}$
with
$\|\boldsymbol\eta_{\boldsymbol\delta}-\boldsymbol\eta^*\|_2
>C\|\boldsymbol\delta\|_2$.
Then:
\begin{enumerate}
\item
If $\ker(K_\eta)\not\subseteq\ker(J_\eta)$, i.e., if there exists
$\boldsymbol v\neq\boldsymbol0$ with
$K_\eta\boldsymbol v=\boldsymbol0$ and
$J_\eta\boldsymbol v\neq\boldsymbol0$,
then the parameter vector is practically non-identifiable in the
Lipschitz sense.
\item
If $J_\eta$ has full column rank, then the parameter vector is
practically non-identifiable in the Lipschitz sense if and only if
$K_\eta$ is singular.
\end{enumerate}
\end{corollary}

\begin{corollary}[Practical identifiability in the quadratic loss approximation]
\label{cor:quadratic-identifiability}
Under Assumption~\ref{ass:quadratic-loss-model}, suppose that
$K_\eta\succeq0$ and that $\mathcal D_\eta$ is a fixed compact
neighborhood of the origin, and let
$\boldsymbol{\eta} ^{\boldsymbol\delta}_{\mathrm{quad}}
:=\operatorname*{argmin}_{\boldsymbol d\in\mathcal D_\eta}
q^{\boldsymbol\delta}(\boldsymbol d)$
with $q^{\boldsymbol\delta}$ as in Eq.~\eqref{eq:quadratic-loss-model}.
Then the full parameter vector is Lipschitz identifiable in the
quadratic model---that is,
$
\|\boldsymbol{\eta} ^{\boldsymbol\delta}_{\mathrm{quad}}-\boldsymbol{\eta^*}\|_2\le C\|\boldsymbol\delta\|_2$
for some $C>0$ and all sufficiently small $\boldsymbol\delta$---if and
only if $K_\eta\succ0$. More precisely:
\begin{enumerate}
\item
If $K_\eta\succ0$, then for all sufficiently small $\boldsymbol\delta$,
$\mathcal{M}^{\boldsymbol\delta}_{\mathrm{quad}}=\{-K_\eta^{-1}J_\eta^\top\boldsymbol\delta\}$, and the optimal Lipschitz coefficient is
\begin{equation}
\label{eq:quadratic-optimal-lipschitz}
L_{\mathrm{quad}}
=\|K_\eta^{-1}J_\eta^\top\|_2
\le\frac{\|J_\eta\|_2}{\lambda_{\min}(K_\eta)} .
\end{equation}

\item
If $K_\eta$ is singular, then
$\mathcal D_\eta\cap\ker(K_\eta)\neq\{\boldsymbol0\}$,
and for every $C>0$ and $\varepsilon_0>0$ there exist
$\boldsymbol\delta$ with $0<\|\boldsymbol\delta\|_2<\varepsilon_0$
and $\|\boldsymbol{\eta} ^{\boldsymbol\delta}_{\mathrm{quad}}-\boldsymbol{\eta^*}\|_2>C\|\boldsymbol\delta\|_2$.
The parameter vector is thus practically non-identifiable in the
quadratic-model Lipschitz sense.
\end{enumerate}
\end{corollary}

\paragraph{From zero to nonzero residual.}
Theorems~\ref{prop:holder} and~\ref{thm:local_lip} describe two limits.
Theorem~\ref{prop:holder} assumes $r^*=0$ and obtains
Hölder-$\tfrac12$ control of $\ker F_\eta$ for all small $\varepsilon$;
Theorem~\ref{thm:local_lip} assumes $K_\eta\succ0$, which for singular $F_\eta$ requires
$r^*\neq0$, and obtains Lipschitz control as $\varepsilon\to0$. Neither says
how the residual size and the perturbation size together decide between the
two rates. Let $\nu:=\|r^*\|_2$ and, for $\nu>0$, $\hat r:=r^*/\nu$. Since
$J_\eta^\top r^*=0$, the residual lies in $\operatorname{range}(J_\eta)^\perp$,
and every entry of the correction is a contraction of $r^*$ with the normal
acceleration,
\begin{equation}
\label{eq:C-normal}
x^\top C_\eta y=\bigl\langle r^*,\Pi^\perp H_\eta[x,y]\bigr\rangle,
~
|x^\top C_\eta y|\le \nu M_\perp\|x\|_2\|y\|_2 ,
\end{equation}
with $M_\perp$ of Eq.~\eqref{eq:holder-constants}. The same object $\Pi^\perp H_\eta$ thus governs
both theorems: through its norm when $r^*=0$, and through its component along
$\hat r$ when $r^*\neq0$. The coercivity constant $c_2$ of Eq.~\eqref{eq:c2-def} measures
the first; the second is measured by the residual-directional
curvature
\begin{equation}
\label{eq:gamma}
\gamma:=\min_{\xi\in\mathcal N_\eta,\ \|\xi\|_2=1}
\bigl\langle\hat r,\Pi^\perp H_\eta[\xi,\xi]\bigr\rangle ,
~ 0\le\gamma\le c_2 .
\end{equation}
In the notation of Lemma~\ref{thm:residual-schur} with
$\operatorname{range}(V)=\mathcal N_\eta$,
$\gamma\nu=\lambda_{\min}(D)$ is the curvature that the residual adds to the
null directions, and Lemma~\ref{lem:residual-structure} below shows that the
Schur complement of Lemma~\ref{thm:residual-schur}~(3) satisfies
$\tfrac12\gamma\nu\le\mu:=\lambda_{\min}(S)\le\gamma\nu$ for small $\nu$.
The perturbation adds curvature of size at most $c_2\varepsilon$ to the same
directions (through $\langle\delta,\Pi^\perp H_\eta[\xi,\xi]\rangle$, whose
coefficient is at least $c_2$ in norm), so the two competing curvatures are
equal at the \emph{crossover scale}
\begin{equation}
\label{eq:eps-star}
\varepsilon_*:=\frac{\gamma\nu}{c_2}=\hat\gamma\,\nu,
~
\hat\gamma:=\frac{\gamma}{c_2}\in[0,1],
\end{equation}
a fraction $\hat\gamma$ of the residual norm, with $\hat\gamma$ the part of the
normal acceleration of the null directions that points along the residual.
Both sides of Eq.~\eqref{eq:eps-star} carry the units of the data.
\begin{theorem}[Residual and perturbation size decide the order of stability]
\label{thm:crossover}
Let Assumption~\ref{ass:admissible} hold, let $F_\eta$ be singular with
$c_2>0$ as in Eq.~\eqref{eq:c2-def}, and fix $\gamma_0\in(0,c_2]$. Assume that either $\nu=0$, or
$\nu>0$ and $\gamma\ge\gamma_0$ in Eq.~\eqref{eq:gamma}; set $\varepsilon_*:=0$ when
$\nu=0$. There are $\nu_1>0$ and $C_{\mathcal N},C_{\mathcal R}>0$, depending only
on $\sigma_+$, $M_\Pi$, $M_\perp$, $c_2$, $\gamma_0$ and, through the radius
$\varrho$ of Eq.~\eqref{eq:rho-cond}, on the radius $R$ of $U_\eta$ in
Assumption~\ref{ass:admissible} and the modulus $\omega$ of
Eq.~\eqref{eq:pf-taylor}, and there is $\varepsilon_0>0$, such that for $\nu\le\nu_1$,
$\varepsilon=\|\delta\|_2<\varepsilon_0$ and every
$\eta^\delta\in\mathcal M^\delta$, with $\varepsilon_*=\gamma\nu/c_2$ of
Eq.~\eqref{eq:eps-star},
\begin{equation}
\label{eq:cross-main}
\|\Delta\eta_{\mathcal N}\|_2\le C_{\mathcal N}\,
\frac{\varepsilon}{\sqrt{\varepsilon+\varepsilon_*}},
~
\|\Delta\eta_{\mathcal R}\|_2\le C_{\mathcal R}\,\varepsilon .
\end{equation}
The following statements hold.
\begin{enumerate}
\item[(1)] Lipschitz stability. If $\nu>0$ and $\varepsilon\le\varepsilon_*$, then
$\|\eta^\delta-\eta^*\|_2\le\bigl(C_{\mathcal N}\varepsilon_*^{-1/2}+C_{\mathcal R}\bigr)\,\varepsilon$.
\item[(2)] Hölder stability. If $\varepsilon\ge\varepsilon_*$, then
$\|\Delta\eta_{\mathcal N}\|_2\le C_{\mathcal N}\,\varepsilon^{1/2}$.
\item[(3)] Limits. The two regimes are separated by the line
$\varepsilon=\hat\gamma\nu$ in the $(\nu,\varepsilon)$ plane. For $\nu=0$ only (2)
occurs, with the rates of Theorem~\ref{prop:holder}. For fixed $\nu>0$, (1)
gives Lipschitz stability as $\varepsilon\to0$, as in
Theorem~\ref{thm:local_lip}, on a window of width
$\hat\gamma\nu\le\nu$ that closes linearly as $\nu\to0$, with a constant of
order $(\hat\gamma\nu)^{-1/2}$.
\end{enumerate}
\end{theorem}
The bound \eqref{eq:cross-main} is uniform in $\varepsilon$ but conservative below $\varepsilon_*$. Restricted to a null direction $\boldsymbol\xi$ and to second order in the model output, the perturbed loss is a quartic potential in the null coordinate $y$ with quadratic coefficient $\gamma_{\boldsymbol\xi}\nu+\langle\boldsymbol\delta,\Pi^\perp H_\eta[\boldsymbol\xi,\boldsymbol\xi]\rangle$, $\gamma_{\boldsymbol\xi}:=\langle\hat{\boldsymbol r},\Pi^\perp H_\eta[\boldsymbol\xi,\boldsymbol\xi]\rangle$. While the perturbation cannot make this coefficient negative, the minimizer is $y=0$ and the null component does not move; once it can, the minimizer bifurcates to $y=\pm[2(c_{\boldsymbol\xi}\varepsilon-\gamma_{\boldsymbol\xi}\nu)]^{1/2}/c_{\boldsymbol\xi}$, $c_{\boldsymbol\xi}:=\|\Pi^\perp H_\eta[\boldsymbol\xi,\boldsymbol\xi]\|_2$, a pitchfork at $\varepsilon=\gamma_{\boldsymbol\xi}\nu/c_{\boldsymbol\xi}$. Below $\varepsilon_*$ the displacement is carried by the range component; the $\varepsilon^{1/2}$ law appears only above it.
\subsection{Coordinate parameter identifiability}
\subsubsection{Coordinate parameter identifiability definintion}
Coordinate parameter identifiability concerns the recovery of a single parameter $\eta_p$ while the remaining parameters are conditionally reoptimized as nuisance parameters, following the nested optimization structure of profile likelihood \cite{raue2009structural,raue2013joining}: the nuisance parameters are optimized at fixed $\eta_p$, and the profiled objective is then minimized over $\eta_p$. Rather than testing whether a confidence interval is bounded at a prescribed threshold, we ask whether the profiled minimizer remains stably localized near $\eta_p^*$ under data perturbations. The resulting notion quantifies the rate at which the displacement of the coordinate estimate vanishes as the perturbation tends to zero.
\begin{assumption}[Local continuity and coordinate uniqueness]
\label{ass:P}
Write $\boldsymbol\eta=(\tilde{\boldsymbol\eta},\eta_p)$ and
$\boldsymbol\eta^*=(\tilde{\boldsymbol\eta}^*,\eta_p^*)$.
There exist a compact interval
$I=[\eta_p^*-R_p,\eta_p^*+R_p]$, $R_p>0$, and a compact neighborhood
$\tilde U_\eta$ of $\tilde{\boldsymbol\eta}^*$, both independent of
$\boldsymbol\delta$, such that $U_{\eta,p}:=\tilde U_\eta\times I$ lies in
the logarithmic parameter domain and contains
$\boldsymbol\eta^*$ in its interior.
The observation map $\boldsymbol\eta\mapsto\boldsymbol h(\boldsymbol t,\boldsymbol\eta)$
is continuous on $U_{\eta,p}$, $\boldsymbol\eta^*$ minimizes the
unperturbed loss on $U_{\eta,p}$, and every unperturbed minimizer shares
its $p$-th coordinate:
\begin{equation}
\label{eq:coordinate-reference-uniqueness}
\Bigl\{\boldsymbol e_p^\top\boldsymbol\eta:
\boldsymbol\eta\in
\operatorname*{arg\,min}_{\boldsymbol\zeta\in U_{\eta,p}}
l(\boldsymbol\zeta;\widehat{\boldsymbol h})\Bigr\}
=\{\eta_p^*\}.
\end{equation}
\end{assumption}
Uniqueness is required only of the reference coordinate, not of the full
parameter vector: multiple nuisance-parameter minimizers are allowed, and
continuity with compactness guarantees that the conditional and joint
minima below are attained.

For perturbed observations $\widehat{\boldsymbol h}-\boldsymbol\delta$ let
$\mathcal M_p^{\boldsymbol\delta}
:=\operatorname*{arg\,min}_{\boldsymbol\eta\in U_{\eta,p}}
l(\boldsymbol\eta;\widehat{\boldsymbol h}-\boldsymbol\delta)$
be the localized estimate set, and for fixed $\eta_p\in I$ define the
conditional nuisance minimizers
$\widetilde\Omega^{\boldsymbol\delta}(\eta_p)
:=\operatorname*{arg\,min}_{\tilde{\boldsymbol\eta}\in\tilde U_\eta}
l\bigl((\tilde{\boldsymbol\eta},\eta_p);\widehat{\boldsymbol h}-\boldsymbol\delta\bigr)$
and the profile value
\begin{equation}
\label{eq:profile-perturbed}
l^{\boldsymbol\delta}(\eta_p)
:=\min_{\tilde{\boldsymbol\eta}\in\tilde U_\eta}
l\bigl((\tilde{\boldsymbol\eta},\eta_p);\widehat{\boldsymbol h}-\boldsymbol\delta\bigr).
\end{equation}
The profile value is well defined without selecting a nuisance branch
\cite{terazono2015continuity}, and
$(\eta_p,\boldsymbol\delta)\mapsto l^{\boldsymbol\delta}(\eta_p)$ is
continuous because the loss is jointly continuous and $\tilde U_\eta$ is
compact.
Hence the profiled coordinate set
$\Omega_p^{\boldsymbol\delta}
:=\operatorname*{arg\,min}_{\eta_p\in I} l^{\boldsymbol\delta}(\eta_p)$
is nonempty and compact, and it coincides with the coordinate
projection of the joint minimizers,
$\Omega_p^{\boldsymbol\delta}
=\{\boldsymbol e_p^\top\boldsymbol\eta:
\boldsymbol\eta\in\mathcal M_p^{\boldsymbol\delta}\}$:
a joint minimizer minimizes the loss over the nuisance parameters at its
own $p$-th coordinate, which in turn minimizes the profile value, and
conversely any profile minimizer combined with a conditional nuisance
minimizer is a joint minimizer.
The following definition therefore concerns only this coordinate
projection and requires neither stable nor unique recovery of the
nuisance parameters.

\begin{definition}[Coordinate practical identifiability]
\label{def:coord-general}
Under Assumption~\ref{ass:P}, the coordinate $\eta_p$ is practically
identifiable at $\widehat{\boldsymbol h}$ with stability index
$\omega_p$ if there exist $C_p>0$ and $\varepsilon_p>0$ such that
\begin{equation}
\label{eq:coord-index-stability}
\sup_{\eta_p\in\Omega_p^{\boldsymbol\delta}}|\eta_p-\eta_p^*|
\le C_p\,\omega_p(\|\boldsymbol\delta\|_2)
~(\|\boldsymbol\delta\|_2<\varepsilon_p),
\end{equation}
where $\omega_p:[0,\infty)\to[0,\infty)$ is continuous, strictly
increasing, and $\omega_p(0)=0$.
\end{definition}

Since $\omega_p(\varepsilon)\to0$, the estimate forces
$\Omega_p^{\boldsymbol\delta}\subset\operatorname{int}I$ for small
$\varepsilon_p$, so interiority of the perturbed coordinate estimates
need not be assumed separately.
The Lipschitz case $\omega_p(\varepsilon)=\varepsilon$ gives
$\sup_{\eta_p\in\Omega_p^{\boldsymbol\delta}}|\eta_p-\eta_p^*|
\le C_p\|\boldsymbol\delta\|_2$.
These estimates measure perturbation-induced changes of the recovered
coordinate rather than the width of a likelihood confidence interval at
a fixed noise level.

\subsubsection{Coordinate splitting of the first- and second-order derivatives}
\paragraph{Coordinate splitting}
Fix a coordinate index $p\in\{1,\dots,k\}$ and split
$\boldsymbol\eta=[\tilde{\boldsymbol\eta};\eta_p]$ into the nuisance block
$\tilde{\boldsymbol\eta}\in\mathbb R^{k-1}$ and the profiled coordinate
$\eta_p\in\mathbb R$. Define the permutation matrix
\begin{equation}
\label{eq:P}
  P:=I_k-(\boldsymbol e_p-\boldsymbol e_k)(\boldsymbol e_p-\boldsymbol e_k)^\top
   \in\mathbb R^{k\times k},
  \quad  P\Delta\boldsymbol\eta=\big[\Delta\tilde{\boldsymbol\eta};\Delta\eta_p\big],
\end{equation}
which interchanges the $p$-th and $k$-th coordinates and satisfies
$P^\top=P$, $P^2=I_k$, $\|P\|_2=1$; consequently
$\|\Delta\boldsymbol\eta\|_2^2
=\|\Delta\tilde{\boldsymbol\eta}\|_2^2+(\Delta\eta_p)^2$. The associated
index mapping and selection matrix are
\begin{align}
\label{eq:pi}
  &\pi:\ \pi(p)=k,\quad\pi(k)=p,\quad\pi(a)=a\ \ (a\neq p,k);\nonumber\\
  &\boldsymbol E=\big[\boldsymbol e_{\pi(1)}^\top;\dots;
     \boldsymbol e_{\pi(k-1)}^\top\big]\in\mathbb R^{(k-1)\times k},
  \quad P=\big[\boldsymbol E;\boldsymbol e_p^\top\big],
\end{align}
so that $\boldsymbol E$ consists of the first $k-1$ rows of $P$ and
\begin{equation}
\label{eq:EE}
  \boldsymbol E\boldsymbol E^\top=I_{k-1},
~
  \boldsymbol E^\top\boldsymbol E=I_k-\boldsymbol e_p\boldsymbol e_p^\top,
~
  \Delta\boldsymbol\eta
   =\boldsymbol E^\top\Delta\tilde{\boldsymbol\eta}
    +\boldsymbol e_p\Delta\eta_p .
\end{equation}
 
\paragraph{Splitting of the first-and second-order derivatives}
Since $P^\top=[\boldsymbol E^\top~|~\boldsymbol e_p]$,
\begin{align}
\label{eq:Jsplit}
&\boldsymbol {J_\eta}P^\top=\big[\tilde{\boldsymbol J}~|~j_p\big],
~
   \tilde{\boldsymbol J}:=\boldsymbol J\boldsymbol E^\top
=\nabla_{\tilde{\boldsymbol\eta}}\boldsymbol h
    \in\mathbb R^{NL\times(k-1)},\nonumber \\
& j_p:=\boldsymbol J_\eta\boldsymbol e_p
    =\frac{\partial\boldsymbol h}{\partial\eta_p}\in\mathbb R^{NL},~
  \boldsymbol J_\eta\Delta\boldsymbol\eta
   =\tilde{\boldsymbol J}\Delta\tilde{\boldsymbol\eta}+j_p\Delta\eta_p;
\end{align}
that is, $\tilde{\boldsymbol J_\eta}$ is $\boldsymbol J$ with its $p$-th column
deleted, and $j_p$ is that column. Write $\mathcal H^{(n,l)}:=\nabla_{\boldsymbol\eta}^2h_l(t_n,\boldsymbol\eta^*)
\in\mathbb R^{k\times k}$ for the Hessian of the $(n,l)$-th model output, so that
$\boldsymbol H_\eta[\Delta\boldsymbol\eta,\Delta\boldsymbol\eta]_{(n,l)}
=\Delta\boldsymbol\eta^\top\mathcal H^{(n,l)}\Delta\boldsymbol\eta$.
Congruence by $P$ splits each $\mathcal H^{(n,l)}$ into three blocks,
\begin{align}
\label{eq:Hsplit}
  P\mathcal H^{(n,l)}P^\top
   &=\begin{bmatrix}
      \tilde{\mathcal H}^{(n,l)} & \boldsymbol g^{(n,l)}\\[2pt]
      \big(\boldsymbol g^{(n,l)}\big)^\top & \zeta^{(n,l)}
    \end{bmatrix}
   =\begin{bmatrix}
      \boldsymbol E\mathcal H^{(n,l)}\boldsymbol E^\top &
      \boldsymbol E\mathcal H^{(n,l)}\boldsymbol e_p\\[2pt]
      \boldsymbol e_p^\top\mathcal H^{(n,l)}\boldsymbol E^\top &
      \boldsymbol e_p^\top\mathcal H^{(n,l)}\boldsymbol e_p
    \end{bmatrix},\\[4pt]
  \tilde{\mathcal H}^{(n,l)}&\in\mathbb R^{(k-1)\times(k-1)},
   \quad\boldsymbol g^{(n,l)}\in\mathbb R^{k-1},
   \quad\zeta^{(n,l)}\in\mathbb R ,\nonumber
\end{align}
where $\tilde{\mathcal H}^{(n,l)}$ is $\mathcal H^{(n,l)}$ with its $p$-th row
and $p$-th column deleted, $\boldsymbol g^{(n,l)}$ collects the mixed
derivatives $\partial^2h_l/\partial\eta_{\pi(\cdot)}\partial\eta_p$, and
$\eta^{(n,l)}=\partial^2h_l/\partial\eta_p^2$. The factor $2$ in the cross
term below comes from the two off-diagonal blocks being transposes of each
other:
\begin{equation}
\label{eq:Hexp}
  \boldsymbol H_\eta[\Delta\boldsymbol\eta,\Delta\boldsymbol\eta]_{(n,l)}
   =\Delta\tilde{\boldsymbol\eta}^\top\tilde{\mathcal H}^{(n,l)}
     \Delta\tilde{\boldsymbol\eta}
    +2\big(\boldsymbol g^{(n,l)}\big)^\top\Delta\tilde{\boldsymbol\eta}\,
      \Delta\eta_p
    +\zeta^{(n,l)}(\Delta\eta_p)^2 .
\end{equation}
Stacking over $(n,l)$ gives one bilinear map, one matrix and one vector,
\begin{align}
\label{eq:blocks}
  &\tilde{\boldsymbol H}[\boldsymbol u,\boldsymbol w]_{(n,l)}
   :=\boldsymbol u^\top\boldsymbol E\mathcal H^{(n,l)}\boldsymbol E^\top
     \boldsymbol w,
   \quad
   \tilde{\boldsymbol H}[\cdot,\cdot]:\mathbb R^{k-1}\times\mathbb R^{k-1}
     \to\mathbb R^{NL},\nonumber\\
  &\boldsymbol G_{(n,l)}
   :=\boldsymbol e_p^\top\mathcal H^{(n,l)}\boldsymbol E^\top
   \in\mathbb R^{1\times(k-1)},
   \quad\boldsymbol G\in\mathbb R^{NL\times(k-1)},\nonumber\\
  &\boldsymbol\zeta_{(n,l)}
   :=\boldsymbol e_p^\top\mathcal H^{(n,l)}\boldsymbol e_p ,
   \quad\boldsymbol\zeta\in\mathbb R^{NL} .
\end{align}
Equivalently $\tilde{\boldsymbol H}[\boldsymbol u,\boldsymbol w]
=\boldsymbol H[\boldsymbol E^\top\boldsymbol u,\boldsymbol E^\top\boldsymbol w]$,
and $\boldsymbol G=\partial j_p/\partial\tilde{\boldsymbol\eta}
=\partial\tilde{\boldsymbol J}/\partial\eta_p$.
 
\subsubsection{Coordinate Lipschitz stability}
We use the derivative splittings in
Eqs.~\eqref{eq:Jsplit}--\eqref{eq:blocks} and retain the
residual contribution to $K_\eta$.
The nuisance parameters are eliminated through a generalized
Schur complement, without requiring a unique or smooth
conditional minimum.

\noindent Recall that
\[
P^\top=[\boldsymbol E^\top~|~\boldsymbol e_p],
~
P\Delta\boldsymbol\eta
=
\begin{bmatrix}
\Delta\tilde{\boldsymbol\eta}\\
\Delta\eta_p
\end{bmatrix},
~
\boldsymbol J_\eta P^\top
=
[\tilde{\boldsymbol J}~|~j_p].
\]
The residual-corrected Hessian admits the block decomposition
\begin{align}
PK_\eta P^\top
=
\begin{pmatrix}
\tilde K_\eta&\boldsymbol k_{\eta,p}\\
\boldsymbol k_{\eta,p}^\top&K_{\eta,pp}
\end{pmatrix},~\tilde K_\eta
:=
\boldsymbol E K_\eta\boldsymbol E^\top
=
\tilde{\boldsymbol J}^{\top}\tilde{\boldsymbol J}
+
\sum_{n=1}^N\sum_{l=1}^L
r_{n,l}^*\tilde{\mathcal H}^{(n,l)},
\nonumber\\
\boldsymbol k_{\eta,p}
:=
\boldsymbol E K_\eta\boldsymbol e_p
=
\tilde{\boldsymbol J}^{\top}j_p
+
\boldsymbol G^\top\boldsymbol r^*,~
K_{\eta,pp}
:=
\boldsymbol e_p^\top K_\eta\boldsymbol e_p
=
j_p^\top j_p
+
\langle\boldsymbol r^*,\boldsymbol\zeta\rangle.
\end{align}
\begin{theorem}[Coordinate Lipschitz stability of the quadratic loss approximation]
\label{thm:profile-quadratic-lipschitz}
Suppose that $K_\eta\succeq0$ and
$\ker(\tilde K_\eta)\subseteq\ker(\tilde{\boldsymbol J})$.
Define
$\kappa_p^K
:=K_{\eta,pp}
-\boldsymbol k_{\eta,p}^{\top}
 \tilde K_\eta^\dagger\boldsymbol k_{\eta,p},
~\boldsymbol g_p^K
:=j_p-\tilde{\boldsymbol J}
\tilde K_\eta^\dagger\boldsymbol k_{\eta,p}$. Then $\kappa_p^K\ge0$, and the nuisance minimum defining
\[
q_p^{\boldsymbol\delta}(x)
:=
\min_{\boldsymbol u\in\mathbb R^{k-1}}
q^{\boldsymbol\delta}
(\boldsymbol E^\top\boldsymbol u+\boldsymbol e_p x)
\]
is attained for every $x$ and $\boldsymbol\delta$. The profiled coordinate is locally Lipschitz stable,
for the entire minimizer set on $I-\eta_p^*$,
if and only if
\begin{equation}
\label{eq:coordinate-K-nullspace}
\kappa_p^K>0
\quad\Longleftrightarrow\quad
\boldsymbol e_p^\top\boldsymbol z=0
\quad\text{for every }\boldsymbol z\in\ker(K_\eta).
\end{equation}
In this case, for sufficiently small perturbations,
the unique coordinate minimum is
\begin{equation}
\label{eq:quadratic-coordinate-response}
\Delta\eta_p^{\mathrm{quad}}(\boldsymbol\delta)
=
-\frac{(\boldsymbol g_p^K)^\top\boldsymbol\delta}
       {\kappa_p^K},
~
L_p^K=\frac{\|\boldsymbol g_p^K\|_2}{\kappa_p^K},
\end{equation}
where $L_p^K$ is the exact Lipschitz coefficient.
If $\kappa_p^K=0$, the coordinate is practically
non-identifiable in this quadratic-model Lipschitz sense.
\end{theorem}

When $C_\eta=0$, the block definitions reduce to
$\tilde K_\eta
=
\tilde{\boldsymbol J}^{\top}\tilde{\boldsymbol J},
~
\boldsymbol k_{\eta,p}
=
\tilde{\boldsymbol J}^{\top}j_p,
~
K_{\eta,pp}=j_p^\top j_p$. The compatibility condition is then automatic, and
\[
\boldsymbol g_p^K
=
\left(
I-\tilde{\boldsymbol J}\tilde{\boldsymbol J}^\dagger
\right)j_p,
~
\kappa_p^K=\|\boldsymbol g_p^K\|_2^2.
\]
Thus, we obtain the constant $(\kappa_p^K)^{-1/2}$.
For nonzero residual corrections, this identity
need not hold: $\kappa_p^K$ measures profile curvature,
whereas $\|\boldsymbol g_p^K\|_2/\kappa_p^K$
measures the quadratic-model sensitivity to data perturbations.

\subsubsection{Hessian correction of coordinate identifiability and its global bounds}
\paragraph{Residual-Hessian correction to coordinate Lipschitz stability.}
We compare the Gauss--Newton and residual-corrected quadratic
models at the same fitted parameter $\boldsymbol\eta^*$,
holding $J_\eta$ fixed.
Assume the hypotheses of
Theorem~\ref{thm:profile-quadratic-lipschitz}, with
$\kappa_p^K>0$.
Define 
$\boldsymbol g_p^F
=
(I-\widetilde J\widetilde J^\dagger)\boldsymbol j_p,
~
\kappa_p^F=\|\boldsymbol g_p^F\|_2^2,
~
L_p^F=\frac{1}{\sqrt{\kappa_p^F}}$.
The coordinate response in
Eq.~\eqref{eq:quadratic-coordinate-response}
separates effective data forcing from profiled restoring
curvature. Accordingly, the residual-Hessian correction
to the Lipschitz coefficient factors as
\begin{equation}
\mathcal R_p^{\mathrm{RH}}
:=
\frac{L_p^K}{L_p^F}
=
\underbrace{
\frac{\|\boldsymbol g_p^K\|_2}{\|\boldsymbol g_p^F\|_2}
}_{\mathcal R_{p,\mathrm{resp}}^{\mathrm{RH}}}
\underbrace{
\frac{\kappa_p^F}{\kappa_p^K}
}_{\mathcal R_{p,\mathrm{curv}}^{\mathrm{RH}}}.
\label{eq:rh-total-ratio}
\end{equation}
Values below or above one indicate improved or reduced
coordinate Lipschitz stability, respectively.
Equality preserves the worst-case amplification, although
the response to an individual perturbation may change.

To interpret the response factor, define the change in
nuisance compensation by $\boldsymbol d_p
:=
\widetilde K_\eta^\dagger\boldsymbol k_{\eta,p}
-\widetilde J^\dagger\boldsymbol j_p$. Then
$\boldsymbol g_p^K
=\boldsymbol g_p^F-\widetilde J\boldsymbol d_p$.
Since $\boldsymbol g_p^F$ is orthogonal to
$\operatorname{range}(\widetilde J)$,
\begin{equation}
\|\boldsymbol g_p^K\|_2^2
=
\kappa_p^F+\|\widetilde J\boldsymbol d_p\|_2^2,
~
\mathcal R_{p,\mathrm{resp}}^{\mathrm{RH}}
=
\sqrt{1+
\frac{\|\widetilde J\boldsymbol d_p\|_2^2}{\kappa_p^F}}
\ge1.
\label{eq:rh-response-orthogonal-decomposition}
\end{equation}
Thus, changing nuisance compensation can only increase
the effective response norm relative to its Gauss--Newton
value, with equality precisely when
$\widetilde J\boldsymbol d_p=0$.
Consequently, coordinate stability improves if and only if
\begin{equation}
\frac{\kappa_p^K}{\kappa_p^F}
>
\sqrt{1+
\frac{\|\widetilde J\boldsymbol d_p\|_2^2}{\kappa_p^F}}.
\label{eq:rh-improvement-condition}
\end{equation}
The relative curvature gain must therefore exceed the
relative amplification of the effective response. Equivalently, the logarithmic correction separates additively:
\begin{equation}
\mathcal I_p^{\mathrm{RH}}
:=
\log\mathcal R_p^{\mathrm{RH}}
=
\underbrace{
\frac12\log\left(
1+\frac{\|\widetilde J\boldsymbol d_p\|_2^2}{\kappa_p^F}
\right)
}_{\mathcal I_{p,\mathrm{resp}}^{\mathrm{RH}}\ge0}
\underbrace{
-\log\frac{\kappa_p^K}{\kappa_p^F}
}_{\mathcal I_{p,\mathrm{curv}}^{\mathrm{RH}}}.
\label{eq:rh-log-decomposition}
\end{equation}
The response contribution is nonnegative, while increased
profiled curvature contributes negatively.
These exact endpoint comparisons require no assumption
that the two nuisance null spaces coincide.

\paragraph{Global bounds on coordinate residual-Hessian corrections.}
The following bounds follow from the block inverse formula and the spectral bounds of a congruence transform~\cite{horn2012matrix}; we record them because they turn the single scalar $\rho$ into a diagnostic for every coordinate at once.

\begin{proposition}[Global-to-coordinate residual-Hessian bounds]
\label{thm:rh-global-coordinate-bounds}
Suppose that $F_\eta\succ0$ and $K_\eta\succ0$, evaluated
at the same fitted parameter with $J_\eta$ fixed. Define
$B_\eta:=F_\eta^{-1/2}C_\eta F_\eta^{-1/2},
~
\rho:=\|B_\eta\|_2,
~
A_\eta:=I+B_\eta$, and let
$\alpha_\eta:=\lambda_{\min}(A_\eta)>0$ and
$\beta_\eta:=\lambda_{\max}(A_\eta)$.
Then the following statements hold.

\begin{enumerate}
\item
For
$\boldsymbol u_p:=
F_\eta^{-1/2}\boldsymbol e_p/
\sqrt{\boldsymbol e_p^\top F_\eta^{-1}\boldsymbol e_p}$,
the exact correction factors satisfy
\begin{equation}
\mathcal R_p^{\mathrm{RH}}
=\|A_\eta^{-1}\boldsymbol u_p\|_2,
~
\mathcal R_{p,\mathrm{curv}}^{\mathrm{RH}}
=\boldsymbol u_p^\top A_\eta^{-1}\boldsymbol u_p,
~
\mathcal R_{p,\mathrm{resp}}^{\mathrm{RH}}
=\frac{\mathcal R_p^{\mathrm{RH}}}
       {\mathcal R_{p,\mathrm{curv}}^{\mathrm{RH}}}.
\label{eq:rh-exact-spectral}
\end{equation}
Consequently,
\begin{equation}
\mathcal R_p^{\mathrm{RH}},
\ \mathcal R_{p,\mathrm{curv}}^{\mathrm{RH}}
\in
\left[\frac1{\beta_\eta},\frac1{\alpha_\eta}\right],
~
1\le\mathcal R_{p,\mathrm{resp}}^{\mathrm{RH}}
\le
\frac{\alpha_\eta+\beta_\eta}
     {2\sqrt{\alpha_\eta\beta_\eta}}.
\label{eq:rh-general-bounds}
\end{equation}

\item
If $\rho<1$, then
\begin{equation}
\mathcal R_p^{\mathrm{RH}},
\ \mathcal R_{p,\mathrm{curv}}^{\mathrm{RH}}
\in
\left[\frac1{1+\rho},\frac1{1-\rho}\right],
~
1\le\mathcal R_{p,\mathrm{resp}}^{\mathrm{RH}}
\le\frac1{\sqrt{1-\rho^2}}.
\label{eq:rh-small-bounds}
\end{equation}

\item
If $\rho\ge1$, then $\beta_\eta=1+\rho$, so
Eq.~\eqref{eq:rh-general-bounds} applies with this substitution.
For $k\ge2$, no finite upper bound on
$\mathcal R_p^{\mathrm{RH}}$ depending only on $\rho$
holds in general.
\end{enumerate}
\end{proposition}

\subsection{Uncertainty quantification}
We adopt a Bayesian posterior approach to quantify parameter uncertainty
and propagate it through the model to assess prediction uncertainty \cite{stuart2010inverse,monsalve2022analysis}.
The parameter posterior is constructed using the residual-corrected
quadratic loss approximation, and all identifiability statements in
this section refer to this quadratic loss approximation. All intervals reported below are Bayesian credible intervals under the stated noise model and prior; frequentist coverage is not addressed.
\subsubsection{Parameter estimation credible interval}
Assume Gaussian observation noise with covariance $\sigma^2I$.
At the fitted parameter $\boldsymbol\eta^*$, the quadratic
likelihood is proportional to
$\exp[-q^{\boldsymbol0}(\Delta\boldsymbol\eta)/\sigma^2]$,
where $q^{\boldsymbol0}$ is defined in
Eq.~\eqref{eq:quadratic-loss-model}.
Write $K_\eta=U_K\Lambda_+U_K^\top,
~
\Lambda_+\succ0,
~
\operatorname{range}(V_K)=\ker(K_\eta)$,
with $[U_K,V_K]$ orthogonal, and decompose
\begin{equation}
\Delta\boldsymbol\eta
=U_K\boldsymbol\phi+V_K\boldsymbol\zeta.
\end{equation}

The quadratic likelihood confines $\boldsymbol\phi$ but is
independent of $\boldsymbol\zeta$.
For the local Gaussian calculation, we neglect set
boundary effects and use a flat prior on $\boldsymbol\phi$
together with an independent proper prior
$\boldsymbol\zeta\sim\mathcal N(\boldsymbol0,\tau_0^2I)$.
This explicitly regularizes the quadratic null space without
adding prior curvature to the positive-curvature subspace.

The resulting quadratic-model posterior is
\begin{equation}
\boldsymbol\eta\mid\widehat{\boldsymbol h}
\overset{\mathrm{quad}}{\sim}
\mathcal N(\boldsymbol\eta^*,\Sigma_\eta),
~
\Sigma_\eta
=
\sigma^2K_\eta^\dagger+\tau_0^2V_KV_K^\top.
\label{eq:uq-quadratic-posterior}
\end{equation}
Thus, positive-curvature directions have posterior covariance
$\sigma^2\Lambda_+^{-1}$, whereas the null-space posterior
coincides with its prior. For each coordinate, the posterior variance is
\begin{equation}
s_p^2
:=
\operatorname{Var}_{\mathrm{quad}}
(\eta_p\mid\widehat{\boldsymbol h})
=
\sigma^2(K_\eta^\dagger)_{pp}
+\tau_0^2\|V_K^\top\boldsymbol e_p\|_2^2,~(K_\eta^\dagger)_{pp} = e_p^\top K_\eta^\dagger e_p.
\label{eq:uq-coordinate-variance}
\end{equation}
Under the conditions of
Theorem~\ref{thm:profile-quadratic-lipschitz},
a coordinate is Lipschitz identifiable precisely when
$\kappa_p^K>0$, equivalently
$V_K^\top\boldsymbol e_p=0$.
For such a coordinate,
$s_p^2={\sigma^2}/{\kappa_p^K}$,
independently of $\tau_0$.
For a non-identifiable coordinate, the second term in
Eq.~\eqref{eq:uq-coordinate-variance} is positive:
its finite posterior uncertainty depends on the imposed
null-space prior and does not imply identifiability from
the quadratic likelihood.

Let $z_{1-\alpha/2}$ denote the standard normal quantile.
The corresponding approximate equal-tailed credible intervals are
\begin{equation}
\begin{aligned}
\mathrm{CrI}_{\eta_p}(1-\alpha)
&=
\left[
\eta_p^*-z_{1-\alpha/2}s_p,\,
\eta_p^*+z_{1-\alpha/2}s_p
\right].
\end{aligned}
\label{eq:uq-credible-intervals}
\end{equation}
For identifiable coordinates, the posterior standard deviation
$\sigma/\sqrt{\kappa_p^K}$ measures uncertainty under the
quadratic posterior, while $L_p^K$ in
Eq.~\eqref{eq:quadratic-coordinate-response}
measures sensitivity to deterministic data perturbations;
these quantities generally differ.

\subsubsection{Model prediction credible interval}

Model prediction credible interval is obtained by propagating the
quadratic parameter posterior in
Eq.~\eqref{eq:uq-quadratic-posterior}
through the deterministic model.
Let
$\boldsymbol a_l(t):=
\nabla_\eta h_l(t;\boldsymbol\eta^*)$.
Under a first-order approximation of the prediction map,
\begin{equation}
h_l(t;\boldsymbol\eta)-h_l(t;\boldsymbol\eta^*)
\approx
\boldsymbol a_l(t)^\top\Delta\boldsymbol\eta.
\label{eq:prediction-linearization}
\end{equation}
The resulting output variance separates into
\begin{equation}
\begin{aligned}
\mathcal V_l(t)
=
\boldsymbol a_l(t)^\top\Sigma_\eta\boldsymbol a_l(t)=
\underbrace{
\sigma^2\boldsymbol a_l(t)^\top
K_\eta^\dagger\boldsymbol a_l(t)
}_{\mathcal V_{l,\mathrm{informed}}(t)}
+
\underbrace{
\tau_0^2\|V_K^\top\boldsymbol a_l(t)\|_2^2
}_{\mathcal V_{l,\mathrm{null}}(t;\tau_0)}.
\end{aligned}
\label{eq:prediction-variance-decomposition}
\end{equation}
The first term propagates uncertainty in the
positive-curvature directions; the second propagates
prior-controlled uncertainty in the quadratic null space.
Accordingly, $\mathcal I_{l,\mathrm{null}}(t)
:=
\|V_K^\top\boldsymbol a_l(t)\|_2$ measures the first-order sensitivity of the prediction
to quadratically non-identifiable directions.
If this index vanishes, those directions do not contribute
to the linearized prediction variance, even though the
full parameter vector may be non-identifiable.

An approximate credible interval for the deterministic output is $h_l(t;\boldsymbol\eta^*)
\ \pm\
z_{1-\alpha/2}\sqrt{\mathcal V_l(t)}$.
For nonlinear prediction maps, credible intervals can instead
be computed by propagating samples from the quadratic parameter
posterior through the full model, sampling both
$\boldsymbol\phi$ and $\boldsymbol\zeta$.
This retains nonlinear output effects while keeping the
quadratic approximation to the parameter posterior.
The resulting intervals describe uncertainty in the
deterministic response; prediction of a future noisy
observation additionally includes its measurement noise.

\section{Results}
\subsection{JAK–STAT signaling pathway model as a benchmark for parameter identifiability analysis}
We evaluated the proposed method using a JAK--STAT signaling pathway model \cite{swameye2003identification} describing receptor-mediated STAT phosphorylation, dimerization, nuclear translocation, and recycling (Fig.~\ref{fig:2}A; see the Appendix section ``JAK--STAT signaling pathway model'' and Table~\ref{tab:jak_stat_parameters}). The model was fitted to time courses of two observables, $h_1=s_1(x_2+2x_3)$ and $h_2=s_2(x_1+x_2+2x_3)$, driven by the measured receptor activity $u(t)$, with seven unknown quantities: the kinetic parameters $p_1$--$p_4$, the scaling factors $s_1$ and $s_2$, and the initial condition $x_1(0)$. The fitted trajectories captured the main temporal features of both observables \cite{raue2009structural} (Fig.~\ref{fig:S1}A).

\begin{figure}[ht]
\centering
\includegraphics[width=1.0\linewidth]{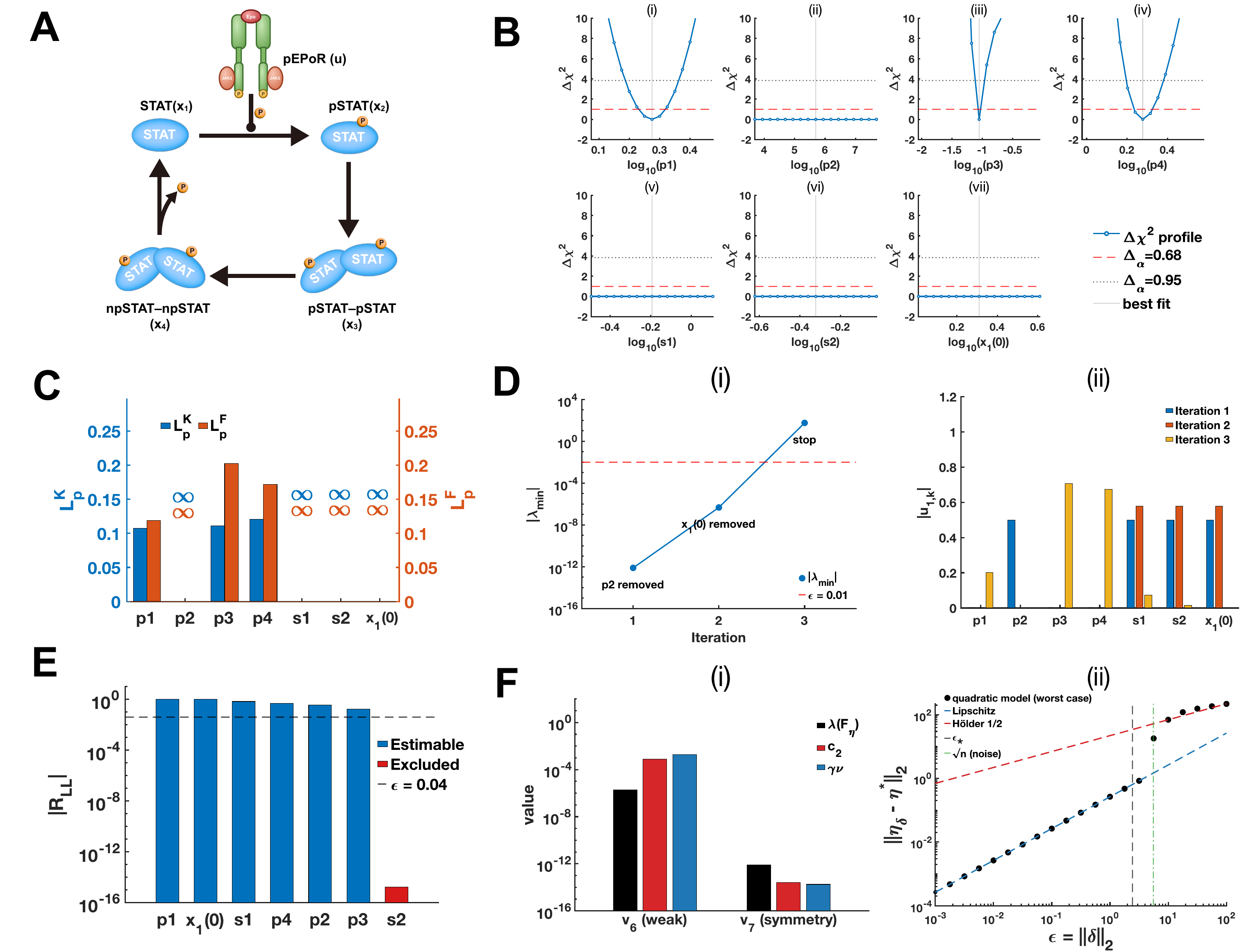}
\caption{\scriptsize\textbf{Comparison of parameter identifiability methods for a JAK--STAT signaling pathway model.}
\textbf{(A)} Pathway schematic: pEpoR-driven STAT phosphorylation, dimerization, nuclear translocation and recycling.
\textbf{(B)} Profile likelihoods of the seven parameters as $\Delta\chi^2$ relative to the fit; dashed and dotted lines mark the 68\% and 95\% thresholds.
\textbf{(C)} Coordinate Lipschitz coefficients $L_p^{K}$ (blue) and $L_p^{F}$ (orange); $\infty$ marks $\kappa_p=0$.
\textbf{(D)} Quaiser method: (i) smallest absolute eigenvalue per iteration against $\epsilon=0.01$, removing $p_2$ and $x_1(0)$; (ii) components of the associated eigenvector.
\textbf{(E)} Yao method: sensitivity scores $|R_{LL}|$ against $\epsilon=0.04$; only $s_2$ is excluded.
\textbf{(F)} Second-order analysis along the weakest eigendirections $v_6$ and $v_7$ of $F_\eta$: (i) first-order curvature $\lambda(F_\eta)$, coercivity $c_2$ and residual curvature $\gamma\nu$; (ii) worst-case displacement versus perturbation size in the quadratic model (second-order Taylor expansion of $\boldsymbol h$), with Lipschitz (blue) and H\"older (red) rates, crossover $\varepsilon_*=\gamma\nu/c_2$ and noise scale $\sqrt{n}$.}
\label{fig:2}
\end{figure}

Profile likelihood \cite{raue2009structural} identified $p_1$, $p_3$ and $p_4$ as practically identifiable, their profiles crossing the confidence thresholds on both sides, whereas the profiles of $p_2$, $s_1$, $s_2$ and $x_1(0)$ were flat (Fig.~\ref{fig:2}B). The proposed criterion gave the same classification from the local Jacobian and Hessian alone, without repeated re-optimization: $p_1$, $p_3$ and $p_4$ had finite $L_p^{K}$, and the other four coordinates had $\kappa_p^{K}=0$ and unbounded $L_p^{K}$ (Fig.~\ref{fig:2}C). For the identifiable coordinates $L_p^{K}<L_p^{F}$, most markedly for $p_3$, so the Fisher approximation overestimated the sensitivity of the estimates to data perturbations.

The Quaiser method \cite{quaiser2009systematic} removed $p_2$ and $x_1(0)$ and classified $\{p_1,p_3,p_4,s_1,s_2\}$ as identifiable (Fig.~\ref{fig:2}D); the Yao method \cite{yao2003modeling}, with threshold $\epsilon=0.04$, excluded only $s_2$ (Fig.~\ref{fig:2}E). Both methods assess the estimability of a subset conditional on fixing or excluding the remaining parameters, so retaining a parameter does not establish its identifiability when all seven quantities are free; this explains why $s_1$ and $s_2$ (Quaiser) and $p_2$, $s_1$ and $x_1(0)$ (Yao) were retained despite flat profiles. The $L_p^{K}$ criterion evaluates each coordinate with all others free and therefore recovered the profile-likelihood classification.

Because $v_7$ is an exact symmetry, we repeated the analysis in the invariant combinations $p_2x_1(0)$, $s_1x_1(0)$ and $s_2x_1(0)$ together with $p_1$, $p_3$ and $p_4$ (Fig.~\ref{fig:S1}). No combination changes along $v_7$, and a unit step along $v_6$ changes only $p_2x_1(0)$ (Fig.~\ref{fig:S1}F), which therefore carries the remaining weak direction. Profile likelihood bounded all combinations except $p_2x_1(0)$, whose profile flattened on the right (Fig.~\ref{fig:S1}B). The coefficients $L_p^{K}$ and $L_p^{F}$ were both largest for $p_2x_1(0)$ and of the same order for the other five, with the Hessian correction lowering $L_p^{K}$ by about three orders of magnitude relative to $L_p^{F}$ for $p_2x_1(0)$ (Fig.~\ref{fig:S1}C).

Because $v_6$ is weakly but not exactly unconstrained, $L_p^{K}$ is finite for $p_2x_1(0)$: the local criterion ranks it as the least stable combination, whereas profile likelihood, which tests boundedness at a fixed confidence level, reports it as non-identifiable. After the symmetry is removed, the Quaiser method also singles out $p_2x_1(0)$ (Fig.~\ref{fig:S1}E), while the Yao method classifies all six combinations as estimable (Fig.~\ref{fig:S1}D).

The second-order analysis separates two mechanisms behind the four non-identifiable coordinates (Fig.~\ref{fig:2}F). Along the two weakest eigen-directions of $F_\eta$ we compared the first-order curvature $\lambda(F_\eta)=v^{\top}F_\eta v$, the second-order coercivity $c_2=\|\Pi^{\perp}H_\eta[v,v]\|_2$, and the residual-induced curvature $\gamma\nu=v^{\top}C_\eta v$, $\nu=\|\boldsymbol r^{*}\|_2$ (Fig.~\ref{fig:2}F(i)). Along $v_7\propto(0,-1,0,0,-1,-1,1)$, in the coordinates $(\ln p_1,\ldots,\ln p_4,\ln s_1,\ln s_2,\ln x_1(0))$, all three vanished to numerical precision. This direction is the scaling symmetry $x\mapsto cx$, $p_2\mapsto p_2/c$, $s_i\mapsto s_i/c$, which leaves both observables unchanged; it is invisible at every order, and no stability rate applies, consistent with the flat profiles of $p_2$, $s_1$, $s_2$ and $x_1(0)$.

Along $v_6$, by contrast, $\lambda(F_\eta)\approx10^{-6}$ while $\gamma\nu$ and $c_2$ were of order $10^{-3}$: the direction is weakly informed by the Fisher information but constrained by the Hessian correction (See details of the numerical experiment at Algorithm~\ref{alg:thm3-jakstat}). The worst-case parameter displacement in the quadratic model, obtained by replacing $\boldsymbol h(\boldsymbol\eta)$ with its second-order Taylor expansion at $\boldsymbol\eta^{*}$, showed the two regimes predicted by Theorem~\ref{thm:crossover} (Fig.~\ref{fig:2}F(ii)): Lipschitz scaling $\propto\varepsilon$ below the crossover $\varepsilon_{*}=\gamma\nu/c_2$, where the residual curvature controls $v_6$, and H\"older scaling $\propto\varepsilon^{1/2}$ above it, where only the second-order response does. The noise level $\sqrt{n}$ lay above $\varepsilon_{*}$, so at the resolution of the data the weak direction is already in the H\"older stability.

\subsection{Patient-specific parameter identifiability in HIV host–virus dynamics}
We next investigated how contrasting HIV infection dynamics \cite{perelson2002modelling,stafford2000modeling} affect parameter identifiability using a host–virus model that incorporates 
\begin{figure}[ht]
\centering
\includegraphics[width=0.9\linewidth]{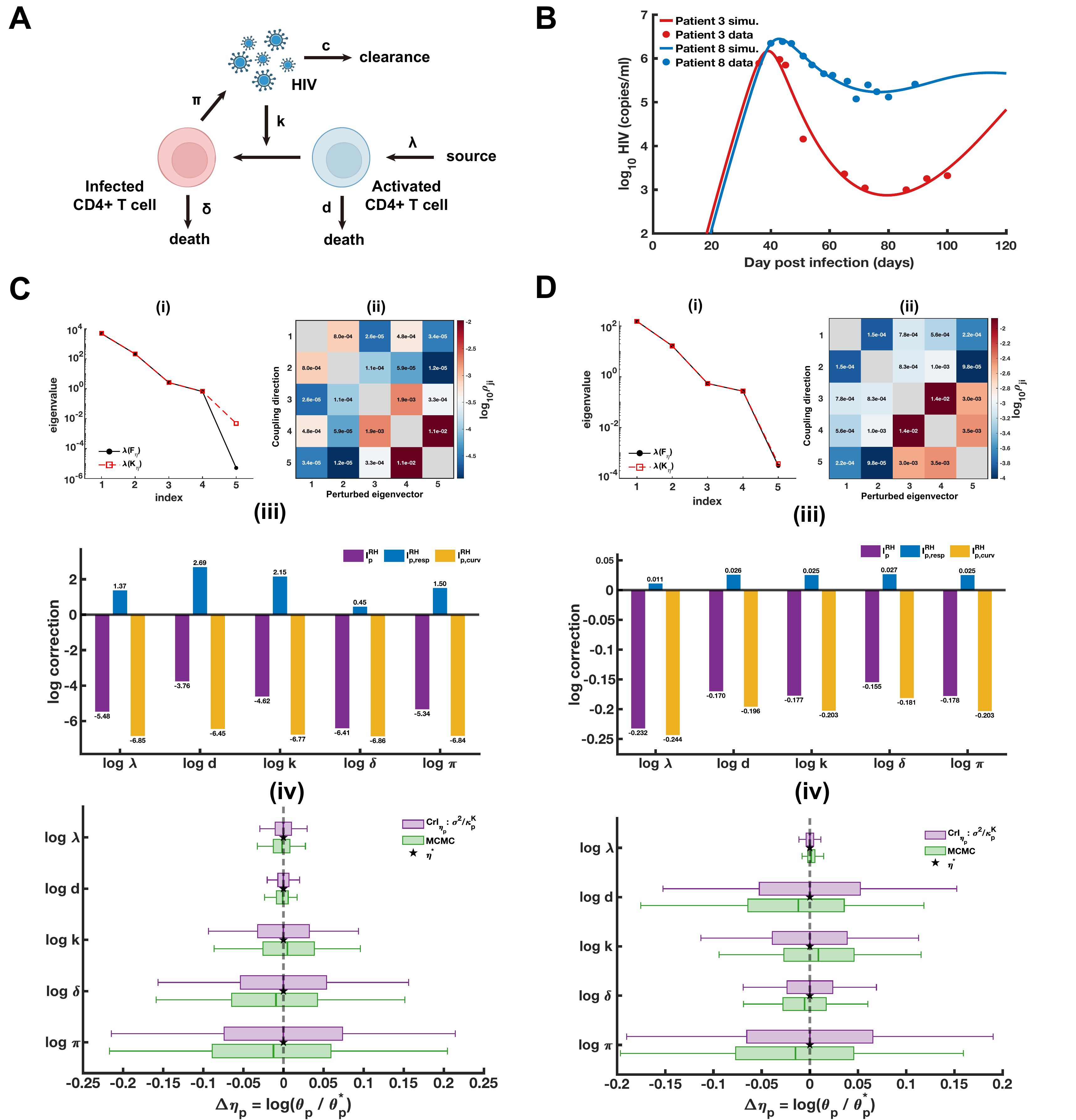}
\caption[Parameter identifiability analysis of HIV host--virus dynamics]{%
\textbf{Parameter identifiability analysis of contrasting HIV host--virus dynamics.}
\textbf{(A)} Schematic of the HIV infection model. Created with BioRender.com.
\textbf{(B)} Model fits to longitudinal plasma HIV measurements from patients 3 (red) and 8 (blue). Circles represent clinical data \protect\cite{stafford2000modeling}, and solid curves represent fitted trajectories of $h(t;\boldsymbol{\eta})=\log_{10}V(t;\boldsymbol{\eta})$. 
\textbf{(C,D)} Identifiability analyses for patients 3 and 8, respectively.
(i) Eigenvalue spectra of $F_{\eta}$ (black circles) and $K_{\eta}$ (red squares).
(ii) Heatmaps of the pairwise eigenvector perturbation measures $\rho_{ij}$. Cell labels show $\rho_{ij}$, and colors indicate $\log_{10}\rho_{ij}$.
(iii) Parameter-wise logarithmic corrections (purple), decomposed into relative-response (blue) and relative-curvature (yellow) contributions. 
(iv) Comparison of uncertainty estimates based on $K_{\eta}$ for patient 3 (purple; variance $\sigma^2/\kappa_p^K$, $\sigma=0.01$) and patient 8 (purple; variance $\sigma^2/\kappa_p^K$,$\sigma=0.0025$) with MCMC results (green). Parameter deviations are expressed as $\Delta\eta_p=\log(\theta_p/\theta_p^*)$. Black stars and the vertical dashed line indicate the fitted reference values.
}
\label{fig:3}
\end{figure}
six fundamental parameters characterizing target-cell turnover, infection, viral production, and clearance (Fig.~\ref{fig:3}A; see Tables~\ref{tab:S1}--\ref{tab:hiv_patient_parameters} for detailed parameter descriptions and section “HIV host-virus dynamics” in the Appendix for further model details). 
The observable was defined as the log-transformed plasma HIV concentration,
$h(t;\boldsymbol{\eta})=\log_{10}V(t;\boldsymbol{\eta}))$, consistent with longitudinal measurements from HIV-infected patients \cite{stafford2000modeling}. 

We considered two patients with contrasting post-peak trajectories (Fig.~\ref{fig:3}B). Patient 3 exhibited a pronounced decline followed by rebound from a low viral load, whereas patient 8 exhibited a shallower decline and rebound at a substantially higher viral load. These trajectories were interpreted here as representing stronger and weaker innate immune responses, respectively. The fitted models captured both temporal patterns, allowing us to examine how these distinct dynamics were reflected in the parameter information described by the residual-corrected Hessian $K_{\eta}$ and the Fisher information matrix $F_{\eta}$.

The comparison revealed a marked difference in the effect of the residual-Hessian correction between the two patients. For patient 3, the four largest eigenvalues of $K_{\eta}$ and $F_{\eta}$ were closely aligned, whereas the smallest eigenvalue increased by nearly three orders of magnitude after correction (Fig.~\ref{fig:3}C(i)). The discrepancy was therefore concentrated in the least constrained parameter direction. Despite this substantial eigenvalue correction, the pairwise eigenvector perturbation measures remained small, reaching a maximum between eigenvectors 4 and 5 ($\rho_{45}\approx1.1\times10^{-2}$; Fig.~\ref{fig:3}C(ii)), indicating limited changes in mode orientation. At the parameter level, variation in the corrections was primarily driven by the relative-response contributions, while the relative-curvature contributions were comparatively uniform (Fig.~\ref{fig:3}C(iii)). 

In contrast, the eigenvalue spectra for patient 8 nearly overlapped, and the eigenvector perturbation measures were also small, with a maximum of approximately $1.4\times10^{-2}$ between eigenvectors 3 and 4 (Fig.~\ref{fig:3}D(i)--(ii)). Both the relative-response and relative-curvature contributions were substantially smaller in magnitude than those for patient 3, yielding modest net logarithmic corrections ranging from approximately $-0.23$ to $-0.16$. Unlike patient 3, variation across parameters was mainly attributable to the relative-curvature contributions, whereas the relative-response contributions were comparatively uniform (Fig.~\ref{fig:3}D(iii)).

These differences in the local information geometry were reflected in the parameter credible intervals. For patient 3, the $K_{\eta}$-based interval widths closely matched those obtained by MCMC (Markov chain Monte Carlo) \cite{brooks1998markov}, with width ratios ranging from 0.985 to 1.024, whereas the $F_{\eta}$-based intervals were approximately 25--31 times wider (Fig.~\ref{fig:3}C(iv); Table~\ref{tab:patient3_credible_intervals}). Consistent with the smaller corrections observed for patient 8, both approximations produced interval widths close to those from MCMC: the $K_{\eta}$-based intervals were approximately 2--11\% narrower (Fig.~3D(iv);Table~\ref{tab:patient8_credible_intervals}), while the $F_{\eta}$-based intervals were approximately 1--8\% wider (Table~\ref{tab:patient8_credible_intervals}). Together, these results show that the residual-Hessian correction was necessary to recover MCMC-consistent interval widths for patient 3, but had a comparatively modest effect for patient 8.

\subsection{Hessian Information shifts the timing of minimum predictive uncertainty in $A\beta$ dynamics} 
To investigate how the residual-Hessian correction modifies the effective null space of the local information matrix and how this changes the width and temporal evolution of prediction credible intervals, we studied the spatiotemporal dynamics of amyloid-$\beta$ (A$\beta$) with neuroimaging data from the ADNI cohort.
\begin{figure}[ht]
\centering
\includegraphics[width=1.0\linewidth]{Fig4.pdf}
\caption{\scriptsize\textbf{Parameter identifiability and predictive uncertainty in amyloid-$\beta$ (A$\beta$) spatiotemporal dynamics.}
\textbf{(A)} Image-processing pipeline: regional A$\beta$-PET data on 68 cortical regions define a brain network and its graph Laplacian.
\textbf{(B)} Observed and model-reconstructed regional A$\beta$ SUVRs at five ages (color bar 0.3--0.6).
\textbf{(C)} Model prediction credible intervals for accumulated A$\beta$ from $K_\eta$ (orange) and $F_\eta$ (blue); black circles are observations and the black curve the fitted trajectory. Interval width is narrowest at 80.73 years for $K_\eta$ and 84.36 years for $F_\eta$ (dashed lines).
\textbf{(D)} (i) Eigenvalues of $F_\eta$ (blue) and $K_\eta$ (orange) in descending order; open squares mark $|\lambda_i(K_\eta)|$ for negative eigenvalues, and the horizontal line is the threshold $\epsilon=10^{-6}$, giving 65 and 90 identifiable eigenvectors. (ii) Predictive variance decomposed into contributions of the identifiable subspace (solid) and the regularized non-identifiable subspace (dotted) for $K_\eta$ and $F_\eta$; the shift of the minimum-width age is driven by the non-identifiable contributions.}
\label{fig:4}
\end{figure}
The brain was partitioned into 68 cortical regions, a graph Laplacian was constructed from the structural connectome (Fig.~\ref{fig:4}A; see the Appendix section ``Spatiotemporal dynamics of A$\beta$ on the brain''), and the connectivity was incorporated into a networked PDE model (Eq.~\ref{eq:S76}). The fitted model reproduced the progressive accumulation and spatial heterogeneity of A$\beta$ at five observation ages, although regional discrepancies remained (Fig.~\ref{fig:4}B; parameter values in Tables~\ref{tab:S3}--\ref{tab:S9}).

The residual-Hessian correction acted mainly on the weakly constrained part of the information spectrum, while the leading eigenvalues of $K_{\eta}$ and $F_{\eta}$ remained closely aligned (Fig.~\ref{fig:4}D(i)). At the threshold $\epsilon=10^{-6}$ the number of identifiable eigenvectors increased from 65 for $F_{\eta}$ to 90 for $K_{\eta}$, reducing the complementary subspace requiring regularization from 71 to 46 dimensions; one eigenvalue of $K_{\eta}$, of order $-10^{-15}$, fell below the threshold in the small-magnitude tail. The correction therefore enlarged the identifiable subspace by 25 directions and changed the structure of the effective null space.

These changes produced distinct temporal patterns in the credible intervals for accumulated A$\beta$ (Fig.~\ref{fig:4}C). The $K_{\eta}$-based interval was narrower over much of the intermediate observation period but began widening earlier than the $F_{\eta}$-based interval. Excluding the shared zero-width initial condition, the interval width reached its minimum at 80.73 years for $K_{\eta}$ and 84.36 years for $F_{\eta}$, a shift of 3.63 years. The variance decomposition localized this difference to the regularized non-identifiable contributions rather than the identifiable ones (Fig.~\ref{fig:4}D(ii)): the identifiable components were small and varied smoothly, whereas the regularized components showed pronounced troughs whose timing tracked the respective minima in interval width.

\section{Discussions}
We have developed a framework that defines parameter identifiability through the well-posedness of parameter estimation under data perturbations, treating the full parameter vector and individual coordinates in a unified way while retaining the full Hessian of the loss rather than its Fisher approximation. The loss Hessian $K_\eta$ is the classical observed information \cite{EfronHinkley1978}; the contribution is not the matrix but the stability theory built on it: the order of stability, Lipschitz or H\"older-$\tfrac12$, and the crossover scale at which the residual and the perturbation exchange control of the weakly informed directions. Using local derivatives at a fitted solution, the framework complements profile likelihood (Definition~\ref{def:like}) \cite{raue2009structural,raue2013joining,wieland2021structural,simpson2026parameter}, which traces the loss along each parameter while re-optimizing the others but requires repeated conditional optimization and returns no stability constants \cite{stapor2018optimization}. Our analysis instead characterizes the stability of the full estimate and of each coordinate after conditional adjustment of the nuisance parameters, under stated assumptions.

The starting point is stability of the inverse problem: a parameter is identifiable when small changes in the observations produce controlled changes in all estimates (Definitions~\ref{def:lip}--\ref{def:coord-general}). Within this formulation, the residual-Hessian correction extends the zero-residual analysis of \cite{wang2025systematic} to fitted solutions with nonzero residuals, and the distinction is consequential. The equivalence between Lipschitz stability and positive definiteness of the Fisher information matrix does not carry over unchanged: the residual contribution to the loss Hessian must be considered, and positive definiteness of the corrected matrix $K_\eta$ decides Lipschitz stability (Theorem~\ref{thm:local_lip}). If $K_\eta$ is singular, its null directions are unconstrained in the quadratic approximation but need not remain unconstrained for the original squared $\ell_2$ loss (Theorem~\ref{thm:local_lip}), which is where second-order information enters.

Theorem~\ref{prop:holder} quantifies this in the zero-residual case: if $F_\eta$ is singular but the second directional response is coercive on $\ker(F_\eta)$, the null component of the estimate is still controlled by the data, with a H\"older bound of order $\|\boldsymbol\delta\|_2^{1/2}$ whose constant is set by $c_2$, while the range component keeps the Lipschitz bound governed by $\sigma^{+}_{\min}(J_\eta)$. Corollary~\ref{cor:fullrank} shows that this degeneracy is a null-space effect: for full-rank $F_\eta$ the H\"older term disappears and the classical Lipschitz rate returns. Theorem~\ref{thm:crossover} joins the two regimes. For a nonzero residual of size $\nu$, the residual curves the null directions by $\gamma\nu$, so that recovery is Lipschitz for perturbations below the crossover $\varepsilon_*=\gamma\nu/c_2$ and H\"older-$\tfrac12$ above it. The order of stability is therefore decided by the ratio $\varepsilon/\varepsilon_*$, and the Lipschitz window closes linearly as the residual tends to zero.

This distinction has a natural geometric interpretation. The model Hessian tensor $H_{\eta}$ is the Bates--Watts acceleration of the solution locus \cite{bates1980relative}: its tangential projection measures parameter-effects curvature and its normal component $\Pi^{\perp}H_{\eta}$ measures intrinsic curvature. At an interior least-squares minimum the residual is normal to the tangent plane, so $\boldsymbol r^{*}\otimes H_{\eta}$ in $K_{\eta}$ is the contraction of the residual with the intrinsic acceleration, while in the zero-residual assumotion the intrinsic acceleration itself controls the null-space component of the estimate (Theorem~\ref{prop:holder}). Theorem~\ref{thm:crossover} ties the two roles to a single object: $c_2$ is the norm of the intrinsic acceleration on $\ker(F_\eta)$ and $\gamma$ its component along the residual direction $\hat{\boldsymbol r}=\boldsymbol r^{*}/\nu$, so that $\gamma\nu$ is the curvature the residual adds to the null directions and the crossover $\varepsilon_*=\gamma\nu/c_2$ is the residual norm scaled by the fraction of intrinsic curvature aligned with the residual. The Lipschitz window is widest when the solution locus curves away from the data along the residual and closes when the two are orthogonal. These quantities, which describe departures from a tangent-plane approximation, are precisely what the Fisher approximation $K_{\eta}\approx F_{\eta}$ discards; the Hessian-corrected Schur complement $\kappa_p^{K}$ (Theorem~\ref{thm:profile-quadratic-lipschitz}) gives the coordinate-wise description.

The JAK--STAT signaling pathway served as a benchmark at both levels of the analysis. For the full parameter vector, it verified Theorem~\ref{thm:crossover}: along the weakest eigendirection of $F_\eta$ the worst-case displacement grew linearly with the perturbation below $\varepsilon_*=\gamma\nu/c_2$ and as its square root above, whereas the exact scaling symmetry, with $c_2=0$, remained invisible at every order. For individual coordinates, the metrics of Theorem~\ref{thm:profile-quadratic-lipschitz} recovered the profile-likelihood classification without repeated profile optimization, whereas the Jacobian-based Yao and Quaiser methods \cite{yao2003modeling,quaiser2009systematic} did not; selecting an estimable subset differs from assessing one parameter with all others free. This agreement supports the coordinate-wise analysis but does not establish equivalence with profile likelihood for general nonlinear models.

To connect coordinate-wise behavior with the full parameter space, we treated the residual-Hessian term as a perturbation of the Fisher information matrix. The spectral tools involved are classical: the Schur complement of Lemma~\ref{thm:residual-schur} decides the rank and null space of $K_{\eta}$, and the eigenvalue perturbation expansion behind the diagnostic $\rho_{ji}$ is standard \cite{Kato1995,StewartSun1990}. Their role here is to compress the change in every coordinate metric into the single scalar $\rho=\|F_\eta^{-1/2}C_\eta F_\eta^{-1/2}\|_2$: when $F_{\eta}$ and $K_{\eta}$ are positive definite, Proposition~\ref{thm:rh-global-coordinate-bounds} bounds each correction factor by functions of $\rho$ alone, so the adequacy of the Fisher approximation can be read from one number. The HIV examples confirm these bounds numerically (Tables~\ref{tab:patient3_theoretical_bounds}--\ref{tab:patient8_theoretical_bounds}) and show their consequences for uncertainty quantification.

For patient~3 the $F_{\eta}$-based credible intervals were 25--31 times wider than the MCMC intervals, whereas the $K_{\eta}$-based widths matched them (Table~\ref{tab:patient3_credible_intervals}); for patient~8 both approximations were close to MCMC (Table~\ref{tab:patient8_credible_intervals}). The correction acted mainly on the smallest eigenvalue, with little change in eigenvector orientation (Fig.~\ref{fig:3}C(i)--(ii)), so coordinate-wise uncertainty can change substantially even when the eigendirections do not. The diagnostic $\epsilon_U$ captured the contrast: the Fisher approximation was adequate for $\epsilon_U<1$ and overestimated interval widths for $\epsilon_U\gg1$. Residual curvature must therefore be assessed against the local information structure, not eigenvector changes alone.

High-dimensional prediction raises a further issue: the distinction between an exact null space and a subspace defined by spectral truncation. Sloppiness describes a hierarchy of sensitivities and is not by itself evidence of practical non-identifiability \cite{gutenkunst2007universally,white2016limitations}; a clear spectral gap makes a threshold-based partition more defensible. Such a gap was present in the spectrum of $K_{\eta}$ for the A$\beta$ model (Fig.~\ref{fig:4}C(i)), supporting the subthreshold eigenspace as an effective null subspace. Differences in its regularized contributions shifted the age of the narrowest prediction interval from 84.36 to 80.73 years (Fig.~\ref{fig:4}E), into the period of rapid A$\beta$ accumulation.

In summary, these results establish a residual-aware framework for assessing practical identifiability through stability under data perturbations. Within a local quadratic approximation, the framework connects spectral changes to coordinate-wise identifiability and uncertainty in both parameter estimates and model predictions. The JAK--STAT, HIV, and $A\beta$ applications illustrate its ability to reproduce profile-likelihood classifications, recover MCMC-consistent interval widths when Fisher-based approximations fail, and reveal changes in the temporal pattern of predictive uncertainty. By combining local Jacobian and Hessian information, the approach offers a computationally tractable basis for analyzing high-dimensional biological models while making explicit where higher-order analysis remains necessary.

\section{Acknowledgments}
\paragraph{Funding.}
This research supported by National Institute of General Medical Sciences through grant 1R35GM146894. This research was also supported by the Huck Chair in AI Mathematical Modeling from Penn State University's Huck Institutes of the Life Sciences.
\paragraph{Author contributions.}
W.H. designed research; S.W. performed research; S.W. and W.H. contributed new reagents/analytic tools; S.W. analyzed data; and S.W. and W.H. wrote the paper.
\paragraph{Competing interests.}
There are no competing interests to declare.
\subsection*{Data and code availability:} 
\paragraph{JACK-STAT signaling pathway dataset:}The experimental data used to calibrate the JAK–STAT signaling pathway model were obtained from public dataset \cite{swameye2003identification,raue2009structural}.
\paragraph{HIV concentration dataset:} The data used in this study were obtained from the publicly available dataset \cite{stafford2000modeling}. Each data point consists of two values: the first indicates the relative time in days, and the second represents the viral concentration in thousands of HIV-1 RNA copies per milliliter ($\text{ml}^{-1}$).
\paragraph{ADNI dataset:} Data used in the preparation of this article were obtained from the ADNI database (\url{http://adni.loni.usc.edu/}). ADNI was launched in 2003 as a public-private partnership, led by Principal Investigator Michael W. Weiner, MD. Its primary goal is to determine whether serial magnetic resonance imaging (MRI), positron emission tomography (PET), other biological markers, and clinical and neuropsychological assessments can be combined to track the progression of mild cognitive impairment (MCI) and early Alzheimer's disease (AD). Based on our group preivous work \cite{wang2025learning} on the data preprocessing, we accessed multimodal neuroimaging and biomarker data from the ADNI database under an approved data use application. $A\beta$-PET standardized uptake value ratios (SUVRs) was obtained from the datasets titled \textit{“UC Berkeley – Amyloid PET 6mm Res analysis [ADNI1, GO, 2, 3, 4]”}. Regional SUVRs were computed by dividing the standardized uptake values (SUVs) in target regions by the SUV of the whole cerebellum, chosen as the reference region due to its minimal nonspecific binding and stable tracer uptake across subjects. 


\bibliography{sample} 
\bibliographystyle{ws-m3as}

\newpage
\section{Appendices}
\subsection*{Supplementay Proofs}
\paragraph{Theorem 1 }
\begin{myproof}
For $\|\boldsymbol d\|_2\le R$ set
$\boldsymbol\Phi(\boldsymbol d)
:=\boldsymbol h(\boldsymbol t,\boldsymbol\eta^*+\boldsymbol d)
-\boldsymbol h(\boldsymbol t,\boldsymbol\eta^*)$.
Taylor's formula with integral remainder gives
\begin{equation}
\label{eq:pf-taylor}
\boldsymbol\Phi(\boldsymbol d)
=J_\eta \boldsymbol d
+\tfrac12\boldsymbol H_\eta[\boldsymbol d,\boldsymbol d]
+\boldsymbol r(\boldsymbol d),
~
\|\boldsymbol r(\boldsymbol d)\|_2
\le\omega(\|\boldsymbol d\|_2)\,\|\boldsymbol d\|_2^2 ,
\end{equation}
where
$\omega(s):=\tfrac12\sup_{\|\boldsymbol e\|_2\le s}
\bigl\|\nabla_\eta^2\boldsymbol h(\boldsymbol t,\boldsymbol\eta^*+\boldsymbol e)
-\nabla_\eta^2\boldsymbol h(\boldsymbol t,\boldsymbol\eta^*)\bigr\|$
is nondecreasing with $\omega(s)\to0$ as $s\to0$;
the same formula shows
$M_\Pi\le\tfrac12\sup_{U_\eta}\|\nabla_\eta^2\boldsymbol h\|<\infty$.
Since $\boldsymbol x^\top F\boldsymbol x=\|J\boldsymbol x\|_2^2$,
$\ker F=\ker J$, and the projectors satisfy
\begin{equation}
\label{eq:pf-proj}
\Pi J_\eta=J_\eta,~
\Pi^\perp J_\eta=\boldsymbol0,~
J_\eta P_{\mathcal N}=\boldsymbol0,~
J_\eta \Delta\boldsymbol\eta=J\Delta\boldsymbol\eta_{\mathcal R},~
\|J_\eta \boldsymbol x\|_2\ge\sigma_+\|\boldsymbol x\|_2
\ \ (\boldsymbol x\in\mathcal R_\eta).
\end{equation}
Polarization,
$\boldsymbol H_\eta[\boldsymbol x,\boldsymbol y]
=\tfrac14\bigl(\boldsymbol H_\eta[\boldsymbol x+\boldsymbol y,\boldsymbol x+\boldsymbol y]
-\boldsymbol H_\eta[\boldsymbol x-\boldsymbol y,\boldsymbol x-\boldsymbol y]\bigr)$,
gives
$\|\Pi^\perp\boldsymbol H_\eta[\boldsymbol x,\boldsymbol y]\|_2
\le M_\perp\|\boldsymbol x\|_2\|\boldsymbol y\|_2$
for all $\boldsymbol x,\boldsymbol y$, and bilinearity together with
Eq.~\eqref{eq:c2-def} gives
$\|\Pi^\perp\boldsymbol H_\eta[\boldsymbol\beta,\boldsymbol\beta]\|_2
\ge c_2\|\boldsymbol\beta\|_2^2$ for all $\boldsymbol\beta\in\mathcal N_\eta$. Define
$a:=\|\Delta\boldsymbol\eta_{\mathcal R}\|_2$,
$b:=\|\Delta\boldsymbol\eta_{\mathcal N}\|_2$,
$s:=\|\Delta\boldsymbol\eta\|_2=(a^2+b^2)^{1/2}$, and
$K_0:=4M_\perp^2/c_2+M_\perp+c_2/8$.
Fix a radius $\varrho\in(0,R]$ such that
\begin{equation}
\label{eq:eta-thresholds}
\text{(T1)}\ \ M_\Pi\varrho\le\tfrac14\sigma_+,
~
\text{(T2)}\ \ 2\omega(\varrho)\le\tfrac18c_2,
~
\text{(T3)}\ \ \frac{32K_0M_\Pi^2}{9\sigma_+^2}\,\varrho^2\le\tfrac18c_2 ,
\end{equation}
which is possible because all three left-hand sides vanish as
$\varrho\to0$.

Because $\widehat{\boldsymbol h}=\boldsymbol h(\boldsymbol t,\boldsymbol\eta^*)$,
Eq.~\eqref{eq:perturbed-loss-notation} reads
$l_{\boldsymbol\delta}(\boldsymbol\eta^*+\boldsymbol d)
=\tfrac12\|\boldsymbol\Phi(\boldsymbol d)+\boldsymbol\delta\|_2^2$,
and in particular
$l_{\boldsymbol\delta}(\boldsymbol\eta^*)=\tfrac12\varepsilon^2$.
Minimality of $\boldsymbol\eta_{\boldsymbol\delta}$ over
$U_\eta$ gives
$\|\boldsymbol\Phi(\Delta\boldsymbol\eta)+\boldsymbol\delta\|_2\le\varepsilon$,
hence
\begin{equation}
\label{eq:pf-rho}
\|\boldsymbol\Phi(\Delta\boldsymbol\eta)\|_2\le2\varepsilon .
\end{equation}

By Lemma~\ref{lem:localization}, $s\to0$ as $\varepsilon\to0$.
Choose $\varepsilon_0>0$ such that $\varepsilon<\varepsilon_0$ implies
$s\le\varrho$ for every perturbed minimizer. Applying $\Pi$ to $\boldsymbol\Phi(\Delta\boldsymbol\eta)$ and using
Eq.~\eqref{eq:pf-proj}, the definition of $M_\Pi$ and
Eq.~\eqref{eq:pf-rho},
\[
\sigma_+a
\le\|J_\eta \Delta\boldsymbol\eta_{\mathcal R}\|_2
=\|J_\eta \Delta\boldsymbol\eta\|_2
\le\|\Pi\boldsymbol\Phi(\Delta\boldsymbol\eta)\|_2+M_\Pi s^2
\le2\varepsilon+M_\Pi\,(a^2+b^2).
\]
By (T1), $M_\Pi a^2\le M_\Pi\varrho\,a\le\tfrac14\sigma_+a$, so that
\begin{equation}
\label{eq:pf-R}
a\le\frac{8}{3\sigma_+}\,\varepsilon+\frac{4M_\Pi}{3\sigma_+}\,b^2 .
\end{equation}

Applying $\Pi^\perp$ to Eq.~\eqref{eq:pf-taylor}, the first-order
term drops out by $\Pi^\perp J_\eta=\boldsymbol0$, so
$\Pi^\perp\boldsymbol\Phi(\Delta\boldsymbol\eta)
=\tfrac12\Pi^\perp\boldsymbol H_\eta[\Delta\boldsymbol\eta,\Delta\boldsymbol\eta]
+\Pi^\perp\boldsymbol r(\Delta\boldsymbol\eta)$,
and Eq.~\eqref{eq:pf-rho} yields
\[
\bigl\|\Pi^\perp\boldsymbol H_\eta[\Delta\boldsymbol\eta,\Delta\boldsymbol\eta]\bigr\|_2
\le4\varepsilon+2\omega(s)\,s^2 .
\]
Expanding
$\boldsymbol H_\eta[\Delta\boldsymbol\eta,\Delta\boldsymbol\eta]
=\boldsymbol H_\eta[\Delta\boldsymbol\eta_{\mathcal N},\Delta\boldsymbol\eta_{\mathcal N}]
+2\boldsymbol H_\eta[\Delta\boldsymbol\eta_{\mathcal R},\Delta\boldsymbol\eta_{\mathcal N}]
+\boldsymbol H_\eta[\Delta\boldsymbol\eta_{\mathcal R},\Delta\boldsymbol\eta_{\mathcal R}]$
and using the coercivity on $\mathcal N_\eta$ together with the
polarization bound,
\begin{equation}
\label{eq:pf-N}
c_2b^2
\le\bigl\|\Pi^\perp\boldsymbol H_\eta[\Delta\boldsymbol\eta_{\mathcal N},
\Delta\boldsymbol\eta_{\mathcal N}]\bigr\|_2
\le4\varepsilon+2\omega(s)(a^2+b^2)+2M_\perp ab+M_\perp a^2 .
\end{equation}
By (T2) and $s\le\varrho$, the remainder term satisfies
$2\omega(s)(a^2+b^2)\le\tfrac18c_2\,(a^2+b^2)$, and Young's inequality
splits the cross term as
$2M_\perp ab\le\tfrac14c_2\,b^2+\tfrac{4M_\perp^2}{c_2}\,a^2$.
The $b^2$ terms on the right-hand side of Eq.~\eqref{eq:pf-N} then total
$\tfrac38c_2\,b^2$; moving them to the left-hand side gives
\begin{equation}
\label{eq:pf-Nstar}
\tfrac58c_2\,b^2\le4\varepsilon+K_0\,a^2 ,
~
K_0=\frac{4M_\perp^2}{c_2}+M_\perp+\frac{c_2}{8}.
\end{equation}
Squaring Eq.~\eqref{eq:pf-R} with $(x+y)^2\le2x^2+2y^2$ and using
$b^4\le\varrho^2b^2$,
\[
a^2\le\frac{128}{9\sigma_+^2}\,\varepsilon^2
+\frac{32M_\Pi^2}{9\sigma_+^2}\,\varrho^2\,b^2 .
\]
Inserting this into Eq.~\eqref{eq:pf-Nstar} and using (T3),
\[
\tfrac58c_2\,b^2\le4\varepsilon+C_1\varepsilon^2+\tfrac18c_2\,b^2,
~
C_1:=\frac{128K_0}{9\sigma_+^2},
\]
that is, $\tfrac12c_2\,b^2\le4\varepsilon+C_1\varepsilon^2$.
Reducing $\varepsilon_0$ so that $C_1\varepsilon_0\le4$ gives
$\tfrac12c_2\,b^2\le8\varepsilon$, which is the second bound in
Eq.~\eqref{eq:holder-bound}:
\begin{equation}
\label{eq:pf-null-final}
b\le\frac{4}{\sqrt{c_2}}\,\varepsilon^{1/2}.
\end{equation}
Substituting $b^2\le16\varepsilon/c_2$ into Eq.~\eqref{eq:pf-R} gives
the first bound,
\begin{equation}
\label{eq:pf-range-final}
a\le\frac{8}{3\sigma_+}\,\varepsilon
+\frac{4M_\Pi}{3\sigma_+}\cdot\frac{16}{c_2}\,\varepsilon
=\frac{8}{3\sigma_+}\Bigl(1+\frac{8M_\Pi}{c_2}\Bigr)\varepsilon .
\end{equation}
The thresholds
in Eq.~\eqref{eq:eta-thresholds} and the condition
$C_1\varepsilon_0\le4$ were used; all of these are closed-form
expressions in $R$, $\sigma_+$, $M_\Pi$, $M_\perp$, $c_2$ and $\omega$,
so no further ``sufficiently small'' constant is left implicit.

By Eq.~\eqref{eq:pf-range-final}, $a=O(\varepsilon)$; hence
$ab=O(\varepsilon^{3/2})$, $a^2=O(\varepsilon^2)$, and
$\omega(s)s^2=o(\varepsilon)$ because $s^2=O(\varepsilon)$ and
$\omega(s)\to0$.
Returning to Eq.~\eqref{eq:pf-N} without the Young split,
$c_2b^2\le4\varepsilon\,(1+o(1))$, so
$b\le2c_2^{-1/2}\varepsilon^{1/2}(1+o(1))$, and
$s\le a+b=b+O(\varepsilon)$ gives Eq.~\eqref{eq:holder-total};
the factor $2$ between Eq.~\eqref{eq:pf-null-final} and
Eq.~\eqref{eq:holder-total} is the price of a bound that is uniform
for all $\varepsilon<\varepsilon_0$ rather than asymptotic.
Finally, Eq.~\eqref{eq:pf-N} involves
$\Delta\boldsymbol\eta_{\mathcal N}$ only through
$\boldsymbol H_\eta[\Delta\boldsymbol\eta_{\mathcal N},\Delta\boldsymbol\eta_{\mathcal N}]$,
which is invariant under
$\Delta\boldsymbol\eta_{\mathcal N}\mapsto-\Delta\boldsymbol\eta_{\mathcal N}$;
the second-order constraints therefore fix the null component only up
to sign.
\end{myproof}
\paragraph{Theorem 2 }
\begin{myproof}
Let
$G(\boldsymbol\eta,\boldsymbol\delta)
:=\nabla_\eta l_{\boldsymbol\delta}(\boldsymbol\eta)
=J(\boldsymbol\eta)^\top
\bigl(\boldsymbol h(\boldsymbol t,\boldsymbol\eta)
-\widehat{\boldsymbol h}+\boldsymbol\delta\bigr)$,
where $J(\boldsymbol\eta):=\nabla_\eta\boldsymbol h(\boldsymbol t,\boldsymbol\eta)$.
The regularity in Assumption~\ref{ass:admissible} makes $G$
continuously differentiable near $(\boldsymbol\eta^*,\boldsymbol0)$,
with
\[
G(\boldsymbol\eta^*,\boldsymbol0)=\boldsymbol0,
~
\nabla_\eta G(\boldsymbol\eta^*,\boldsymbol0)
=\nabla_\eta^2 l_{\boldsymbol0}(\boldsymbol\eta^*)=K_\eta,
~
\nabla_{\boldsymbol\delta}G(\boldsymbol\eta^*,\boldsymbol0)=J_\eta^\top .
\]
By Lemma~\ref{lem:localization}, for $\|\boldsymbol\delta\|_2$ small
every $\boldsymbol\eta_{\boldsymbol\delta}\in\mathcal M^{\boldsymbol\delta}$
is interior and therefore stationary,
$G(\boldsymbol\eta_{\boldsymbol\delta},\boldsymbol\delta)=\boldsymbol0$.

\textit{(1) Necessity.}
Since $G$ is $C^1$ and
$\boldsymbol d_{\boldsymbol\delta}\to\boldsymbol0$, a first-order
expansion of the stationarity condition around
$(\boldsymbol\eta^*,\boldsymbol0)$ gives
\begin{equation}
\label{eq:nonlinear-stationarity-expansion}
\boldsymbol0
=K_\eta\boldsymbol d_{\boldsymbol\delta}
+J_\eta^\top\boldsymbol\delta
+o\bigl(\|\boldsymbol d_{\boldsymbol\delta}\|_2
+\|\boldsymbol\delta\|_2\bigr).
\end{equation}
Under the assumed bound
$\|\boldsymbol d_{\boldsymbol\delta}\|_2\le C\|\boldsymbol\delta\|_2$
the remainder is $o(\|\boldsymbol\delta\|_2)$.
Fix $\boldsymbol z\in\ker(K_\eta)$ and $\boldsymbol w\in\mathbb R^{NL}$,
and take $\boldsymbol\delta=t\boldsymbol w$ with $t\to0^+$.
Multiplying Eq.~\eqref{eq:nonlinear-stationarity-expansion} by
$\boldsymbol z^\top$ eliminates the first term and leaves
$0=t\,(J_\eta\boldsymbol z)^\top\boldsymbol w+o(t)$;
dividing by $t$ and letting $t\to0^+$ gives
$(J_\eta\boldsymbol z)^\top\boldsymbol w=0$.
Since $\boldsymbol w$ is arbitrary, $J_\eta\boldsymbol z=\boldsymbol0$,
which proves Eq.~\eqref{eq:K-J-necessary}.

\textit{(2) Sufficiency.}
Let $m:=\lambda_{\min}(K_\eta)>0$.
The Hessian
\[
\nabla_\eta^2 l_{\boldsymbol\delta}(\boldsymbol\eta)
=\nabla_\eta^2 l_{\boldsymbol0}(\boldsymbol\eta)
+\sum_{n=1}^{N}\sum_{l=1}^{L}\delta_{n,l}\,
\nabla_\eta^2 h_l(t_n,\boldsymbol\eta)
\]
depends continuously on $(\boldsymbol\eta,\boldsymbol\delta)$ and
equals $K_\eta$ at $(\boldsymbol\eta^*,\boldsymbol0)$.
Hence there exist a convex neighborhood
$\mathcal B\subseteq\operatorname{int}(U_\eta)$ of $\boldsymbol\eta^*$
and $\varepsilon_0>0$ such that
$\nabla_\eta^2 l_{\boldsymbol\delta}(\boldsymbol\eta)\succeq\frac m2 I$
for all $\boldsymbol\eta\in\mathcal B$ and
$\|\boldsymbol\delta\|_2\le\varepsilon_0$;
by Lemma~\ref{lem:localization}, after reducing $\varepsilon_0$ every
perturbed minimizer lies in $\mathcal B$.
Stationarity at $\boldsymbol\eta_{\boldsymbol\delta}$ and at
$\boldsymbol\eta^*$ gives
$G(\boldsymbol\eta_{\boldsymbol\delta},\boldsymbol\delta)=\boldsymbol0$
and
$G(\boldsymbol\eta^*,\boldsymbol\delta)
=G(\boldsymbol\eta^*,\boldsymbol0)+J_\eta^\top\boldsymbol\delta
=J_\eta^\top\boldsymbol\delta$.
The fundamental theorem of calculus along the segment joining the two
points therefore yields
\[
-J_\eta^\top\boldsymbol\delta
=
\left[
\int_0^1
\nabla_\eta^2 l_{\boldsymbol\delta}
\bigl(\boldsymbol\eta^*+s\boldsymbol d_{\boldsymbol\delta}\bigr)\,ds
\right]
\boldsymbol d_{\boldsymbol\delta}.
\]
Taking the inner product with $\boldsymbol d_{\boldsymbol\delta}$ and
using the uniform lower bound on the Hessian,
\[
\frac m2\|\boldsymbol d_{\boldsymbol\delta}\|_2^2
\le
-\boldsymbol d_{\boldsymbol\delta}^\top J_\eta^\top\boldsymbol\delta
\le
\|\boldsymbol d_{\boldsymbol\delta}\|_2\,
\|J_\eta\|_2\,\|\boldsymbol\delta\|_2 ,
\]
which is Eq.~\eqref{eq:K-lipschitz-bound}
(trivially also when $\boldsymbol d_{\boldsymbol\delta}=\boldsymbol0$).
The conditional stability estimate of Definition~\ref{def:lip}
holds with the index function $\omega(\varepsilon)=\varepsilon$.

\textit{(3) Equivalence.}
If $J_\eta$ has full column rank, then $\ker(J_\eta)=\{\boldsymbol0\}$,
and part~(1) shows that Lipschitz stability forces
$\ker(K_\eta)=\{\boldsymbol0\}$.
Since $K_\eta\succeq0$ at the interior minimum $\boldsymbol\eta^*$,
this is equivalent to $K_\eta\succ0$.
The converse is part~(2).
\end{myproof}
\paragraph{Theorem 3 }

\begin{myproof} Write $d:=\eta^\delta-\eta^*$, $x:=P_{\mathcal R}d$,
$\xi:=P_{\mathcal N}d$, $a:=\|x\|_2$, $b:=\|\xi\|_2$, $s:=\|d\|_2$, so that
$s^2=a^2+b^2$, and $\Phi(d):=h(t,\eta^*+d)-h(t,\eta^*)$. By
Eq.~\eqref{eq:pf-taylor},
\begin{equation}
\label{eq:taylor2}
\Phi(d)=J_\eta d+\tfrac12H_\eta[d,d]+\mathrm r(d),
~
\|\mathrm r(d)\|_2\le\omega(s)\,s^2 .
\end{equation}
 
Let
\begin{equation}
\label{eq:consts}
\begin{gathered}
A_0:=\frac{16}{5\sigma_+},\quad A_1:=\frac{32M_\perp}{5\sigma_+^2},~
A_2:=\frac{8M_\Pi}{5\sigma_+},~
K:=M_\perp+\frac{4M_\perp^2}{\gamma_0}+\frac{c_2}4,\\
G:=\sqrt{2K/\gamma_0},~
K_1:=\frac{8M_\perp^2}{c_2}+M_\perp+\frac{c_2}8,~
K_1':=K_1+\frac{8c_2M_\perp^2}{\gamma_0^2}+\frac{c_2(M_\perp+c_2/8)}{\gamma_0}.
\end{gathered}
\end{equation}
Fix $\varrho\in(0,R]$ such that
\begin{equation}
\label{eq:rho-cond}
\text{(R1) } M_\Pi\varrho\le\tfrac18\sigma_+,\quad
\text{(R2) } \omega(\varrho)\le\tfrac1{16}\gamma_0,\quad
\text{(R3) } GA_2\varrho\le\tfrac14,\quad
\text{(R4) } 2K_1'A_2^2\varrho^2\le\tfrac18c_2 ,
\end{equation}
and let $\nu_1>0$ be such that every $\nu\le\nu_1$ satisfies
\begin{equation}
\label{eq:nu-cond}
\text{(N1) } \nu\bigl(M_\perp+\tfrac14c_2\bigr)\le\tfrac1{64}\sigma_+^2,\quad
\text{(N2) } GA_1\nu\le\tfrac14,\quad
\text{(N3) } \nu\le1 .
\end{equation}
By Lemma~\ref{lem:localization}, choose $\varepsilon_0\in(0,1]$ such that $\varepsilon<\varepsilon_0$
implies $s\le\varrho$ for every perturbed minimizer. Then $a,b\le\varrho$ and,
by (R2) and $\gamma_0\le\gamma\le c_2$,
\begin{equation}
\label{eq:omega-small}
2\omega(s)\le\tfrac18\gamma_0\le\tfrac18\gamma\le\tfrac18c_2 .
\end{equation}
Conditions (N3) and $\varepsilon_0\le1$ fix the units of the data so that
$\varepsilon^2\le\varepsilon$ and $\varepsilon_*\le1$; they only affect the
constants.
 
Since
$h(t,\eta^*+d)-(\widehat h-\delta)=\Phi(d)+r^*+\delta$, comparing $l^\delta$ at
$\eta^\delta$ and at $\eta^*$ gives
$\|\Phi\|_2^2+2\langle\Phi,r^*\rangle+2\langle\Phi,\delta\rangle\le0$. By
Eq.~\eqref{eq:taylor2}, $J_\eta^\top r^*=0$ and
$\langle H_\eta[d,d],r^*\rangle=d^\top C_\eta d$,
$\langle\Phi,r^*\rangle\ge\tfrac12d^\top C_\eta d-\nu\,\omega(s)s^2$, hence
\begin{equation}
\label{eq:P0}
\|\Phi(d)\|_2^2+d^\top C_\eta d\le2\varepsilon\|\Phi(d)\|_2+2\nu\,\omega(s)\,s^2 .
\end{equation}
For $\nu=0$ this is Eq.~\eqref{eq:pf-rho} of Theorem~\ref{prop:holder}.
 
Split
$d^\top C_\eta d=\xi^\top C_\eta\xi+2x^\top C_\eta\xi+x^\top C_\eta x$. By
Lemma~\ref{lem:residual-structure}~(a),(b),
\begin{equation}
\label{eq:dCd}
d^\top C_\eta d\ge\gamma\nu\,b^2-2\nu M_\perp ab-\nu M_\perp a^2 .
\end{equation}
First, $2\varepsilon\|\Phi\|_2-\|\Phi\|_2^2\le\varepsilon^2$, so
Eqs.~\eqref{eq:P0}--\eqref{eq:dCd} give
\begin{equation}
\label{eq:P1}
\gamma\nu\,b^2\le\varepsilon^2+2\nu M_\perp ab+\nu M_\perp a^2+2\nu\,\omega(s)(a^2+b^2).
\end{equation}
Second, by Eq.~\eqref{eq:omega-small}, $(2\omega(s)-\gamma)\nu b^2\le0$, so
Eqs.~\eqref{eq:P0}--\eqref{eq:dCd} give $\|\Phi\|_2^2-2\varepsilon\|\Phi\|_2\le\mathcal E$
with
\begin{equation}
\label{eq:P2}
\mathcal E:=2\nu M_\perp ab+\nu\bigl(M_\perp+2\omega(s)\bigr)a^2,
~
\|\Phi(d)\|_2\le\varepsilon+\sqrt{\varepsilon^2+\mathcal E}\le2\varepsilon+\sqrt{\mathcal E}.
\end{equation}
 
As in Eq.~\eqref{eq:pf-R},
$\sigma_+a\le\|J_\eta d\|_2\le\|\Phi(d)\|_2+M_\Pi(a^2+b^2)$. By Eq.~\eqref{eq:P2},
$\sqrt{\mathcal E}\le\sqrt{2\nu M_\perp ab}+\sqrt{\nu(M_\perp+2\omega(s))}\,a$.
By (N1) and Eq.~\eqref{eq:omega-small}, $\sqrt{\nu(M_\perp+2\omega(s))}\le\sigma_+/8$;
by (R1), $M_\Pi a^2\le M_\Pi\varrho\,a\le\sigma_+a/8$; and by Young's inequality,
$\sqrt{2\nu M_\perp ab}\le\sigma_+a/8+4\nu M_\perp b/\sigma_+$. Hence
$\tfrac58\sigma_+a\le2\varepsilon+4\nu M_\perp b/\sigma_++M_\Pi b^2$, that is,
\begin{equation}
\label{eq:range}
a\le A_0\varepsilon+A_1\nu b+A_2b^2 .
\end{equation}
 
 In Eq.~\eqref{eq:P1}, use
$2\nu\omega(s)b^2\le\tfrac18\gamma\nu b^2$ (Eq.~\eqref{eq:omega-small}),
$2\nu M_\perp ab\le\tfrac14\gamma\nu b^2+(4\nu M_\perp^2/\gamma)a^2$, and
$2\omega(s)\le c_2/4$, $\gamma\ge\gamma_0$:
$\tfrac58\gamma\nu b^2\le\varepsilon^2+\nu Ka^2$. Hence
$b^2\le2\varepsilon^2/(\gamma\nu)+(2K/\gamma_0)a^2$ and
$b\le\sqrt2\,\varepsilon/\sqrt{\gamma\nu}+Ga$. Inserting Eq.~\eqref{eq:range}
and using (N2), (R3) and $b\le\varrho$, so that $GA_1\nu b+GA_2b^2\le\tfrac12b$,
\begin{equation}
\label{eq:L-bound}
b\le\frac{2\sqrt2\,\varepsilon}{\sqrt{\gamma\nu}}+2GA_0\,\varepsilon .
\end{equation}
 
By Eq.~\eqref{eq:taylor2}, $\Pi^\perp\Phi(d)=\tfrac12\Pi^\perp H_\eta[d,d]+\Pi^\perp\mathrm r(d)$,
so $\|\Pi^\perp H_\eta[d,d]\|_2\le2\|\Phi\|_2+2\omega(s)s^2$. Coercivity
\eqref{eq:c2-def}, the expansion
$\Pi^\perp H_\eta[\xi,\xi]=\Pi^\perp H_\eta[d,d]-2\Pi^\perp H_\eta[x,\xi]-\Pi^\perp H_\eta[x,x]$
and Lemma~\ref{lem:opbound} give
$c_2b^2\le\|\Pi^\perp H_\eta[d,d]\|_2+2M_\perp ab+M_\perp a^2$. With
Eq.~\eqref{eq:P2}, $2\omega(s)\le c_2/8$ and
$2M_\perp ab\le\tfrac18c_2b^2+(8M_\perp^2/c_2)a^2$,
\begin{equation}
\label{eq:H1}
\tfrac34c_2b^2\le4\varepsilon+2\sqrt{\mathcal E}+K_1a^2 .
\end{equation}
If $\nu=0$, then $\mathcal E=0$. If $\nu>0$, then
$\varepsilon\ge\varepsilon_*=\gamma\nu/c_2$ gives $\nu/\varepsilon\le c_2/\gamma\le c_2/\gamma_0$,
and $2\sqrt{\mathcal E}\le\varepsilon+\mathcal E/\varepsilon$ with
\[
\frac{\mathcal E}{\varepsilon}\le\frac{c_2}{\gamma_0}\Bigl(2M_\perp ab+\bigl(M_\perp+\tfrac18c_2\bigr)a^2\Bigr)
\le\tfrac18c_2b^2+\Bigl(\frac{8c_2M_\perp^2}{\gamma_0^2}+\frac{c_2(M_\perp+c_2/8)}{\gamma_0}\Bigr)a^2 ,
\]
by Young's inequality applied to the cross term.
In both cases Eq.~\eqref{eq:H1} yields $\tfrac58c_2b^2\le5\varepsilon+K_1'a^2$.
Moreover $\nu b\le(c_2\varepsilon/\gamma_0)\varrho$, so Eq.~\eqref{eq:range} gives
$a\le A_3\varepsilon+A_2b^2$ with $A_3:=A_0+A_1c_2\varrho/\gamma_0$, and, since
$\varepsilon\le1$ and $b\le\varrho$,
$a^2\le2A_3^2\varepsilon+2A_2^2\varrho^2b^2$. By (R4),
\begin{equation}
\label{eq:H-bound}
b\le C_H\,\varepsilon^{1/2},
~
C_H^2:=\frac2{c_2}\bigl(5+2K_1'A_3^2\bigr).
\end{equation}
 
Recall $\varepsilon_*=\gamma\nu/c_2\le\nu\le1$
by $\gamma\le c_2$ and (N3). If $\nu>0$ and $\varepsilon\le\varepsilon_*$, then
$\varepsilon+\varepsilon_*\le2\varepsilon_*\le2$, so
$\varepsilon/\sqrt{\gamma\nu}=\varepsilon/\sqrt{c_2\varepsilon_*}
\le\sqrt2\,\varepsilon/\sqrt{c_2(\varepsilon+\varepsilon_*)}$ and
$\varepsilon\le\sqrt2\,\varepsilon/\sqrt{\varepsilon+\varepsilon_*}$; Eq.~\eqref{eq:L-bound} gives
$b\le(4c_2^{-1/2}+2\sqrt2\,GA_0)\,\varepsilon/\sqrt{\varepsilon+\varepsilon_*}$. If
$\varepsilon\ge\varepsilon_*$, then $\varepsilon+\varepsilon_*\le2\varepsilon$, and
Eq.~\eqref{eq:H-bound} gives $b\le\sqrt2\,C_H\,\varepsilon/\sqrt{\varepsilon+\varepsilon_*}$.
This proves the first inequality in Eq.~\eqref{eq:cross-main} with
$C_{\mathcal N}:=\max\{4c_2^{-1/2}+2\sqrt2\,GA_0,\ \sqrt2\,C_H\}$. For the second, the first
gives $\nu b\le C_{\mathcal N}\varepsilon\,\nu/\sqrt{\varepsilon_*}
=C_{\mathcal N}\varepsilon\sqrt{c_2\nu/\gamma}\le
C_{\mathcal N}\sqrt{c_2\nu_1/\gamma_0}\,\varepsilon$ and
$b^2\le C_{\mathcal N}^2\varepsilon$, so Eq.~\eqref{eq:range} yields
$a\le C_{\mathcal R}\varepsilon$ with
$C_{\mathcal R}:=A_0+A_1C_{\mathcal N}\sqrt{c_2\nu_1/\gamma_0}+A_2C_{\mathcal N}^2$.
 
(1) For $\varepsilon\le\varepsilon_*$,
$\varepsilon/\sqrt{\varepsilon+\varepsilon_*}\le\varepsilon/\sqrt{\varepsilon_*}$, and
$\|d\|_2\le b+a$. (2) For any $\varepsilon$,
$\varepsilon/\sqrt{\varepsilon+\varepsilon_*}\le\varepsilon^{1/2}$. (3) For $\nu=0$,
$\varepsilon_*=0$ and Eq.~\eqref{eq:cross-main} reads $b\le C_{\mathcal N}\varepsilon^{1/2}$,
$a\le C_{\mathcal R}\varepsilon$. For fixed $\nu>0$, (1) applies to all
$\varepsilon\le\min\{\varepsilon_*,\varepsilon_0\}$.
\end{myproof}

\paragraph{Corollary 3 }
\begin{myproof}
\textit{(1)}
If $K_\eta\succ0$, the quadratic objective $q^{\boldsymbol\delta}$ is
strictly convex on $\mathbb R^k$ with the unique unconstrained
minimizer $\boldsymbol d=-K_\eta^{-1}J_\eta^\top\boldsymbol\delta$.
For sufficiently small $\boldsymbol\delta$ this point lies in the
interior of $\mathcal D_\eta$ and is therefore also the unique
constrained minimizer.
The map
$\boldsymbol\delta\mapsto-K_\eta^{-1}J_\eta^\top\boldsymbol\delta$
is linear, so its operator norm is the optimal Lipschitz coefficient,
and
$\|K_\eta^{-1}J_\eta^\top\|_2\le\|K_\eta^{-1}\|_2\|J_\eta\|_2$
gives Eq.~\eqref{eq:quadratic-optimal-lipschitz}.

\textit{(2)}
If $K_\eta$ is singular, then
$q^{\boldsymbol0}(\boldsymbol d)=\tfrac12\boldsymbol d^\top K_\eta\boldsymbol d\ge0$
vanishes exactly on $\ker(K_\eta)$, so
$\mathcal D_\eta\cap\ker(K_\eta)$ contains nonzero displacements and
the Lipschitz bound already fails at $\boldsymbol\delta=\boldsymbol0$.
To establish failure also for nonzero perturbations, suppose, to the
contrary, that some $C,\varepsilon_0>0$ satisfy
$\|\boldsymbol{\eta} ^{\boldsymbol\delta}_{\mathrm{quad}}-\boldsymbol{\eta^*}\|_2\le C\|\boldsymbol\delta\|_2$ and
$0<\|\boldsymbol\delta\|_2<\varepsilon_0$.
Choose $a>0$ with $\overline B(\boldsymbol0,a)\subseteq\mathcal D_\eta$,
and a nonzero perturbation with $\|\boldsymbol\delta\|_2<\varepsilon_0$
and $C\|\boldsymbol\delta\|_2<a/4$.
Every minimizer $\boldsymbol \eta^{\boldsymbol\delta}_{\mathrm{quad}}$ then satisfies
$\|\boldsymbol{\eta} ^{\boldsymbol\delta}_{\mathrm{quad}}-\boldsymbol{\eta^*}\|_2<a/4$, is an interior point of
$\mathcal D_\eta$, and is stationary:
$K_\eta\boldsymbol d_{\boldsymbol\delta}+J_\eta^\top\boldsymbol\delta
=\boldsymbol0$.
For a unit vector $\boldsymbol z\in\ker(K_\eta)$ this gives
$\boldsymbol\delta^\top J_\eta\boldsymbol z
=\boldsymbol z^\top
(K_\eta\boldsymbol d_{\boldsymbol\delta}+J_\eta^\top\boldsymbol\delta)
=0$, hence
\[
q^{\boldsymbol\delta}
\bigl(\boldsymbol d_{\boldsymbol\delta}+\tfrac a2\boldsymbol z\bigr)
=q^{\boldsymbol\delta}(\boldsymbol d_{\boldsymbol\delta})
+\tfrac a2\bigl(\boldsymbol z^\top K_\eta\boldsymbol d_{\boldsymbol\delta}
+\boldsymbol\delta^\top J_\eta\boldsymbol z\bigr)
+\tfrac{a^2}{8}\,\boldsymbol z^\top K_\eta\boldsymbol z
=q^{\boldsymbol\delta}(\boldsymbol d_{\boldsymbol\delta}).
\]
The shifted point has norm at most $a/4+a/2<a$, so it belongs to
$\mathcal D_\eta$ and is also a minimizer; but its norm is at least
$a/2-a/4=a/4>C\|\boldsymbol\delta\|_2$, contradicting the assumed
bound.
\end{myproof}

\paragraph{Theorem 4 }
\begin{myproof}
Positive semidefiniteness of $K_\eta$ implies
$\tilde K_\eta\succeq0,
~
\boldsymbol k_{\eta,p}\in\operatorname{range}(\tilde K_\eta)$. Moreover, the kernel assumption gives $\operatorname{range}(\tilde{\boldsymbol J}^{\top})
\subseteq\operatorname{range}(\tilde K_\eta)$. Consequently, for every $x$ and $\boldsymbol\delta$,
$\boldsymbol b
:=
\boldsymbol k_{\eta,p}x
+\tilde{\boldsymbol J}^{\top}\boldsymbol\delta
\in\operatorname{range}(\tilde K_\eta)$. Completing the square yields
\begin{align}
q^{\boldsymbol\delta}
(\boldsymbol E^\top\boldsymbol u+\boldsymbol e_p x)
=
\frac12
(\boldsymbol u+\tilde K_\eta^\dagger\boldsymbol b)^\top
\tilde K_\eta
(\boldsymbol u+\tilde K_\eta^\dagger\boldsymbol b)+\frac12K_{\eta,pp}x^2
+(j_p^\top\boldsymbol\delta)x
-\frac12\boldsymbol b^\top
\tilde K_\eta^\dagger\boldsymbol b.
\end{align}
The nuisance minimum is therefore attained at
$\boldsymbol u=-\tilde K_\eta^\dagger\boldsymbol b
+\boldsymbol w$, with
$\boldsymbol w\in\ker(\tilde K_\eta)$, and
$q_p^{\boldsymbol\delta}(x)
=
\frac12\kappa_p^K x^2
+(\boldsymbol g_p^K)^\top\boldsymbol\delta\,x
-\frac12\boldsymbol\delta^\top
\tilde{\boldsymbol J}\tilde K_\eta^\dagger
\tilde{\boldsymbol J}^{\top}\boldsymbol\delta$.
Taking $\boldsymbol\delta=0$ shows that $\kappa_p^K\ge0$.

If $\kappa_p^K>0$, this profile is strictly convex,
and its unique unconstrained minimum is
Eq.~\eqref{eq:quadratic-coordinate-response}.
It tends to zero with $\boldsymbol\delta$ and hence
belongs to the interior of $I-\eta_p^*$ for sufficiently
small perturbations. Furthermore,
$|\Delta\eta_p^{\mathrm{quad}}(\boldsymbol\delta)|
\le
\frac{\|\boldsymbol g_p^K\|_2}{\kappa_p^K}
\|\boldsymbol\delta\|_2$.
Equality is attained for perturbations parallel to
$\boldsymbol g_p^K$ when $\boldsymbol g_p^K\ne0$;
if $\boldsymbol g_p^K=0$, the response vanishes.
Thus the coefficient is exact.

To establish the null-space characterization, write
$\boldsymbol z=\boldsymbol E^\top\boldsymbol u
+\boldsymbol e_p x$. The same completion of the square gives
\[
\boldsymbol z^\top K_\eta\boldsymbol z
=
(\boldsymbol u+\tilde K_\eta^\dagger
 \boldsymbol k_{\eta,p}x)^\top
\tilde K_\eta
(\boldsymbol u+\tilde K_\eta^\dagger
 \boldsymbol k_{\eta,p}x)
+\kappa_p^K x^2.
\]
If $\kappa_p^K>0$, every vector in  $\ker(K_\eta)$
must therefore have $x=0$.
Conversely, if $\kappa_p^K=0$, the vector
$\boldsymbol z_0
:=
\boldsymbol e_p-\boldsymbol E^\top
\tilde K_\eta^\dagger\boldsymbol k_{\eta,p}$ satisfies
$\boldsymbol z_0^\top K_\eta\boldsymbol z_0=0$,
hence $K_\eta\boldsymbol z_0=0$, while
$\boldsymbol e_p^\top\boldsymbol z_0=1$.
This proves Eq.~\eqref{eq:coordinate-K-nullspace}.

Finally, suppose that $\kappa_p^K=0$.
If $\boldsymbol g_p^K=0$, the profile is constant in $x$
for every perturbation, so its entire minimizer set
cannot converge to zero.
If $\boldsymbol g_p^K\ne0$, take
$\boldsymbol\delta=t\boldsymbol g_p^K/
\|\boldsymbol g_p^K\|_2$ with $t>0$.
The profile is then strictly increasing in $x$, and its
minimum on $I-\eta_p^*$ is the fixed negative endpoint.
Its displacement does not tend to zero as $t\to0$.
In either case, coordinate Lipschitz stability fails.
 \end{myproof}
\paragraph{Theorem 5 }
\begin{myproof}
 The factorization
$K_\eta=F_\eta^{1/2}A_\eta F_\eta^{1/2}$
gives $A_\eta\succ0$ and $K_\eta^{-1}
=F_\eta^{-1/2}A_\eta^{-1}F_\eta^{-1/2}$. By the block inverse formula and the coordinate response
in Eq.~\eqref{eq:quadratic-coordinate-response}, $\frac1{\kappa_p^K}
=\boldsymbol e_p^\top K_\eta^{-1}\boldsymbol e_p,
~
(L_p^K)^2
=\boldsymbol e_p^\top
K_\eta^{-1}F_\eta K_\eta^{-1}\boldsymbol e_p$. For the Fisher information matrix,
$(L_p^F)^2=1/\kappa_p^F
=\boldsymbol e_p^\top F_\eta^{-1}\boldsymbol e_p$.
Substitution yields the first two identities in
Eq.~\eqref{eq:rh-exact-spectral}; the third follows from
Eq.~\eqref{eq:rh-total-ratio}.

Since $\|\boldsymbol u_p\|_2=1$, the spectral bounds for
$A_\eta^{-1}$ immediately give the bounds on
$\mathcal R_p^{\mathrm{RH}}$ and
$\mathcal R_{p,\mathrm{curv}}^{\mathrm{RH}}$.
For the response factor, set 
$T:=A_\eta^{-1},~
m:=\beta_\eta^{-1},~M:=\alpha_\eta^{-1},
~
x:=\boldsymbol u_p^\top T\boldsymbol u_p,~
y:=\boldsymbol u_p^\top T^2\boldsymbol u_p$. Cauchy--Schwarz gives $x^2\le y$.
Moreover, $mI\preceq T\preceq MI$ implies
$T^2\preceq(m+M)T-mMI$, and hence
\[
1\le
\left(\mathcal R_{p,\mathrm{resp}}^{\mathrm{RH}}\right)^2
=\frac{y}{x^2}
\le
\frac{(m+M)x-mM}{x^2}
\le
\frac{(m+M)^2}{4mM}.
\]
The last inequality follows by completing the square:
\[
\frac{(m+M)^2}{4mM}
-\frac{(m+M)x-mM}{x^2}
=
\frac{\bigl((m+M)x-2mM\bigr)^2}{4mMx^2}
\ge0.
\]
This proves Eq.~\eqref{eq:rh-general-bounds}.

If $\rho<1$, then
$(1-\rho)I\preceq A_\eta\preceq(1+\rho)I$. Applying the preceding estimates with
$m=(1+\rho)^{-1}$ and $M=(1-\rho)^{-1}$
gives Eq.~\eqref{eq:rh-small-bounds}.

If $\rho\ge1$, positive definiteness of $A_\eta$ gives
$\lambda_{\min}(B_\eta)>-1$.
Thus every negative eigenvalue of $B_\eta$ has magnitude
less than one, so
$\lambda_{\max}(B_\eta)=\rho$ and
$\beta_\eta=1+\rho$.
To show that $\rho$ alone provides no finite upper bound,
fix $\rho\ge1$ and take
$F_\eta=I_2,~
C_\eta=
\begin{pmatrix}
\rho&0\\
0&-1+\varepsilon
\end{pmatrix},
~ 0<\varepsilon<1$.
Then $K_\eta=\operatorname{diag}(1+\rho,\varepsilon)\succ0$
and $\|B_\eta\|_2=\rho$, but
$\mathcal R_2^{\mathrm{RH}}=\varepsilon^{-1}\to\infty$
as $\varepsilon\to0^+$.
For $k>2$, the same example extends by adding identity
blocks to $F_\eta$ and $K_\eta$.
\end{myproof}

\subsection*{Supplementary Lemmas}

\begin{lemma}[Localization of perturbed minimizers]
\label{lem:localization}
Under Assumption~\ref{ass:admissible},
\begin{equation}
\label{eq:nonlinear-minimizer-localization}
\|\boldsymbol\eta_{\boldsymbol\delta}-\boldsymbol\eta^*\|_2
\longrightarrow0
~\text{as }\boldsymbol\delta\to\boldsymbol0 .
\end{equation}
In particular, there exists $\varepsilon_1>0$ such that, for
$\|\boldsymbol\delta\|_2<\varepsilon_1$, every perturbed minimizer is
an interior point of $U_\eta$ and hence a stationary point of
$l_{\boldsymbol\delta}$.
\end{lemma}

\begin{myproof}
Expanding the square in
Eq.~\eqref{eq:perturbed-loss-notation},
\[
l_{\boldsymbol\delta}(\boldsymbol\eta)
=l_{\boldsymbol0}(\boldsymbol\eta)
+\boldsymbol\delta^\top\bigl(\boldsymbol h(\boldsymbol t,\boldsymbol\eta)
-\widehat{\boldsymbol h}\bigr)
+\tfrac12\|\boldsymbol\delta\|_2^2 ,
\]
so that
$\sup_{U_\eta}|l_{\boldsymbol\delta}-l_{\boldsymbol0}|
\le\|\boldsymbol\delta\|_2
\sup_{U_\eta}\|\boldsymbol h-\widehat{\boldsymbol h}\|_2
+\tfrac12\|\boldsymbol\delta\|_2^2\to0$:
$l_{\boldsymbol\delta}$ converges to $l_{\boldsymbol0}$ uniformly on
the compact set $U_\eta$.
If Eq.~\eqref{eq:nonlinear-minimizer-localization} failed, there
would exist $\rho>0$, perturbations $\boldsymbol\delta_j\to\boldsymbol0$
and minimizers $\boldsymbol\eta_j$
with $\|\boldsymbol\eta_j-\boldsymbol\eta^*\|_2\ge\rho$.
By compactness a subsequence converges to some
$\bar{\boldsymbol\eta}\in U_\eta$ with
$\|\bar{\boldsymbol\eta}-\boldsymbol\eta^*\|_2\ge\rho$.
Passing to the limit in
$l_{\boldsymbol\delta_j}(\boldsymbol\eta_j)
\le l_{\boldsymbol\delta_j}(\boldsymbol\eta)$, valid for every
$\boldsymbol\eta\in U_\eta$, and using uniform convergence gives
$l_{\boldsymbol0}(\bar{\boldsymbol\eta})\le l_{\boldsymbol0}(\boldsymbol\eta)$
for all $\boldsymbol\eta\in U_\eta$.
Thus $\bar{\boldsymbol\eta}$ is a second unperturbed minimizer,
contradicting Assumption~\ref{ass:admissible}.
The last claim follows because $\boldsymbol\eta^*$ lies in the
interior of $U_\eta$.
\end{myproof}

\begin{lemma}[Spectral form of the projection onto the tangent space]
\label{lem:Q}
Let $\boldsymbol J_\eta\in\mathbb R^{NL\times k}$ and $\boldsymbol F:=\boldsymbol J_\eta^\top\boldsymbol J_\eta$,
and let $\boldsymbol F=\sum_i\lambda_i\boldsymbol v_i\boldsymbol v_i^\top$ be a spectral
decomposition with $\{\boldsymbol v_i\}$ orthonormal. Writing $\mathcal I:=\{i:\lambda_i>0\}$
for the indices of the nonzero eigenvalues, set
$\boldsymbol u_i:=\lambda_i^{-1/2}\,\boldsymbol J_\eta\boldsymbol v_i,\quad i\in\mathcal I $.
Then $\{\boldsymbol u_i\}_{i\in\mathcal I}$ is orthonormal and
\begin{equation}
\label{eq:Qform}
\boldsymbol P:=\sum_{i\in\mathcal I}\boldsymbol u_i\boldsymbol u_i^\top
=\boldsymbol J\boldsymbol F^{\dagger}\boldsymbol J^{\top},
~
\boldsymbol Q:=\boldsymbol I_{NL}-\boldsymbol P,
\end{equation}
where $\boldsymbol F^{\dagger}=\sum_{i\in\mathcal I}\lambda_i^{-1}\boldsymbol v_i\boldsymbol v_i^\top$
is the Moore--Penrose pseudoinverse. $\boldsymbol P$ is the orthogonal projection of
$\mathbb R^{NL}$ onto $\operatorname{range}(\boldsymbol J_\eta)$ and $\boldsymbol Q$ the orthogonal
projection onto $\operatorname{range}(\boldsymbol J_\eta)^{\perp}$; in particular
\begin{equation}
\label{eq:Qprops}
\boldsymbol Q^{\top}=\boldsymbol Q,~ \boldsymbol Q^{2}=\boldsymbol Q,~
\boldsymbol Q\boldsymbol J=\boldsymbol 0,~
\|\boldsymbol Q\boldsymbol y\|_2\le\|\boldsymbol y\|_2\ \ \forall\boldsymbol y .
\end{equation}
\end{lemma}
\begin{myproof}
For $i,j\in\mathcal I$,
$\boldsymbol u_i^\top\boldsymbol u_j=(\lambda_i\lambda_j)^{-1/2}\boldsymbol v_i^\top\boldsymbol F\boldsymbol v_j
=(\lambda_i\lambda_j)^{-1/2}\lambda_j\delta_{ij}=\delta_{ij}$, so the $\boldsymbol u_i$ are
orthonormal and $\boldsymbol P$ is an orthogonal projection. Substituting the spectral form of
$\boldsymbol F^{\dagger}$ gives
$\boldsymbol J_\eta\boldsymbol F^{\dagger}\boldsymbol J_\eta^{\top}
=\sum_{i\in\mathcal I}\lambda_i^{-1}(\boldsymbol J_\eta\boldsymbol v_i)(\boldsymbol J_\eta\boldsymbol v_i)^\top
=\sum_{i\in\mathcal I}\boldsymbol u_i\boldsymbol u_i^\top=\boldsymbol P$. Since
$\boldsymbol F\boldsymbol v=\boldsymbol 0$ iff
$\|\boldsymbol J_\eta\boldsymbol v\|_2^2=\boldsymbol v^\top\boldsymbol F\boldsymbol v=0$ iff
$\boldsymbol J\boldsymbol v=\boldsymbol 0$, the vectors $\{\boldsymbol v_i\}_{i\in\mathcal I}$
span $(\ker\boldsymbol J)^{\perp}$, so $\{\boldsymbol u_i\}_{i\in\mathcal I}$ is an orthonormal
basis of $\operatorname{range}(\boldsymbol J_\eta)$ and $\operatorname{range}(\boldsymbol P)
=\operatorname{range}(\boldsymbol J)$. Consequently $\boldsymbol P\boldsymbol J_\eta=\boldsymbol J_\eta$,
i.e.\ $\boldsymbol Q\boldsymbol J_\eta=\boldsymbol 0$; the remaining identities hold for any
orthogonal projection.
\end{myproof}

\begin{lemma}[Operator bound for a projected symmetric bilinear map]
\label{lem:opbound}
Let $\boldsymbol H[\cdot,\cdot]:\mathbb R^k\times\mathbb R^k\to\mathbb R^{NL}$ be symmetric and
bilinear, and let $\boldsymbol B\in\mathbb R^{NL\times NL}$ be an orthogonal projection
\textup(so that $\|\boldsymbol B\boldsymbol z\|_2\le\|\boldsymbol z\|_2$ for all
$\boldsymbol z$\textup). Define the diagonal supremum
\begin{equation}
\label{eq:opb-M}
M_{\boldsymbol B}:=\sup_{\|\boldsymbol\gamma\|_2=1}
\big\|\boldsymbol B\,\boldsymbol H[\boldsymbol\gamma,\boldsymbol\gamma]\big\|_2 ,
\end{equation}
which is finite whenever $\boldsymbol H$ is bounded. Then
\begin{equation}
\label{eq:opb-diag}
\big\|\boldsymbol B\,\boldsymbol H[\boldsymbol z,\boldsymbol z]\big\|_2
\le M_{\boldsymbol B}\,\|\boldsymbol z\|_2^2
~(\boldsymbol z\in\mathbb R^k),
\end{equation}
and, more generally,
\begin{equation}
\label{eq:opb-cross}
\big\|\boldsymbol B\,\boldsymbol H[\boldsymbol x,\boldsymbol y]\big\|_2
\le M_{\boldsymbol B}\,\|\boldsymbol x\|_2\,\|\boldsymbol y\|_2
~(\boldsymbol x,\boldsymbol y\in\mathbb R^k).
\end{equation}
\end{lemma}

\begin{myproof}
For $\boldsymbol z=\boldsymbol0$ both sides vanish. For $\boldsymbol z\neq\boldsymbol0$ put
$\boldsymbol\gamma=\boldsymbol z/\|\boldsymbol z\|_2$; bilinearity gives
$\boldsymbol H[\boldsymbol z,\boldsymbol z]=\|\boldsymbol z\|_2^2\,\boldsymbol H[\boldsymbol\gamma,\boldsymbol\gamma]$,
hence by Eq.~\eqref{eq:opb-M},
\begin{equation}
\big\|\boldsymbol B\,\boldsymbol H[\boldsymbol z,\boldsymbol z]\big\|_2
=\|\boldsymbol z\|_2^2\,\big\|\boldsymbol B\,\boldsymbol H[\boldsymbol\gamma,\boldsymbol\gamma]\big\|_2
\le M_{\boldsymbol B}\,\|\boldsymbol z\|_2^2 .
\end{equation}

Symmetry and bilinearity yield
\begin{equation}
\boldsymbol H[\boldsymbol x+\boldsymbol y,\boldsymbol x+\boldsymbol y]
-\boldsymbol H[\boldsymbol x-\boldsymbol y,\boldsymbol x-\boldsymbol y]
=4\,\boldsymbol H[\boldsymbol x,\boldsymbol y],
\end{equation}
so, applying the linear map $\boldsymbol B$, the triangle inequality, and Eq.~\eqref{eq:opb-diag}
to each diagonal term,
\begin{align}
\big\|\boldsymbol B\,\boldsymbol H[\boldsymbol x,\boldsymbol y]\big\|_2
&=\tfrac14\big\|\boldsymbol B\,\boldsymbol H[\boldsymbol x+\boldsymbol y,\boldsymbol x+\boldsymbol y]
-\boldsymbol B\,\boldsymbol H[\boldsymbol x-\boldsymbol y,\boldsymbol x-\boldsymbol y]\big\|_2 \nonumber \\
&\le\tfrac14 M_{\boldsymbol B}\Big(\|\boldsymbol x+\boldsymbol y\|_2^2+\|\boldsymbol x-\boldsymbol y\|_2^2\Big).    
\end{align}
By the parallelogram identity
$\|\boldsymbol x+\boldsymbol y\|_2^2+\|\boldsymbol x-\boldsymbol y\|_2^2
=2\|\boldsymbol x\|_2^2+2\|\boldsymbol y\|_2^2$, this gives the intermediate bound
\begin{equation}
\label{eq:opb-mid}
\big\|\boldsymbol B\,\boldsymbol H[\boldsymbol x,\boldsymbol y]\big\|_2
\le\tfrac12 M_{\boldsymbol B}\big(\|\boldsymbol x\|_2^2+\|\boldsymbol y\|_2^2\big).
\end{equation}

The left-hand side of Eq.~\eqref{eq:opb-mid} is invariant under
$(\boldsymbol x,\boldsymbol y)\mapsto(t\boldsymbol x,\,t^{-1}\boldsymbol y)$ for any $t>0$, by
bilinearity, whereas the right-hand side is not. If $\boldsymbol x=\boldsymbol0$ or
$\boldsymbol y=\boldsymbol0$ then Eq.~\eqref{eq:opb-cross} holds trivially; otherwise apply
Eq.~\eqref{eq:opb-mid} to $(t\boldsymbol x,\,t^{-1}\boldsymbol y)$:
\begin{equation}
\big\|\boldsymbol B\,\boldsymbol H[\boldsymbol x,\boldsymbol y]\big\|_2
\le\tfrac12 M_{\boldsymbol B}\Big(t^2\|\boldsymbol x\|_2^2+t^{-2}\|\boldsymbol y\|_2^2\Big).  
\end{equation}
Minimising the right-hand side over $t>0$, the optimum
$t^2=\|\boldsymbol y\|_2/\|\boldsymbol x\|_2$ attains equality in the arithmetic--geometric mean
inequality and gives $\tfrac12 M_{\boldsymbol B}\cdot2\|\boldsymbol x\|_2\|\boldsymbol y\|_2
=M_{\boldsymbol B}\|\boldsymbol x\|_2\|\boldsymbol y\|_2$, which is Eq.~\eqref{eq:opb-cross}.
\end{myproof}

\begin{lemma}[Scalar quadratic bound]
\label{lem:quad}
Let $\alpha>0$ and $C_0,C_1\ge0$, and suppose $b\ge0$ satisfies
$\alpha b^2\le C_0+C_1 b$. Then
\begin{equation}
\label{eq:quad-root}
b\le\frac{C_1+\sqrt{C_1^2+4\alpha C_0}}{2\alpha}
\le\sqrt{\frac{2C_0}{\alpha}}+\frac{2C_1}{\alpha}.
\end{equation}
In particular, if $C_0\le c_0\varepsilon$ and $C_1\le c_1\varepsilon$ for constants
$c_0,c_1\ge0$, then 
\begin{equation}
b\le\sqrt{2c_0/\alpha}\,\varepsilon^{1/2}+2(c_1/\alpha)\varepsilon,
\end{equation}
so
the leading term is of order $\varepsilon^{1/2}$ and is governed entirely by $C_0$.
\end{lemma}

\begin{lemma}[Residual correction of curvature and eigenspaces]
\label{thm:residual-schur}
At an interior local minimum, let
$K_\eta=F_\eta+C_\eta\succeq0$, where
$\operatorname{rank}(F_\eta)=r$ and $1\le r<k$.
Choose an orthogonal matrix $Q=[U,V]$ such that
\begin{equation}
Q^\top F_\eta Q
=
\begin{pmatrix}
\Lambda&0\\
0&0
\end{pmatrix},
~
\Lambda=\operatorname{diag}(\lambda_1,\ldots,\lambda_r)
\succ0,
~Q^\top C_\eta Q
=
\begin{pmatrix}
E&B\\
B^\top&D
\end{pmatrix},
~
A:=\Lambda+E.
\label{eq:residual-blocks}
\end{equation}
Then the following statements hold.

\begin{enumerate}
\item
\emph{Positive curvature.}
The quadratic form of $K_\eta$ is positive definite on
$\operatorname{range}(U)$ if and only if
\begin{equation}
A\succ0
\quad\Longleftrightarrow\quad
\lambda_{\min}
\bigl(\Lambda^{-1/2}E\Lambda^{-1/2}\bigr)>-1.
\label{eq:residual-A-positive}
\end{equation}

\item
\emph{Relative curvature change.}
The maximal relative correction on $\operatorname{range}(U)$ is
\begin{equation}
\epsilon_U
:=
\sup_{\boldsymbol a\ne0}
\frac{|\boldsymbol a^\top E\boldsymbol a|}
     {\boldsymbol a^\top\Lambda\boldsymbol a}
=
\|\Lambda^{-1/2}E\Lambda^{-1/2}\|_2.
\label{eq:residual-relative-curvature}
\end{equation}
Consequently,
$(1-\epsilon_U)\Lambda\preceq A
\preceq(1+\epsilon_U)\Lambda$;
in particular, $\epsilon_U<1$ implies $A\succ0$.

\item
\emph{Rank and null space.}
If $A\succ0$, the Schur complement
$S:=D-B^\top A^{-1}B$ satisfies $S\succeq0$, and
\begin{equation}
\begin{aligned}
\operatorname{rank}(K_\eta)
=r+\operatorname{rank}(S),
~\ker(K_\eta)
=\left\{
(V-UA^{-1}B)\boldsymbol b:
\boldsymbol b\in\ker(S)
\right\}.
\end{aligned}
\label{eq:residual-rank-kernel}
\end{equation}
In particular, $K_\eta\succ0$ if and only if $S\succ0$.

\item
\emph{Eigenvalue and eigenvector changes.}
Item (4) is the standard second-order perturbation expansion of a simple eigenvalue \cite{kato1966perturbation}, stated in the present notation; it is used only to define the diagnostic $\rho_{ji}$. Write $Q=[q_1,\ldots,q_k]$,
$\widehat C:=Q^\top C_\eta Q$, and set
$\lambda_{r+1}=\cdots=\lambda_k=0$.
For a simple positive eigenvalue $\lambda_i$, let
$\gamma_i:=\min_{j\ne i}|\lambda_i-\lambda_j|$.
As $\|C_\eta\|_2/\gamma_i\to0$, the corresponding
eigenpair $(\widehat\lambda_i,\widehat q_i)$ of $K_\eta$
satisfies
\begin{equation}
\begin{aligned}
\widehat\lambda_i
=
\lambda_i+\widehat C_{ii}
+\sum_{j\ne i}
\frac{\widehat C_{ji}^{\,2}}{\lambda_i-\lambda_j}
+O\!\left(\frac{\|C_\eta\|_2^3}{\gamma_i^2}\right),
~\widehat q_i-q_i
=
\sum_{j\ne i}
\frac{\widehat C_{ji}}{\lambda_i-\lambda_j}q_j
+O\!\left(\frac{\|C_\eta\|_2^2}{\gamma_i^2}\right),
\end{aligned}
\label{eq:residual-eigenpair-approximation}
\end{equation}
where $\widehat q_i$ is a unit eigenvector with
$q_i^\top\widehat q_i>0$.
Thus, the norm of the first-order eigenvector correction is
$\left(\sum_{j\ne i}\rho_{ji}^2\right)^{1/2},
~
\rho_{ji}:=
\frac{|\widehat C_{ji}|}{|\lambda_i-\lambda_j|}$.
\end{enumerate}
\end{lemma}

\begin{lemma}[Structure of the residual correction]
\label{lem:residual-structure}
Let $\eta^*$ be as in Assumption~\ref{ass:admissible}, so that
$J_\eta^\top r^*=0$ and $K_\eta\succeq0$. The following statements hold.
\begin{enumerate}
\item[(a)] For all $x,y\in\mathbb R^k$, Eq.~\eqref{eq:C-normal} holds.
\item[(b)] If $\nu>0$, then $\lambda_{\min}(D)=\gamma\nu$, and
$0\le\gamma\le c_2\le M_\perp$.
\item[(c)] If $\nu>0$ and $\nu\le\gamma\sigma_+^2/(4M_\perp^2)$, then
$A\succeq\tfrac12\sigma_+^2I$,
\begin{equation}
\label{eq:mu-bound}
\tfrac12\gamma\nu\le\mu\le\gamma\nu ,
~
\tfrac12\mu\le\lambda_{\min}(K_\eta)\le\mu ,
\end{equation}
and in particular $K_\eta\succ0$ whenever $\gamma>0$.
\end{enumerate}
\end{lemma}
 
\begin{myproof}
(a) Since $C_\eta=\sum_a r^*_a\nabla^2_\eta h_a$,
$x^\top C_\eta y=\langle r^*,H_\eta[x,y]\rangle$. Stationarity gives
$r^*\in\operatorname{range}(J_\eta)^\perp$, so $r^*=\Pi^\perp r^*$ and
$\langle r^*,H_\eta[x,y]\rangle=\langle r^*,\Pi^\perp H_\eta[x,y]\rangle$. The
bound follows from the Cauchy--Schwarz inequality and Lemma~\ref{lem:opbound} applied with
the projection $\Pi^\perp$.
(b) For $\xi=Vw$ with $\|w\|_2=1$, part (a) gives
$w^\top Dw=\xi^\top C_\eta\xi=\nu\langle\hat r,\Pi^\perp H_\eta[\xi,\xi]\rangle$,
so $\lambda_{\min}(D)=\gamma\nu$. Since $F_\eta\xi=0$,
$\xi^\top C_\eta\xi=\xi^\top K_\eta\xi\ge0$, hence $\gamma\ge0$. If $\beta_0$
attains the minimum in Eq.~\eqref{eq:c2-def}, then
$\gamma\le\langle\hat r,\Pi^\perp H_\eta[\beta_0,\beta_0]\rangle\le
\|\Pi^\perp H_\eta[\beta_0,\beta_0]\|_2=c_2$, and $c_2\le M_\perp$ by definition.
(c) By (b), $\nu M_\perp\le\gamma\sigma_+^2/(4M_\perp)\le\sigma_+^2/4$. Since
$\Lambda\succeq\sigma_+^2I$ and $\|E\|_2,\|B\|_2\le\nu M_\perp$ by (a),
$A=\Lambda+E\succeq\tfrac34\sigma_+^2I\succeq\tfrac12\sigma_+^2I$. Then
$S=D-B^\top A^{-1}B\preceq D$ gives $\mu\le\gamma\nu$, and
$S\succeq D-\|B\|_2^2\|A^{-1}\|_2I$ gives
$\mu\ge\gamma\nu-2\nu^2M_\perp^2/\sigma_+^2\ge\tfrac12\gamma\nu$. For
$z=Uz_U+Vz_V$ and $y:=z_U+A^{-1}Bz_V$,
\[
z^\top K_\eta z=y^\top Ay+z_V^\top Sz_V\ge\tfrac12\sigma_+^2\|y\|_2^2+\mu\|z_V\|_2^2 .
\]
With $\kappa:=\|A^{-1}B\|_2\le2\nu M_\perp/\sigma_+^2\le\tfrac12$, we have
$\|z_U\|_2\le\|y\|_2+\kappa\|z_V\|_2$, so
$\|z\|_2^2\le2\|y\|_2^2+2\|z_V\|_2^2$; since $\mu\le\gamma\nu\le\sigma_+^2/4$,
the right-hand side above is at least $\mu(\|y\|_2^2+\|z_V\|_2^2)\ge\tfrac12\mu\|z\|_2^2$.
Choosing $z_V$ with $Sz_V=\mu z_V$ and $z_U=-A^{-1}Bz_V$ gives
$z^\top K_\eta z=\mu\|z_V\|_2^2\le\mu\|z\|_2^2$.
\end{myproof}

\subsection*{Computing the FIM and H for ODE}
We introduce how to obtain the gradient vector $\nabla_\theta h_l(t;\boldsymbol{\theta})$ and Hessian matrix $\nabla^2_\theta h_l(t;\boldsymbol{\theta})$ $(l=1,2,\cdots,L)$ for ordinary differential equation (ODE). When we have an ODE as
\begin{equation}
\frac{dx}{dt} = f(x; \boldsymbol{\theta}), ~ x \in \mathbb{R}^n, ~ f \in \mathbb{R}^n,~\boldsymbol{\theta} \in \mathbb{R}^k,
\end{equation}
we compute the gradient of $\boldsymbol{\theta}$ as
\begin{align}
\frac{dy}{dt} = \frac{\partial}{\partial \theta} \frac{dx}{dt} = \frac{\partial}{\partial \theta} f(x; \theta),~y = \nabla_{ \boldsymbol{\theta}}x, ~ y \in \mathbb{R}^{n \times k}.
\end{align}
Then, we can write the ODE of each component  of $y$ as
\begin{align}
\label{eq:S57}
\frac{d}{dt} y_{ij} &= \frac{\partial}{\partial \theta_j} f_i(x; \theta) = \sum_{m=1}^n \frac{\partial f_i}{\partial x_m} \frac{\partial x_m}{\partial \theta_j} + \frac{\partial f_i}{\partial \theta_j} = \sum_{m=1}^n \frac{\partial f_i}{\partial x_m} y_{mj} + \frac{\partial f_i}{\partial \theta_j}
\end{align}
Then, we futher calculate the Hessian tensor $H$ as
\begin{equation}
\label{eq:S58}
\frac{d}{dt} z_{ijp} = \frac{\partial}{\partial \theta_p} \left( \sum_{m=1}^n \frac{\partial f_i}{\partial x_m} y_{mj} + \frac{\partial f_i}{\partial \theta_j} \right), z = \frac{\partial y}{\partial \theta}, ~ z_{ijp}= \frac{\partial y_{ij}}{\partial \theta_p}, ~ z \in \mathbb{R}^{n \times k \times k}.
\end{equation}
The Eq.~\ref{eq:S58} can be rewritten as five terms as follows
\begin{align}
\text{Term 1} &: \sum_{m=1}^n \frac{\partial f_i}{\partial x_m} \frac{\partial y_{mj}}{\partial \theta_p} = \sum_{m=1}^n \frac{\partial f_i}{\partial x_m} z_{mjp}; \nonumber \\
\text{Term 2} &: \sum_{m=1}^n \left( \sum_{l=1}^n \frac{\partial}{\partial x_l} \left( \frac{\partial f_i}{\partial x_m} \right) \frac{\partial x_l}{\partial \theta_p} \right) y_{mj} = \sum_{m,l}^n \frac{\partial^2 f_i}{\partial x_l \partial x_m} y_{lp} y_{mj}; \nonumber \\
\text{Term 3} &: \quad \sum_{m=1}^n \frac{\partial^2 f_i}{\partial x_m \partial \theta_p} y_{mj}; \nonumber \\
\text{Term 4} &: \quad \sum_{m=1}^n \frac{\partial}{\partial x_m} \left( \frac{\partial f_i}{\partial \theta_j} \right) \frac{\partial x_m}{\partial \theta_p} = \sum_{m=1}^n \frac{\partial^2 f_i}{\partial x_m \partial \theta_j} y_{mp}; \nonumber \\
\text{Term 5} &: \quad \frac{\partial^2 f_i}{\partial \theta_j \partial \theta_p}.
\end{align}
And we sort our these terms as
\begin{align}
\label{eq:S60}
\frac{d}{dt} z_{ijp} =& \sum_{m=1}^n \frac{\partial f_i}{\partial x_m} z_{mjp} + \sum_{m,l}^n \frac{\partial^2 f_i}{\partial x_l \partial x_m} y_{lp} y_{mj} + \sum_{m=1}^n \frac{\partial^2 f_i} 
{\partial x_m \partial \theta_p} y_{mj} \nonumber \\
&+ \sum_{m=1}^n \frac{\partial^2 f_i}{\partial x_m \partial \theta_j} y_{mp} + \frac{\partial^2 f_i}{\partial \theta_j \partial \theta_p}
\end{align}
When the observed variable is $h(x)\in \mathbb{R}^{L}$, the gradient matrix $\hat{y}\in \mathbb{R}^{L\times K}$ is computed as
$\hat{y} = \frac{\partial h}{\partial \theta} = H(x)y,~H(x) = \nabla_x h$
and Hessian tensor $\hat{z}\in \mathbb{R}^{L\times K\times K}$ is rewritten as
\begin{equation}
\hat{z} = \frac{\partial \hat{y}}{\partial \theta} = \frac{\partial H(x)}{\partial \theta} y + H(x) \frac{\partial y}{\partial \theta} = \frac{\partial H(x)}{\partial \theta} y + H(x)z
\end{equation}
\begin{equation}
\frac{\partial H(x)}{\partial \theta} = \frac{\partial H(x)}{\partial x} \frac{\partial x}{\partial \theta} = \mathcal{H}_{obs}\cdot y, ~\hat{z} = \mathcal{H}_{obs}[y,y]+H(x)z
\end{equation}
\begin{equation}
\hat{z}_{ijp} = \sum_{l=1}^{n}\sum_{r=1}^{n} \frac{\partial h_i}{\partial x_l\partial x_r} y_{lj}y_{rp}+\sum_{l=1}^{n} \frac{\partial h_i}{\partial x_l}z_{ljp}~(i\in\{1,2,\cdots,L\},j,p\in\{1,2,\cdots,K\})
\end{equation}
With respect to the initial value of Eqs~\ref{eq:S57} and \ref{eq:S60}, we set zero matrix or tensor since the parameters are independent on the initial value.

\subsection*{JAK--STAT signaling pathway model}
We considered a JAK--STAT signaling pathway model \cite{swameye2003identification,raue2009structural} describing receptor-mediated STAT phosphorylation, dimerization, nuclear translocation, and recycling. The model comprises four state variables: cytoplasmic unphosphorylated STAT ($x_1$), cytoplasmic phosphorylated STAT monomers ($x_2$), cytoplasmic phosphorylated STAT dimers ($x_3$), and nuclear STAT dimers ($x_4$). The external input $u(t)$ represents activated erythropoietin receptor (pEpoR) activity and is obtained by spline interpolation of the experimental receptor measurements (Fig.~\ref{fig:S1}). The governing equations are
\begin{equation}
\begin{aligned}
\dot{x}_1 &= -p_1 x_1 u(t) + 2p_4 x_4^{\tau},~
\dot{x}_2 = p_1 x_1 u(t) - p_2 x_2^2,\\
\dot{x}_3 &= \frac{1}{2}p_2 x_2^2 - p_3 x_3,~
\dot{x}_4 = p_3 x_3 - p_4 x_4^{\tau}.
\end{aligned}
\label{eq:jak_stat_model}
\end{equation}
Here, $p_1$ controls receptor-mediated STAT phosphorylation, $p_2$ governs dimerization, $p_3$ describes nuclear translocation of cytoplasmic dimers, and $p_4$ controls delayed nuclear processing and recycling. The factors $1/2$ and $2$ account for the stoichiometry of dimer formation and the return of two STAT monomers per recycled dimer, respectively.

The superscript $\tau$ denotes delayed processing rather than exponentiation. This delay is implemented using a linear chain approximation with ten intermediate steps, each with transition rate $p_4$, giving an effective mean delay
$\tau = \frac{10}{p_4}$. Accordingly, $p_4$ determines both the transition rate through the chain and the characteristic timescale of STAT recycling; $\tau$ is not treated as an independent parameter.

\subsection*{HIV host-virus dynamics}
The basic model of viral infection \cite{perelson2002modelling,stafford2000modeling} is used to understand the HIV infection dynamics as:
\begin{align}
\frac{dT}{dt} = \lambda -dT-kVT ,~
\frac{dI}{dt} = kVT-\delta I \nonumber ,~
\frac{dV}{dt}  = \pi I - cV
\end{align}
where includes activated CD4+ T cells $(T)$, productively infected CD4+ T cells $(I)$, and HIV concentration $(V)$. The following details regarding the model parameters are summarized in Table~\ref{tab:S1}. To account for the latency between exposure and clinical observation, the interval from the initial infection event to the first viral load measurement was set at 36 days for patient 3 and 40 days for patient 8 days from the model setup \cite{stafford2000modeling}.
\subsection*{Spatiotemporal dynamics of $A\beta$ on the brain}
The spatiotemporal dynamics of $A\beta$ on the brain is established by Network-based PDE model as
\begin{equation}
\label{eq:S76}
\frac{d\boldsymbol{A}}{dt} + L\boldsymbol{A} = \boldsymbol{\lambda} \odot \boldsymbol{A} \odot (\boldsymbol{K}-\boldsymbol{A}).
\end{equation}
Where $\boldsymbol{A}\in \mathbb{R}^{68}, \boldsymbol{\lambda}\in \mathbb{R}^{68}$, and $\boldsymbol{K}\in \mathbb{R}^{68}$ represent the $A\beta$ concentration, recruitment rate of $A\beta$, and capacity of $A\beta$, respectively. $L$ is the patient-specific graph Laplacian
corresponding to the patient’s specific brain functional connectivity network. We use the functional connectivity (FC) matrix of 68 regions to compute the patient-specific graph Laplacian $L$, which refers to the our group previous work on how to learn  patient-specific functional connectivity matrix \cite{li2026data}. 
\newpage
\subsection*{Supplementary Algorithm}
\begin{algorithm}[ht]
\caption{Verification of Theorem~3 on the JAK--STAT model}
\label{alg:thm3-jakstat}
\begin{algorithmic}[1]
\Require Fitted point $\eta^*$; stacked, $\sigma$-whitened
        Jacobian $J_\eta\in\mathbb{R}^{n\times k}$, output Hessians $H_\eta\in\mathbb{R}^{n\times k\times k}$ and
        residual $r^*\in\mathbb{R}^{n}$; rank tolerance $\tau$;
        perturbation grid $\mathcal{E}$.
\Ensure $\gamma,\ c_2,\ \varepsilon_*=\hat\gamma\nu$, and the empirical rate curve
        $\sup_\delta|\xi^\top\Delta\eta|$ versus $\varepsilon$.

\State $[U,\Sigma,V^\top]\gets\operatorname{SVD}(J_\eta)$; \ $z\gets$ last right singular vector
\State $Q\gets$ orthonormal basis of $z^\perp$ \ ($k\times(k-1)$)
\State $J_q\gets J_\eta Q$, ~ $H_q[\cdot,\cdot]\gets Q^\top H_\eta Q$

\State $[U_q,\Sigma_q,V_q^\top]\gets\operatorname{SVD}(J_q)$; \ $r\gets\#\{\sigma_i>\tau\,\sigma_{\max}\}$;
       \newline $\mathcal{V}\gets$ last $(k-1-r)$ columns of $V_q$
       
\State $\nu\gets\|r^*\|_2$ 
\State $C\gets\sum_{a=1}^{n} r^*_a\,\nabla^2_q h_a$; \quad
       $\gamma\gets\lambda_{\min}(\mathcal{V}^\top C\,\mathcal{V})/\nu$
       \Comment{Eq.~\eqref{eq:gamma}}
\State $\Pi_\perp\gets I-U_q[:,1{:}r]\,U_q[:,1{:}r]^\top$
\State $c_2\gets\min_{\|w\|=1}\|\Pi_\perp H_q[\mathcal{V}w,\mathcal{V}w]\|_2$,\quad
       $M_N\gets\max_{\|w\|=1}\|\Pi_\perp H_q[\mathcal{V}w,\mathcal{V}w]\|_2$
\State \textbf{assert} $\gamma>0$ and $c_2>0$ \Comment{hypotheses of Theorem~3}
\State $\hat\gamma\gets\gamma/c_2$,\quad $\varepsilon_*\gets\hat\gamma\,\nu=\gamma\nu/c_2$
       \Comment{Eq.~\eqref{eq:eps-star}}

\ForAll{$\varepsilon\in\mathcal{E}\cup\{\sqrt n\}$}
   \State $\Xi\gets\gamma\nu/(c_2\varepsilon)=\varepsilon_*/\varepsilon$
   \State regime$(\varepsilon)\gets$ \textsc{Lipschitz} if $\Xi>3$; \textsc{Hölder-}$\tfrac12$ if $\Xi<\tfrac13$; else \textsc{crossover}
\EndFor

\State $\xi\gets$ unit vector spanning $\mathcal{V}$;\newline
       $d_\pm\gets\pm\Pi_\perp H_q[\xi,\xi]/\|\Pi_\perp H_q[\xi,\xi]\|$;\quad
       $\{d_j\}_{j=1}^{40}\gets$ random unit vectors in $\mathbb{R}^n$
\State $u_0\gets\arg\min_u\ \tfrac12\|r^*+J_q u+\tfrac12 H_q[u,u]\|_2^2$
\ForAll{$\varepsilon\in\mathcal{E}$}
   \ForAll{$d\in\{d_+,d_-,d_1,\dots,d_{40}\}$}
      \State $\delta\gets\varepsilon d$
      \State $u(\delta)\gets\arg\min_u\ \tfrac12\|r^*+\delta+J_q u+\tfrac12 H_q[u,u]\|_2^2$
   \EndFor
   \State $D(\varepsilon)\gets\max_d\ |\xi^\top(u(\delta)-u_0)|$
\EndFor
\State $s(\varepsilon)\gets\Delta\log D/\Delta\log\varepsilon$ \Comment{local log--log slope}

\State Check $s(\varepsilon)\approx1$ for $\varepsilon<\varepsilon_*$ 
\State Check $s(\varepsilon)\approx\tfrac12$ for $\varepsilon\gg\varepsilon_*$ 
\State Check the transition lies at $\varepsilon\approx\varepsilon_*=\hat\gamma\nu$ 
\State $C_{\mathcal N}\gets\max_\varepsilon D(\varepsilon)\sqrt{\varepsilon+\varepsilon_*}/\varepsilon$;
       verify $D(\varepsilon)\le C_{\mathcal N}\,\varepsilon/\sqrt{\varepsilon+\varepsilon_*}$ for all $\varepsilon$
       \Comment{Eq.~\eqref{eq:cross-main}}
\State \Return $\gamma,c_2,M_N,\varepsilon_*$, regime$(\cdot)$, $D(\cdot)$, $s(\cdot)$, $C_{\mathcal N}$
\end{algorithmic}
\end{algorithm}

\newpage
\subsection*{Supplementary Figures}
\begin{figure}[ht]
\centering
\includegraphics[width=1.0\linewidth]{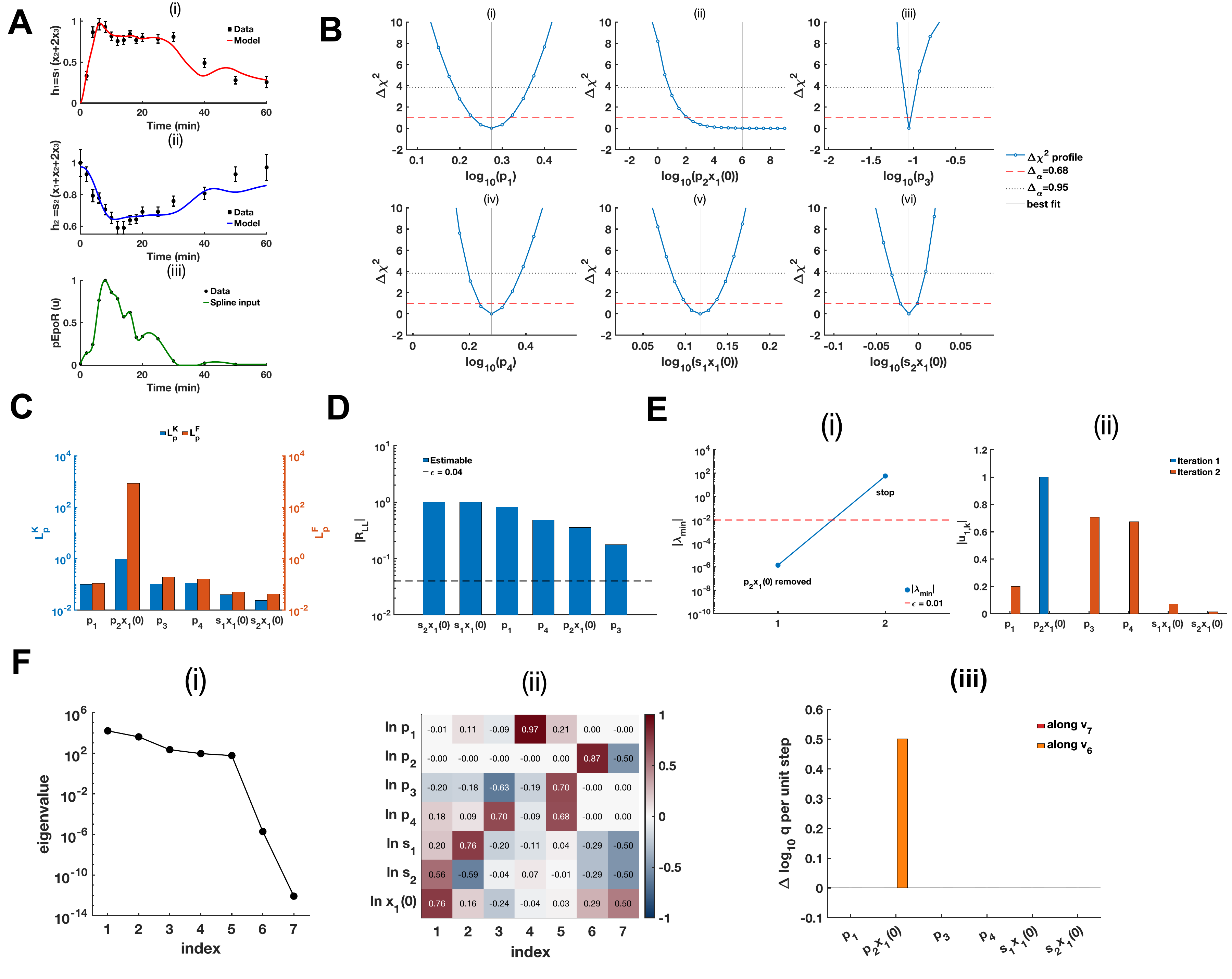}
\caption{\scriptsize\textbf{Practical identifiability of the JAK--STAT model in the symmetry-reduced parameters.}
\textbf{(A)} Fitted outputs $h_1$, $h_2$ (i--ii) and the pEpoR input $u(t)$ (iii).
\textbf{(B)} Profile likelihoods of $p_1$, $p_2x_1(0)$, $p_3$, $p_4$, $s_1x_1(0)$ and $s_2x_1(0)$ as $\Delta\chi^2$; dashed and dotted lines mark the 68\% and 95\% thresholds. Only $p_2x_1(0)$ is unbounded on the right.
\textbf{(C)} Coordinate Lipschitz coefficients $L_p^{K}$ (blue) and $L_p^{F}$ (orange), both largest for $p_2x_1(0)$.
\textbf{(D)} Yao method: all six combinations exceed $\epsilon=0.04$.
\textbf{(E)} Quaiser method: (i) smallest absolute eigenvalue per iteration against $\epsilon=0.01$, removing only $p_2x_1(0)$; (ii) components of the associated eigenvector.
\textbf{(F)} Origin of the combinations: (i) eigenvalues of $F_\eta$ in the original log-parameters; (ii) eigenvector components, $v_7$ being the scaling symmetry $x\mapsto cx$, $p_2\mapsto p_2/c$, $s_i\mapsto s_i/c$; (iii) change in each combination per unit step along $v_7$ and $v_6$: none along $v_7$, only $p_2x_1(0)$ along $v_6$.}
\label{fig:S1}
\end{figure}

\newpage
\subsection*{Supplementary Tables}

\begin{table}[htbp]
\centering
\caption{Biological interpretation and fitted parameter values for the JAK--STAT signaling pathway model. Values are reported on the $\log_{10}$ scale.}
\label{tab:jak_stat_parameters}
\small
\begin{tabular}{p{0.12\textwidth} p{0.64\textwidth} r}
\hline
Parameter & Biological interpretation & $\log_{10}$ value \\
\hline
$p_1$ & Rate coefficient for receptor-mediated phosphorylation of cytoplasmic STAT. & $0.27$ \\[4pt]
$p_2$ & Rate coefficient for dimerization of phosphorylated STAT monomers. & $5.69$ \\[4pt]
$p_3$ & Rate constant for nuclear translocation of phosphorylated STAT dimers. & $-1.05$ \\[4pt]
$p_4$ & Transition rate for delayed nuclear processing and STAT recycling, with effective mean delay $\tau=10/p_4$. & $0.28$ \\[4pt]
$s_1$ & Measurement scaling factor for cytoplasmic phosphorylated STAT, $h_1=s_1(x_2+2x_3)$. & $-0.19$ \\[4pt]
$s_2$ & Measurement scaling factor for total cytoplasmic STAT, $h_2=s_2(x_1+x_2+2x_3)$. & $-0.32$ \\[4pt]
$x_1(0)$ & Initial concentration of cytoplasmic unphosphorylated STAT. & $0.31$ \\
\hline
\end{tabular}
\end{table}

\begin{table}[htbp]
\centering
\footnotesize
\caption{Biological interpretation and units of the parameters in the HIV infection model shown in Fig.~\ref{fig:3}.}
\label{tab:S1}
\renewcommand{\arraystretch}{1.2}
\begin{tabular}{l p{0.52\textwidth} l}
\hline
\textbf{Parameter} & \textbf{Biological description} & \textbf{Unit} \\
\hline
$\lambda$
& Generation rate of activated CD4$^+$ T cells
& $\mathrm{cell}\,\mu\mathrm{L}^{-1}\,\mathrm{day}^{-1}$ \\

$d$
& Death rate constant of activated CD4$^+$ T cells
& $\mathrm{day}^{-1}$ \\

$k$
& Infection rate coefficient for activated CD4$^+$ T cells
& $\mu\mathrm{L}\,\mathrm{virion}^{-1}\,\mathrm{day}^{-1}$ \\

$\delta$
& Death rate constant of productively infected cells
& $\mathrm{day}^{-1}$ \\

$\pi$
& Virion production rate per infected cell
& $\mathrm{virion}\,\mathrm{cell}^{-1}\,\mathrm{day}^{-1}$ \\

$c$
& Viral clearance rate constant
& $\mathrm{day}^{-1}$ \\
\hline
\end{tabular}
\end{table}

\begin{table}[htbp]
\centering
\caption{Parameter values for the HIV host--virus dynamics model for patients 3 and 8.}
\label{tab:hiv_patient_parameters}
\small
\begin{minipage}{0.85\textwidth}
\centering
\begin{tabular*}{\linewidth}{@{\extracolsep{\fill}}lrr@{}}
\hline
Parameter & Patient 3 ($\ln$) & Patient 8 ($\ln$) \\
\hline
$\lambda$ & $-1.70$ & $-2.04$ \\
$d$       & $-3.77$ & $-3.73$ \\
$k$       & $-9.03$ & $-8.45$ \\
$\delta$  & $0.91$  & $-1.43$ \\
$\pi$     & $9.20$  & $7.40$  \\
$c$       & $1.09$  & $1.09$  \\
\hline
\end{tabular*}

\medskip
{\footnotesize
\raggedright
All values are reported on the natural logarithmic scale.
The viral clearance rate was fixed at
$c=3.0\,\mathrm{day}^{-1}$ for both patients \cite{stafford2000modeling},
corresponding to $\ln c \approx 1.0986$.
\par}
\end{minipage}
\end{table}

\begin{table}[htbp]
\centering
\caption{Comparison of additional diagnostic metrics between patients 3 and 8.}
\label{tab:hiv_patient_diagnostics}
\small
\renewcommand{\arraystretch}{1.15}
\begin{tabular}{lrr}
\hline
Metric & Patient 3 & Patient 8 \\
\hline
$\epsilon_U$ & $956.17$ & $0.39$ \\
Total loss  & $3.89\times10^{-1}$ & $1.16\times10^{-1}$ \\
Mean loss   & $4.32\times10^{-2}$ & $8.93\times10^{-3}$ \\
$\lambda_{\min}(K_\eta)$   & $4.81\times10^{-3}$ & $3.81\times10^{-4}$ \\
$\lambda_{\min}(F_\eta)$    & $5.11\times10^{-6}$ & $3.12\times10^{-4}$ \\
\hline
\end{tabular}
\end{table}

\begin{table}[htbp]
\centering
\caption{Theoretical bounds and computed values of the identifiability metrics for patient 3.}
\label{tab:patient3_theoretical_bounds}
\small
\renewcommand{\arraystretch}{1.15}
\begin{tabular}{lrrr}
\hline
Parameter & Lower bound & Computed value & Upper bound \\
\hline
\multicolumn{4}{c}{$L_p^{K}$} \\
\hline
$\ln\lambda$ & $4.84\times10^{-2}$ & $1.93\times10^{-1}$ & $4.65\times10^{1}$ \\
$\ln d$     & $2.71\times10^{-2}$ & $6.04\times10^{-1}$ & $2.61\times10^{1}$ \\
$\ln k$     & $1.47\times10^{-1}$ & $1.39$             & $1.41\times10^{2}$ \\
$\ln\delta$ & $2.57\times10^{-1}$ & $4.05\times10^{-1}$ & $2.47\times10^{2}$ \\
$\ln\pi$   & $3.49\times10^{-1}$ & $1.61$             & $3.36\times10^{2}$ \\
\hline
\multicolumn{4}{c}{$\kappa_p^{K}$} \\
\hline
$\ln\lambda$ & $4.63\times10^{-4}$ & $4.38\times10^{-1}$ & $4.45\times10^{-1}$ \\
$\ln d$     & $1.47\times10^{-3}$ & $9.34\times10^{-1}$ & $1.41426$ \\
$\ln k$     & $5.01\times10^{-5}$ & $4.39\times10^{-2}$ & $4.81\times10^{-2}$ \\
$\ln\delta$ & $1.63\times10^{-5}$ & $1.57\times10^{-2}$ & $1.57\times10^{-2}$ \\
$\ln\pi$   & $8.87\times10^{-6}$ & $8.34\times10^{-3}$ & $8.52\times10^{-3}$ \\
\hline
\multicolumn{4}{c}{$\lVert g_p(K)\rVert_2$} \\
\hline
$\ln\lambda$ & $2.15\times10^{-2}$ & $8.50\times10^{-2}$ & $3.34\times10^{-1}$ \\
$\ln d$     & $3.84\times10^{-2}$ & $5.64\times10^{-1}$ & $5.96\times10^{-1}$ \\
$\ln k$     & $7.09\times10^{-3}$ & $6.11\times10^{-2}$ & $1.10\times10^{-1}$ \\
$\ln\delta$ & $4.05\times10^{-3}$ & $6.37\times10^{-3}$ & $6.29\times10^{-2}$ \\
$\ln\pi$   & $2.98\times10^{-3}$ & $1.34\times10^{-2}$ & $4.63\times10^{-2}$ \\
\hline
\end{tabular}

\medskip
\end{table}

\begin{table}[htbp]
\centering
\caption{Theoretical bounds and computed values of the identifiability metrics for patient 8.}
\label{tab:patient8_theoretical_bounds}
\small
\renewcommand{\arraystretch}{1.15}
\begin{tabular}{lrrr}
\hline
Parameter & Lower bound & Computed value & Upper bound \\
\hline
\multicolumn{4}{c}{$L_p^{K}$} \\
\hline
$\ln\lambda$ & $1.66$             & $1.84$             & $3.85$ \\
$\ln d$     & $2.22\times10^{1}$ & $2.62\times10^{1}$ & $5.15\times10^{1}$ \\
$\ln k$     & $1.65\times10^{1}$ & $1.93\times10^{1}$ & $3.82\times10^{1}$ \\
$\ln\delta$ & $1.01\times10^{1}$ & $1.21\times10^{1}$ & $2.33\times10^{1}$ \\
$\ln\pi$   & $2.77\times10^{1}$ & $3.24\times10^{1}$ & $6.43\times10^{1}$ \\
\hline
\multicolumn{4}{c}{$\kappa_p^{K}$} \\
\hline
$\ln\lambda$ & $1.11\times10^{-1}$ & $2.36\times10^{-1}$ & $2.58\times10^{-1}$ \\
$\ln d$     & $6.23\times10^{-4}$ & $1.25\times10^{-3}$ & $1.44\times10^{-3}$ \\
$\ln k$     & $1.13\times10^{-3}$ & $2.30\times10^{-3}$ & $2.62\times10^{-3}$ \\
$\ln\delta$ & $3.03\times10^{-3}$ & $6.02\times10^{-3}$ & $7.02\times10^{-3}$ \\
$\ln\pi$   & $4.01\times10^{-4}$ & $8.13\times10^{-4}$ & $9.27\times10^{-4}$ \\
\hline
\multicolumn{4}{c}{$\lVert g_p(K)\rVert_2$} \\
\hline
$\ln\lambda$ & $4.30\times10^{-1}$ & $4.35\times10^{-1}$ & $4.68\times10^{-1}$ \\
$\ln d$     & $3.21\times10^{-2}$ & $3.30\times10^{-2}$ & $3.50\times10^{-2}$ \\
$\ln k$     & $4.33\times10^{-2}$ & $4.44\times10^{-2}$ & $4.72\times10^{-2}$ \\
$\ln\delta$ & $7.09\times10^{-2}$ & $7.28\times10^{-2}$ & $7.72\times10^{-2}$ \\
$\ln\pi$   & $2.57\times10^{-2}$ & $2.64\times10^{-2}$ & $2.81\times10^{-2}$ \\
\hline
\end{tabular}
\end{table}

\begin{table}[htbp]
\centering
\caption{Comparison of parameter credible intervals obtained using MCMC, $K_{\eta}$, and $F_{\eta}$ for patient 3.}
\label{tab:patient3_credible_intervals}
\small
\renewcommand{\arraystretch}{1.12}
\begin{tabular}{llrrr}
\hline
Parameter & Method & Lower bound & Upper bound & Interval width \\
\hline
$\lambda$ & MCMC       & $-0.032$ & $0.027$ & $0.06$ \\
          & $K_{\eta}$ & $-0.029$ & $0.029$ & $0.059$ \\
          & $F_{\eta}$ & $-0.90$  & $0.90$  & $1.81$ \\
\hline
$d$       & MCMC       & $-0.023$ & $0.017$ & $0.041$ \\
          & $K_{\eta}$ & $-0.020$ & $0.020$ & $0.040$ \\
          & $F_{\eta}$ & $-0.51$  & $0.51$  & $1.02$ \\
\hline
$k$       & MCMC       & $-0.086$ & $0.096$ & $0.182$ \\
          & $K_{\eta}$ & $-0.093$ & $0.093$ & $0.186$ \\
          & $F_{\eta}$ & $-2.76$   & $2.76$   & $5.53$ \\
\hline
$\delta$  & MCMC       & $-0.158$  & $0.151$  & $0.31031$ \\
          & $K_{\eta}$ & $-0.156$  & $0.156$  & $0.312$ \\
          & $F_{\eta}$ & $-4.83$   & $4.83$   & $9.66$ \\
\hline
$\pi$     & MCMC       & $-0.217$  & $0.204$  & $0.421$ \\
          & $K_{\eta}$ & $-0.214$  & $0.214$  & $0.429$ \\
          & $F_{\eta}$ & $-6.56$   & $6.56$   & $13.132$ \\
\hline
\end{tabular}

\medskip
\end{table}

\begin{table}[htbp]
\centering
\caption{Comparison of parameter credible intervals obtained using MCMC, $K_{\eta}$, and $F_{\eta}$ for patient 8.}
\label{tab:patient8_credible_intervals}
\small
\renewcommand{\arraystretch}{1.12}
\begin{tabular}{llrrr}
\hline
Parameter & Method & Lower bound & Upper bound & Interval width \\
\hline
$\lambda$ & MCMC       & $-0.00819$ & $0.0143$ & $0.0225$ \\
          & $K_{\eta}$ & $-0.0101$  & $0.0101$ & $0.0201$ \\
          & $F_{\eta}$ & $-0.0113$  & $0.0113$ & $0.0227$ \\
\hline
$d$       & MCMC       & $-0.175$ & $0.118$ & $0.294$ \\
          & $K_{\eta}$ & $-0.138$ & $0.138$ & $0.276$ \\
          & $F_{\eta}$ & $-0.152$ & $0.152$ & $0.304$ \\
\hline
$k$       & MCMC       & $-0.0941$ & $0.115$ & $0.209$ \\
          & $K_{\eta}$ & $-0.102$  & $0.102$ & $0.204$ \\
          & $F_{\eta}$ & $-0.112$  & $0.112$ & $0.225$ \\
\hline
$\delta$  & MCMC       & $-0.0689$ & $0.0601$ & $0.129$ \\
          & $K_{\eta}$ & $-0.0631$ & $0.0631$ & $0.126$ \\
          & $F_{\eta}$ & $-0.0691$ & $0.0691$ & $0.138$ \\
\hline
$\pi$     & MCMC       & $-0.196$ & $0.159$ & $0.356$ \\
          & $K_{\eta}$ & $-0.171$ & $0.171$ & $0.343$ \\
          & $F_{\eta}$ & $-0.190$ & $0.190$ & $0.380$ \\
\hline
\end{tabular}
\medskip
\end{table}

\begin{table}[htbp]
\centering
\footnotesize
\caption{Regional parameter values for $A\beta$ spatiotemporal dynamics in cortical regions 1--10. All values are reported on the natural logarithmic ($\ln$) scale.}
\label{tab:S3}
\renewcommand{\arraystretch}{1.2}
\setlength{\tabcolsep}{4pt}
\begin{tabular}{l*{10}{r}}
\hline
Region $i$ & 1 & 2 & 3 & 4 & 5 & 6 & 7 & 8 & 9 & 10 \\
\hline
$\lambda_i$
& $-1.64$ & $-4.18$ & $-4.02$ & $-2.48$ & $-3.07$ & $-13.65$ & $-13.65$ & $-2.90$ & $-4.52$ & $-4.34$ \\
$K_i$
& $0.65$ & $1.84$ & $2.10$ & $-13.65$ & $-13.65$ & $-13.65$ & $-13.65$ & $-13.65$ & $1.64$ & $1.95$ \\
\hline
\end{tabular}
\end{table}

\begin{table}[htbp]
\centering
\footnotesize
\caption{Regional parameter values for $A\beta$ spatiotemporal dynamics in cortical regions 11--20. All values are reported on the natural logarithmic ($\ln$) scale.}
\label{tab:S4}
\renewcommand{\arraystretch}{1.2}
\setlength{\tabcolsep}{4pt}
\begin{tabular}{l*{10}{r}}
\hline
Region $i$ & 11 & 12 & 13 & 14 & 15 & 16 & 17 & 18 & 19 & 20 \\
\hline
$\lambda_i$
& $-2.08$ & $-1.36$ & $-3.52$ & $0.11$ & $-5.87$ & $-1.61$ & $-3.41$ & $-2.13$ & $-4.69$ & $-0.53$ \\
$K_i$
& $-13.65$ & $-13.65$ & $2.13$ & $-13.65$ & $1.71$ & $-13.65$ & $1.84$ & $-13.65$ & $2.04$ & $-13.65$ \\
\hline
\end{tabular}
\end{table}

\begin{table}[htbp]
\centering
\footnotesize
\caption{Regional parameter values for $A\beta$ spatiotemporal dynamics in cortical regions 21--30. All values are reported on the natural logarithmic ($\ln$) scale.}
\label{tab:S5}
\renewcommand{\arraystretch}{1.2}
\setlength{\tabcolsep}{4pt}
\begin{tabular}{l*{10}{r}}
\hline
Region $i$ & 21 & 22 & 23 & 24 & 25 & 26 & 27 & 28 & 29 & 30 \\
\hline
$\lambda_i$
& $-4.00$ & $-3.32$ & $0.43$ & $-3.25$ & $-13.65$ & $-1.86$ & $-3.39$ & $-6.44$ & $-13.65$ & $-4.06$ \\
$K_i$
& $2.08$ & $1.80$ & $-13.65$ & $1.91$ & $-13.65$ & $1.68$ & $1.35$ & $1.98$ & $-13.65$ & $1.96$ \\
\hline
\end{tabular}
\end{table}

\begin{table}[htbp]
\centering
\footnotesize
\caption{Regional parameter values for $A\beta$ spatiotemporal dynamics in cortical regions 31--40. All values are reported on the natural logarithmic ($\ln$) scale.}
\label{tab:S6}
\renewcommand{\arraystretch}{1.2}
\setlength{\tabcolsep}{4pt}
\begin{tabular}{l*{10}{r}}
\hline
Region $i$ & 31 & 32 & 33 & 34 & 35 & 36 & 37 & 38 & 39 & 40 \\
\hline
$\lambda_i$
& $-2.42$ & $-3.94$ & $-3.59$ & $-4.22$ & $-2.50$ & $-4.40$ & $-2.07$ & $-4.49$ & $-1.99$ & $-1.57$ \\
$K_i$
& $-13.65$ & $1.89$ & $-13.65$ & $2.06$ & $1.24$ & $1.87$ & $-13.65$ & $1.75$ & $-13.65$ & $-13.65$ \\
\hline
\end{tabular}
\end{table}

\begin{table}[htbp]
\centering
\footnotesize
\caption{Regional parameter values for $A\beta$ spatiotemporal dynamics in cortical regions 41--50. All values are reported on the natural logarithmic ($\ln$) scale.}
\label{tab:S7}
\renewcommand{\arraystretch}{1.2}
\setlength{\tabcolsep}{4pt}
\begin{tabular}{l*{10}{r}}
\hline
Region $i$ & 41 & 42 & 43 & 44 & 45 & 46 & 47 & 48 & 49 & 50 \\
\hline
$\lambda_i$
& $-2.62$ & $-13.65$ & $-1.06$ & $-13.65$ & $-2.88$ & $-1.91$ & $-4.25$ & $0.43$ & $-3.27$ & $-1.25$ \\
$K_i$
& $-13.65$ & $-12.33$ & $-13.65$ & $-13.65$ & $1.56$ & $-13.65$ & $1.99$ & $-13.65$ & $-13.65$ & $-13.65$ \\
\hline
\end{tabular}
\end{table}

\begin{table}[htbp]
\centering
\footnotesize
\caption{Regional parameter values for $A\beta$ spatiotemporal dynamics in cortical regions 51--60. All values are reported on the natural logarithmic ($\ln$) scale.}
\label{tab:S8}
\renewcommand{\arraystretch}{1.2}
\setlength{\tabcolsep}{4pt}
\begin{tabular}{l*{10}{r}}
\hline
Region $i$ & 51 & 52 & 53 & 54 & 55 & 56 & 57 & 58 & 59 & 60 \\
\hline
$\lambda_i$
& $-3.15$ & $-2.92$ & $-13.65$ & $-1.25$ & $-4.37$ & $-2.90$ & $0.69$ & $-3.76$ & $-1.62$ & $-2.23$ \\
$K_i$
& $1.96$ & $-13.65$ & $-13.65$ & $-13.65$ & $2.03$ & $1.99$ & $-13.65$ & $1.86$ & $-13.65$ & $1.99$ \\
\hline
\end{tabular}
\end{table}

\begin{table}[htbp]
\centering
\footnotesize
\caption{Regional parameter values for $A\beta$ spatiotemporal dynamics in cortical regions 61--68. All values are reported on the natural logarithmic ($\ln$) scale.}
\label{tab:S9}
\renewcommand{\arraystretch}{1.2}
\setlength{\tabcolsep}{4pt}
\begin{tabular}{l*{8}{r}}
\hline
Region $i$ & 61 & 62 & 63 & 64 & 65 & 66 & 67 & 68 \\
\hline
$\lambda_i$
& $-5.25$ & $-0.79$ & $-0.19$ & $0.05$ & $-0.37$ & $-0.97$ & $-2.05$ & $-5.18$ \\
$K_i$
& $1.30$ & $-13.65$ & $-13.65$ & $-13.65$ & $-13.65$ & $-13.65$ & $-13.65$ & $1.70$ \\
\hline
\end{tabular}
\end{table}

\end{document}